\documentclass{article}

    \PassOptionsToPackage{numbers, compress}{natbib}
    \usepackage[final]{neurips_2020}

\usepackage[utf8]{inputenc} % allow utf-8 input
\usepackage[T1]{fontenc}    % use 8-bit T1 fonts
\usepackage{hyperref}       % hyperlinks
\usepackage{url}            % simple URL typesetting
\usepackage{booktabs}       % professional-quality tables
\usepackage{amsfonts}       % blackboard math symbols
\usepackage{nicefrac}       % compact symbols for 1/2, etc.
\usepackage{microtype}      % microtypography

\usepackage{tablefootnote}

\usepackage{amsmath}
\usepackage{amssymb}
\usepackage{amsthm}
\usepackage{bbm}
\usepackage{bm}
\usepackage{graphicx}
\usepackage{comment}
\usepackage{enumerate}
\usepackage[capitalize]{cleveref}
\usepackage{mathtools}
\usepackage{color}

\newtheorem{theorem}{Theorem}[section]
\newtheorem{proposition}[theorem]{Proposition}
\newtheorem{lemma}[theorem]{Lemma}
\newtheorem{definition}[theorem]{Definition}
\newtheorem{claim}[theorem]{Claim}
\newtheorem{corollary}[theorem]{Corollary}

\newtheorem{example}[theorem]{Example}

\newcommand{\old}[1]{{\color{red} [old: #1]}}

\newcommand{\erm}{\mathrm{ERM}}
\newcommand{\E}{\mathbb{E}}
\newcommand{\dd}{\mathrm{d}}
\newcommand{\TP}{\mathrm{TP}}
\newcommand{\usecond}{u}

\newcommand{\interim}{\mathrm{worst}}
\newcommand{\worst}{\mathrm{worst}}
\newcommand{\event}{\mathrm{E}}
\newcommand{\conc}{\mathrm{Conc}}
\newcommand{\Bad}{\mathrm{Bad}}
\newcommand{\first}{\eta}
\newcommand{\second}{\theta}
\newcommand{\defeq}{\stackrel{\mathrm{def}}{=}}

\newcommand{\TPm}{\bm{m}}
\newcommand{\TPT}{\bm{T}}
\newcommand{\TPK}{\bm{K}}
\newcommand{\msum}{\overline{m}}
\newcommand{\cupperboundmhr}{\frac{1}{4e}}

\DeclareMathOperator*{\argmax}{arg\,max}
\DeclareMathOperator*{\argmin}{arg\,min}

\newcommand{\mye}{\mathrm{Mye}}

\newcommand{\hist}{h}
\newcommand{\strat}{\sigma}
\newcommand{\pbeu}{U}

\title{A Game-Theoretic Analysis of the Empirical Revenue Maximization Algorithm with Endogenous Sampling}

\author{%
  Xiaotie Deng \\
  Peking University\\
  \texttt{xiaotie@pku.edu.cn} \\
   \And
   Ron Lavi \\
   Technion - Israel Institute of Technology \\
   \texttt{ronlavi@ie.technion.ac.il} \\
   \AND
   Tao Lin \\
   Peking University\\
  \texttt{lin\_tao@pku.edu.cn} \\
   \And
   Qi Qi \\
   Hong Kong University of Science and Technology \\
   \texttt{kaylaqi@ust.hk} \\
   \And
   Wenwei Wang \\
   Alibaba Group \\
   \texttt{wwangaw@connect.ust.hk} \\
   \And
   Xiang Yan \\
   Shanghai Jiao Tong University \\
   \texttt{xyansjtu@163.com} \\
}

\begin{document}

\maketitle

\begin{abstract}
% New structure

% for MHR $F$, 
%     $\Delta^{\interim}_{N, m}=O\left(m\frac{\log^3 N} {\sqrt{N}}\right)$
%     , if $m=o(\sqrt{N})$ and $\frac{m}{N}\le c\le\cupperboundmhr$.\footnote{We use $a(n)=o(b(n))$ to denote $\lim_{n\to+\infty}\frac{a(n)}{b(n)}=0$.}
%     \item for bounded $F$, $\Delta^{\interim}_{N, m} = O\left(D^{8/3}m^{2/3}\frac{\log^2 N} {N^{1/3}}\right)$

The Empirical Revenue Maximization (ERM) is one of the most important price learning algorithms in auction design: as the literature shows it can learn approximately optimal reserve prices for revenue-maximizing auctioneers in both repeated auctions and uniform-price auctions.  However, in these applications the agents who provide inputs to ERM have incentives to manipulate the inputs to lower the outputted price. We generalize the definition of an incentive-awareness measure proposed by Lavi et al (2019), to quantify the reduction of ERM's outputted price due to a change of $m\ge 1$ out of $N$ input samples, and provide specific convergence rates of this measure to zero as $N$ goes to infinity for different types of input distributions.  By adopting this measure, we construct an efficient, approximately incentive-compatible, and revenue-optimal learning algorithm using ERM in repeated auctions against non-myopic bidders, and show approximate group incentive-compatibility in uniform-price auctions.

\end{abstract}

\section{Introduction}
%\input{delta_centered_intro.tex}
% \new{we need to shorten the intro, and justify iid setting and the study of ERM}

In auction theory, it is well-known \cite{myerson1981optimal} that, when all buyers have values that are independently and identically drawn from a regular distribution $F$, the revenue-maximizing auction is simply the second price auction with anonymous reserve price $p^*=\argmax\{p(1-F(p)\}$: if the highest bid is at least $p^*$, then the highest bidder wins the item and pays the maximum between the second highest bid and $p^*$. 
The computation of $p^*$ requires the exact knowledge of the underlying value distribution, which is unrealistic because the value distribution is often unavailable in practice.  Many works (e.g., \cite{cole2014the, dhangwatnotai2015revenue, huang2018making}) on sample complexity in auctions have studied how to obtain a near-optimal reserve price based on samples from the distribution $F$ instead of knowing the exact $F$.  One of the most important (and most fundamental) price learning algorithms in those works is the \emph{Empirical Revenue Maximization} (ERM) algorithm, which simply outputs the reserve price that is optimal on the uniform distribution over samples (plus some regularization to prevent overfitting). 

\begin{definition}[$c$-Guarded Empirical Revenue Maximization, $\erm^c$]
Draw $N$ samples from a distribution $F$ and sort them non-increasingly, denoted by %$b_1\ge b_2\ge \cdots \ge b_N$.
$v_1\ge v_2\ge \cdots \ge v_N$. 
Given some regularization parameter $0\leq c<1$, choose:
$$ i^* = \argmax_{cN< i \le N}\{i\cdot v_i\},\quad\quad \text{define} ~~ \erm^c(v_1, \ldots, v_N)= v_{i^*}.$$
Assume that the smaller sample (with the larger index) is chosen in case of ties.
\end{definition}

$\erm^c$ was first proposed by~\citet{dhangwatnotai2015revenue} and then extensively studied by~\citet{huang2018making}.  They show that the reserve price outputted by $\erm^c$ is asymptotically optimal on the underlying distribution $F$ as the number of samples $N$ increases if $F$ is bounded or has monotone hazard rate, with an appropriate choice of $c$.  Other papers \cite{babaioff2018two, kanoria2019incentive} have continued to study $\erm^c$.

%The revenue-optimality of ERM suggests that, in real-world scenarios such as repeated second price auctions, the auctioneer can obtain no-regret revenue by using ERM to update the reserve price at each round of auction based on all previous bids, if buyers are myopic.

However, when ERM is put into practice, it is unclear how the samples can be obtained since many times there is no impartial sampling source.  A natural solution is \emph{endogenous sampling}.
%, by which we mean that the samples are collected from those agents who are part of the auction with reserve price set by ERM.
For example, in repeated second price auctions, the auctioneer can use the bids in previous auctions as samples and run ERM to set a reserve price at each round. But this solution has a challenge of strategic issue: since bidders can affect the determination of future reserve prices, they might have an incentive to underbid in order to increase utility in future auctions.
%But this solution ignores strategic issues since bidders can affect the determination of future reserve prices and therefore might have an incentive to bid strategically in order to increase utility in future auctions.  

Another example of endogenous sampling is the uniform-price auction.  In a uniform-price auction the auctioneer sells $N$ copies of a good at some price $p$ to $N$ bidders with i.i.d.~values $\bm{v}$ from $F$ who submit bids $\bm{b}$.  Bidders who bid at least $p$ obtain one copy and pay $p$.  The auctioneer can set the price to be $p=\erm^0(\bm{v})$ to maximize revenue if bids are equal to values.  But \citet{goldberg2006competitive} show that this auction is not incentive-compatible as bidders can lower the price by strategic bidding. 
Therefore, the main question we consider in this paper is:
\emph{To what extent the presence of strategic agents undermines ERM with endogenous sampling?} 

To formally answer the question, we adopt a notion called ``incentive-awareness measure'' originally proposed by \citet{lavi2019redesigning} under bitcoin's fee market context, which measures the reduction of a price learning function $P$ due to a change of at most $m$ samples out of the $N$ input samples.

\begin{definition}[Incentive-awareness measures]\label{def:discount}
Let $P:\mathbb{R}_+^N\to\mathbb{R}_+$ be a function (e.g., $\erm^c$) that maps $N$ samples to a reserve price. Draw $N$ i.i.d.~values $v_1, \ldots, v_{N}$ from a distribution $F$. Let $I\subseteq\{1, \ldots, N\}$ be an index set of size $|I|=m$, and $v_I=\{v_i\mid i\in I\}$, $v_{-I}=\{v_j\mid j\notin I\}$. A bidder can change $v_I$ to any $m$ non-negative bids $b_I$, hence change the price from $P(v_I, v_{-I})$ to $P(b_I, v_{-I})$. Define the \emph{incentive-awareness measure}:
$$\delta_I(v_I, v_{-I}) = 1 - \frac{\inf_{b_I\in\mathbb{R}_+^m} P(b_I, v_{-I})}{P(v_I, v_{-I})}, $$
and \emph{worst-case incentive-awareness measures}: 
\[\delta_{m}^{\interim}(v_{-I}) = \sup_{v_I\in\mathbb{R}_+^m}[ \delta_I(v_I, v_{-I})],\quad\quad\quad\quad\Delta_{N, m}^{\interim} = \E_{v_{-I}\sim F}[ \delta_m^{\interim}(v_{-I})].\]
\end{definition}
A smaller incentive-awareness measure means that the reserve price is decreased by a less amount when a bidder bids strategically.  Since the reduction of reserve price usually increases bidders' utility, a smaller incentive-awareness measure implies that a bidder cannot benefit a lot from strategic bidding, hence the name ``incentive-awareness measure''.\footnote{\citet{lavi2019redesigning} use the name ``discount ratio'' which we feel can be confused with the standard meaning of a discount ratio in repeated games. % And their definition only consider the case where $m=1$.  The generalization to arbitrary $m\ge 1$ is crucial in our two applications.
}   

\citet{lavi2019redesigning} defined incentive-awareness measures only for $m=1$ and showed that for any distribution $F$ with a finite support size, $\Delta^{\interim}_{N, 1}\to 0$ as $N\to\infty$.  Later, \citet{yao2018incentive} showed that $\Delta^{\interim}_{N, 1}\to 0$ for any continuous distribution with support included in $[1, D]$.  They did not provide specific convergent rates of $\Delta^{\interim}_{N, 1}$.  We generalize their definition to allow $m\ge 1$, which is crucial in our two applications to be discussed.  Our main theoretical contribution is to provide upper bounds on $\Delta^{\interim}_{N, m}$ for two types of value distributions $F$: the class of Monotone Hazard Rate (MHR) distributions where $\frac{f(v)}{1-F(v)}$ is non-decreasing over the support of the distribution (we use $F$ to denote the CDF and $f$ for PDF)
and the class of bounded distributions which consists of all (continuous and discontinuous) distributions with support included in $[1, D]$.  MHR distribution can be unbounded so we are the first to consider incentive-awareness measures for unbounded distributions. 

\begin{theorem}[Main]\label{thm:discount_combined}
Let $P=\erm^c$. The worst-case incentive-awareness measure is bounded by
\begin{itemize} 
    \item for MHR $F$, 
    $\Delta^{\interim}_{N, m}=O\left(m\frac{\log^3 N} {\sqrt{N}}\right)$
    , if $m=o(\sqrt{N})$ and $\frac{m}{N}\le c\le\cupperboundmhr$.\footnote{We use $a(n)=o(b(n))$ to denote $\lim_{n\to+\infty}\frac{a(n)}{b(n)}=0$.}
    \item for bounded $F$, $\Delta^{\interim}_{N, m} = O\left(D^{8/3}m^{2/3}\frac{\log^2 N} {N^{1/3}}\right)$, if $m=o(\sqrt{N})$ and  $\frac{m}{N}\le c\le \frac{1}{2D}$. 
\end{itemize}

The constants in the two big $O$'s are independent of $F$ and $c$. 
\end{theorem}

This theorem implies that as long as the fraction of samples controlled by a bidder is relatively small, the strategic behavior of each bidder has little impact on ERM provided that other bidders are truthful. 
Meanwhile, if more than one bidder bid non-truthfully, no bidder can benefit a lot from lying as long as the total number of bids of all non-truthful bidders does not exceed $m$. 
We will discuss intuitions and difficulties of the proof later and give an overview in Section~\ref{sec:discussion}.

%As two applications of our main theorem, we show that (1) a two-phase algorithm that uses $\erm^c$ to update the reserve price in repeated second price auctions can achieve approximate incentive-compatibility and approximate revenue-optimality simultaneously when bidders are non-myopic (Section~\ref{sec:two_phase}); (2) the uniform-auction that sets the price by $\erm^c$ is approximately group incentive-compatible (Section~\ref{sec:uniform-price}).  Finally, we give an overview of the proof of Theorem~\ref{thm:discount_combined} and discuss lower bounds and distributions beyond bounded and MHR distributions in Section~\ref{sec:discussion}. 

\paragraph{Repeated auctions against non-myopic bidders.}
Besides theoretical analysis, we apply the incentive-awareness measure to real-world scenarios to demonstrate the effect of strategic bidding on ERM.  The main application we consider is repeated auctions where bidders participate in multiple auctions and have incentives to bid strategically to affect the auctions the seller will use in the future (Section~\ref{sec:two_phase}).  We consider a two-phase learning algorithm: the seller first runs second price auctions with no reserve for some time to collect samples, and then use these samples to set reserve prices by ERM in the second phase.  The upper bound on the incentive-awareness measure of ERM implies that this algorithm is approximately incentive-compatible. 

\citet{kanoria2019incentive}, \citet{liu2018learning}, and \citet{NIPS2019_abernethy_learning} consider repeated auctions scenarios similar to ours.  \citet{kanoria2019incentive} set \emph{personalized} reserve prices by ERM in repeated second-price auctions, so at least two bidders are needed in each auction and they will face different reserve prices.  We use \emph{anonymous} reserve price so we allow only one bidder to participate in the auctions and when there are more than one bidder they face the same price, thus preventing discrimination.  \citet{liu2018learning} and \citet{NIPS2019_abernethy_learning} design approximately incentive-compatible algorithms using differential privacy techniques rather than pure ERM.  Comparing with them, our two-phase ERM algorithm is more practical as it is much simpler, and their algorithms rely on the boundedness of value distributions while we allow unbounded distributions.  Moreover, their results require a large number of auctions while ours need a large number of samples in the first phase which can be obtained by either few bids in many auctions, many bids in few auctions, or combined.  

{\bf Uniform-price auctions and incentive-compatibility in the large.}  Another scenario to which we apply the incentive-awareness measure of ERM is uniform-price auctions (Section~\ref{sec:uniform-price}).  \citet{azevedo2018strategy} show that, uniform-price auctions are \emph{incentive-compatible in the large} in the sense that truthful bidding is an approximate equilibrium when there are many bidders in the auction.  In fact, incentive-compatibility in the large is the intuition of Theorem~\ref{thm:discount_combined}: when $N$ is large, no bidders can influence the learned price by much.
The proof in \cite{azevedo2018strategy} directly makes use of this intuition, showing that the bid of one bidder can affect the empirical distribution consisting of the $N$ bids only by a little.  However, their argument, which crucially relies on the assumption that bidders' value distribution has a \emph{finite} support and bids must be chosen \emph{from this finite support} as well, fails when the value distribution is \emph{continuous} or \emph{bids can be any real numbers}, as what we allow. 
We instead, appeal to some specific properties of ERM to show that it is incentive-compatible in the large. 

{\bf Additional related works.}
\label{sec:related_work}
%%%%%%%%%%%%%%%%%%%%%%%%%%%%%%%%%%%%%%%
Previous works on ERM mainly focus on its sample complexity, started by \citet{cole2014the}.  While ERM is suitable for the case of i.i.d.~values (e.g., \cite{huang2018making}), the literature on sample complexity has expanded to more general cases of non-i.i.d.~values and multi-dimensional values, e.g.~\cite{morgenstern2015pseudo, devanur2016sample, gonczarowski2018sample, guo2019settling}, or considering non-truthful auctions, e.g.~\cite{hartline2018sample}.  \citet{babaioff2018two} study the performance of ERM with just two samples. 
%, and \citet{kanoria2019incentive} determining personalized reserve prices in repeated second-price auctions by ERM.  %which is generalized to contextual setting by \citet{NIPS2019_Golrezaei_dynamic}.
While this literature assumes that samples are exogenous, our main contribution is to consider endogenous samples that are collected from bidders who are affected by the outcome of the learning algorithm. 

Some works study repeated auctions but with myopic bidders \cite{blum2005near, medina2014learning, cesa2015regret, bubeck2017online}.  Existing works about non-myopic bidders focus on designing various learning algorithms to maximize revenue assuming bidders playing best responds \cite{amin2013learning, amin2014repeated, NIPS2019_deng_robust, NIPS2019_Golrezaei_dynamic} or using no-regret learning algorithm  \cite{braverman_selling_2018}.  We complement that line of works by showing that ERM, the most fundamental algorithm we believe, also has good performance in repeated auctions against strategic bidders.  

Other works about incentive-aware learning (e.g., \cite{devanur_perfect_2015, immorlica2017repeated, balcan2008reducing, epasto2018incentive}) consider settings different from ours. For example, \cite{devanur_perfect_2015} and \cite{immorlica2017repeated} study repeated auctions where buyers' values are drawn from some distribution at first and then fixed throughout. The seller knows the distribution and tries to learn the exact values, which is different from our assumption that the distribution is unknown to the seller.

\section{Main Application: A Two-Phase Model}
\label{sec:two_phase}
Here we consider a \emph{two-phase model} as a real-world scenario where strategic bidding affects ERM: the seller first runs second price auctions with no reserve for some time to collect samples, and then use these samples to set reserve prices by ERM in the second phase.  This model can be regarded as an ``exploration and exploitation'' learning algorithm in repeated auctions, and we will show that this algorithm can be approximately incentive-compatible and revenue-optimal. 
\subsection{The Model}
A two-phase model is denoted by $\TP(\mathcal{M}, P; F, \TPT, \TPm, \TPK, S)$, where $\mathcal{M}$ is a truthful, prior-independent mechanism, $P$ is a price learning function, and $\TPT=(T_1,T_2)$ are the numbers of auctions in the two phases.
We do not assume that every bidder participates in all auctions.
Instead, we assume that each bidder participates in no more than $m_1, m_2$ auctions in the two phases, respectively; $\TPm=(m_1,m_2)$. 
%Instead, we use $\TPm=(m_1,m_2)$ to denote upper bounds on the number of auctions each bidder participates in,
We use $\TPK = (K_1, K_2)$ to denote the numbers of bidders in the auctions of the two phases.\footnote{
It is well-known that when there are $n$ bidders with i.i.d.~regular value distributions in one auction, a second price auction as a prior-independent mechanism is ($1-1/n$)-revenue optimal. 
But in our two-phase model, the numbers of bidders in each auction, $K_1$ and $K_2$, can be small, e.g., 1 or 2.
With few bidders, prior-independent mechanisms do not have good revenue.
}
$S=S_1\times\cdots\times S_n$ is the strategy space, where $s_i\in S_i: \mathbb{R}_+^{m_{i, 1}+m_{i, 2}}\to \mathbb{R}_+^{m_{i, 1}}$ is a strategy of bidder $i=1, \dots, n$. The procedure is:
\begin{itemize}
\item At the beginning, each bidder realizes $\bm{v}_i=(v_{i,1}, \ldots, v_{i, m_{i,1}+m_{i,2}})$ i.i.d.~drawn from $F$. Let $\bm{v}_{-i}$ denote the values of bidders other than $i$. Bidder $i$ knows $\bm{v}_i$ but does not know $\bm{v}_{-i}$.

\item In the exploration phase,
$T_1$ auctions are run using $\mathcal{M}$ and bidders bid according to some strategy $s \in S$. Each auction has $K_1$ bidders and each bidder $i$ participants in $m_{i, 1} \le m_1$ auctions. The auctioneer observes a random vector of bids $\bm{b}=(b_1, \ldots, b_{T_1K_1})$ with the following distribution: let $I$ be an index set corresponding to bidder $i$, with size $|I|=m_{i, 1}$;
then $\bm{b}=(b_I, b_{-I})$, where $b_{I} \sim s_i(\bm{v}_i)$, and $b_{-I}\sim s_{-i}(\bm{v}_{-i})$. 

\item In the exploitation phase,
$T_2$ second price auctions ($K_2\ge2$) or posted price auctions ($K_2=1$) are run, with reserve price $p=P(\bm{b})$. Each auction has $K_2$ bidders and each bidder $i$ participants in $m_{i, 2}\le m_2$ auctions. The auctions in this phase are truthful because $p$ has been fixed.
\end{itemize}

\noindent
{\bf Utilities. }
Denote the utility of bidder $i$ as: 
\begin{equation}\label{eqn:def_TP_utility}
    U^{\TP}_i(\bm{v}_i, b_I, b_{-I}) = U^{\mathcal{M}}_i(\bm{v}_i, b_I, b_{-I})
    + \sum_{t=m_{i, 1}+1}^{m_{i, 1} + m_{i, 2}}
    \usecond^{K_2}(v_{i,t}, P(b_I, b_{-I})),
\end{equation}
where $U^{\mathcal{M}}_i(v_I, b_I, b_{-I})$ is the utility of bidder $i$ in the first phase, and $\usecond^{K_2}(v, p)$ is the interim utility of a bidder with value $v$ in a second price auction 
with reserve price $p$ among $K_2 \geq 1$ bidders: 
\begin{equation}\label{eqn:def_u_second}
	\usecond^{K_2}(v, p)=\E_{X_2, \ldots, X_{K_2}\sim F}\bigg[\big(v - \max\{p, X_2, \ldots, X_{K_2}\}\big)\cdot\mathbb{I}\big[v>\max\{p, X_2, \ldots, X_{K_2}\}\big]\bigg].
\end{equation}
The \emph{interim utility} of bidder $i$ in the two-phase model is
$\E_{\bm{v}_{-i}\sim F}\left[ U^{\TP}_i(\bm{v}_i, b_I, b_{-I})\right]$.

\noindent
{\bf Approximate Bayesian incentive-compatibility.}
We use the additive version of the solution concept of an $\epsilon$-Bayesian-Nash equilibrium ($\epsilon$-BNE), i.e., in such a solution concept, no player can improve her utility by more than $\epsilon$ by deviating from the equilibrium strategy. We say a mechanism is $\epsilon$-approximately Bayesian incentive-compatible ($\epsilon$-BIC) if truthful bidding is an $\epsilon$-BNE, i.e., if for any $\bm{v}_i\in\mathbb{R}_+^{m_{i, 1}+m_{i, 2}}$, any $b_I\in \mathbb{R}_+^{m_{i, 1}}$, 
\[ \E_{\bm{v}_{-i}\sim F}\left[ U^{\TP}_i(\bm{v}_i, b_I, v_{-I})  -  U^{\TP}_i(\bm{v}_i, v_I, v_{-I}) \right] \le \epsilon, \]
If a mechanism is $\epsilon$-BIC and $\lim_{n \to \infty} \epsilon = 0$, then
each bidder knows that if all other bidders are bidding truthfully then the gain from any deviation from truthful bidding is negligible for her.  To realize that strategic bidding cannot benefit them much, bidders do not need to know the underlying distribution, but only the fact that the mechanism is $\epsilon$-BIC. We are therefore going to assume in this paper that, in such a case, all bidders will bid truthfully.
\noindent

% {\bf Approximate Bayesian incentive-compatibility.} \new{change to "Approximate incentive compatible"?}
% We use the additive version of the solution concept of an $\epsilon$ Bayesian-Nash Equilibrium ($\epsilon$-BNE), i.e., in such a solution concept, no player can improve her utility by more than $\epsilon$ by deviating from the equilibrium strategy. We say that truthful bidding is an $\epsilon$-BNE, if for any $\bm{v}_i$ from $F$, any $b_I\in \mathbb{R}_+^{m_{i, 1}}$, 
% \[ \E_{\bm{v}_{-i}\sim F}\left[ U^{\TP}_i(\bm{v}_i, b_I, v_{-I})  -  U^{\TP}_i(\bm{v}_i, v_I, v_{-I}) \right] \le \epsilon, \]
% and a mechanism is $\epsilon$-approximately truthful (or incentive compatible) if truthful bidding is an $\epsilon$-BNE.
% If a mechanism is $\epsilon$-approximately truthful and $\lim_{n \to \infty} \epsilon = 0$,
% bidder $i$ knows that if all other bidders are bidding truthfully then the gain from any deviation from truthful bidding is negligible for her. To realize that strategic bidding cannot benefit them much, bidders do not need to know the underlying distribution, but only the fact that the mechanism is $\epsilon$-approximately truthful. We are therefore going to assume in this paper that, in such a case, all bidders will bid truthfully.

\noindent
{\bf Approximate revenue optimality. }
We say that a mechanism is $(1-\epsilon)$ revenue optimal, for some $0<\epsilon<1$, if its expected revenue is at least $(1-\epsilon)$ times the expected revenue of Myerson auction.\footnote{
%Notice that if the BIC condition in Myerson auction is relaxed to $\epsilon$-BIC, the optimal expected revenue will increase at most $O(\epsilon)$, see e.g. \citet{brustle2019multi-item}.
One may take the optimal $\epsilon$-BIC auction rather than the exact BIC Myerson auction as the revenue benchmark.  However, as shown by e.g., Lemma 1 of \citet{brustle_multi-item_2020}, the revenue of the optimal $\epsilon$-BIC auction is at most $O(\epsilon)$ greater than that of Myerson auction; so all our revenue approximation results hold for this stronger benchmark except for an additive $O(\epsilon)$ term. 
}
\citet{huang2018making} show that a one-bidder auction with posted price set by $\erm^c$ (for an appropriate $c$) and with $N$ samples from the value distribution is $(1-\epsilon)$ revenue optimal with $\epsilon=O((N^{-1}\log N)^{2/3})$ for MHR distributions and $\epsilon = O(\sqrt{DN^{-1}\log N})$ for bounded distributions.
%\citet{huang2018making} have shown that a one-bidder auction with posted price set by $\erm^c$ (for an appropriate $c$) and with $N$ samples from the value distribution is $(1-\epsilon)$ revenue optimal with $\epsilon=O((\frac{\log N}{N})^{2/3})$ for MHR distributions and $\epsilon = O(\sqrt{\frac{D\cdot\log N}{N}})$ for bounded distributions. 

% \noindent
{\bf The i.i.d.~assumption. }  Our assumption of i.i.d.~values is reasonable because in our scenario there is a large population of bidders, and we can regard this population as a distribution and each bidder as a sample from it.  So from each bidder's perspective, the values of other bidders are i.i.d.~samples from this distribution.  Then the $\epsilon$-BIC notion implies that when others bid truthfully, it is approximately optimal for bidder $i$ to bid truthfully no matter what her value is.
%Another justification is:
% Internet companies selling goods to a large amount of buyers often partition buyers into groups by their statistical features or individual properties, so the buyers in each group are believed to have i.i.d.~distribution. 
% Then the Internet company can run the optimal auction for i.i.d.~distribution in each group. 
%In fact, \citet{NIPS2019_abernethy_learning} also assume that buyers are grouped but the seller runs auctions on all groups simultaneously.   

\subsection{Incentive-Compatibility and Revenue Optimality}

Now we show that, as the incentive-awareness measure of $P$ becomes lower, the price learning function becomes more incentive-aware in the sense that bidders gain less from non-truthful bidding:
 
\begin{theorem}\label{thm:connection_TPM_PRM_combined}
In $\TP(\mathcal{M}, P; F, \TPT, \TPm, \TPK, S)$, truthful bidding is an $\epsilon$-BNE, where,

\begin{itemize}

\item for any $P$ and any bounded $F$, 
$\epsilon=m_2D \Delta^{\interim}_{T_1K_1, m_1}$, and

\item for any MHR $F$, if we fix $P=\erm^c$ with $\frac{m_1}{T_1K_1}\le c\le \cupperboundmhr$ and $m_1=o(\sqrt{T_1K_1})$, then
$\epsilon=O\left(m_2v^* \Delta^{\interim}_{T_1K_1, m_1}\right)+O\left(\frac{m_2v^*}{\sqrt{T_1K_1}}\right)$,
where $v^*=\argmax_v\{v[1-F(v)]\}$.
\end{itemize}

The constants in big $O$'s are independent of $F$ and $c$. 
\end{theorem}

% Combined with the upper bound of the incentive-awareness measure, we have explicit bounds on truthfulness and revenue optimality in our two-phase model:

% \begin{corollary}\label{cor:TPM_utility_combined}

% In $\TP(\mathcal{M}, \erm^c; F, \TPT, \TPm, \TPK, S)$, truthful bidding is an $\epsilon$-BNE, where,

% \begin{itemize}

% \item for any MHR $F$, let $v^*=\argmax_v\{v[1-F(v)]\}$,
% %$\epsilon=O\left(\frac{v^*m_2}{\max\{1, e(K_2-1)\}} \sqrt[4]{\frac{\log (nm_1)}{nm_1}}\right)$
% \[\epsilon=O\left(v^* m_2 m_1\frac{\log^3 (T_1K_1) }{\sqrt{T_1K_1}}  \right)\]

% if $m_1=o(\sqrt{T_1K_1})$ and $\frac{m_1}{T_1K_1}\le c\le \cupperboundmhr$.

% \item for any bounded $F$, 
% \[\epsilon=O\left(D^{11/3}m_2 m_1^{2/3} \frac{\log^2 (T_1K_1) }{(T_1K_1)^{1/3}}  \right)\] if $m_1=o(\sqrt{T_1K_1})$ and $\frac{m_1}{T_1K_1}\le c\le \frac{1}{2D}$.
% \end{itemize}

% The constants in the two big $O$'s are independent of $F$ and $c$. 
% \end{corollary}

% \noindent
% This is a direct corollary of Theorems~\ref{thm:connection_TPM_PRM_combined}
% and~\ref{thm:discount_combined} by plugging in $N=T_1K_1$ and $m=m_1$. For example, for both of the cases in the corollary, keeping all the parameters except $T_1$ constant (in particular $m_1$ and $m_2$ are constants) implies that $\epsilon$ goes to zero as $T_1$ goes to infinity at a rate which is not slower than $O((T_1)^{-1/3}\log^3 T_1)$.

Combined with Theorem~\ref{thm:discount_combined} which upper bounds the incentive-awareness measure, we can obtain explicit bounds on truthfulness of the two-phase model by plugging in $N=T_1K_1$ and $m=m_1$. Precisely, for any bounded $F$, $\epsilon=O\left(D^{11/3}m_2 m_1^{2/3} \frac{\log^2 (T_1K_1) }{(T_1K_1)^{1/3}}\right)$ if $m_1=o(\sqrt{T_1K_1})$ and $\frac{m_1}{T_1K_1}\le c\le \frac{1}{2D}$. For any MHR $F$, $\epsilon=O\left(v^* m_2 m_1\frac{\log^3 (T_1K_1) }{\sqrt{T_1K_1}}  \right)$ if $m_1=o(\sqrt{T_1K_1})$ and $\frac{m_1}{T_1K_1}\le c\le \cupperboundmhr$. Thus, for both cases, keeping all the parameters except $T_1$ constant (in particular $m_1$ and $m_2$ are constants) implies that $\epsilon\to 0$ at a rate which is not slower than $O((T_1)^{-1/3}\log^3 T_1)$ as $T_1\to+\infty$.

To simultaneously obtain both approximate BIC and approximate revenue optimality, a certain balance between the number of auctions in the two phases must be maintained. Few auctions in the first phase and many auctions in the second phase hurt truthfulness as the loss from non-truthful bidding (i.e., losing in the first phase) is small compared to the gain from manipulating the reserve price in the second phase. Many auctions in the first phase are problematic as we do not have any good revenue guarantees in the first phase (since we allow any truthful $\mathcal{M}$). Thus, a certain balance must be maintained, as expressed formally in the following theorem:

\begin{theorem}\label{thm:TPM_utility_revenue_combined}
Assume that $K_2\ge K_1\ge 1$ and let $\msum= m_1+m_2$. In $\TP(\mathcal{M}, \erm^c; F, \TPT, \TPm, \TPK, S)$, to simultaneously obtain $\epsilon_1$-BIC and $(1-\epsilon_2)$ revenue optimality (assuming truthful bidding), it suffices to set the parameters as follows: 
\begin{itemize}
\item If $F$ is an MHR distribution, $\frac{\msum}{T_1K_1}\le c\le \cupperboundmhr$, $\msum = o(\sqrt{T_1K_1})$,
%and for any $T_1$ \footnote{We use $a(n)\ge\omega(b(n))$ to denote $\lim_{n\to+\infty}\frac{a(n)}{b(n)}=+\infty$.}.
then \\$\epsilon_1=O\left(v^* \msum^{2} \frac{\log^3 (T_1K_1) }{\sqrt{T_1K_1}}  \right)$, 
and  $\epsilon_2=O\left( \frac{T_1}{T} + \left[\frac{\log (T_1 K_1)}{T_1 K_1}\right]^{\frac{2}{3}} \right)$. 

\item If $F$ is bounded and regular, $\frac{\msum}{T_1K_1}\le c\le \frac{1}{2D}$, $\msum = o(\sqrt{T_1K})$, then \\$\epsilon_1=O\left(D^{11/3} \msum^{5/3} \frac{\log^2 (T_1K_1) }{(T_1K_1)^{1/3}}\right)$, and
$\epsilon_2 = O\left(\frac{T_1}{T}  + \sqrt{\frac{D\cdot\log (T_1K_1)}{T_1K_1}}\right)$.\footnote{The requirement that $F$ is regular in addition to being bounded comes from the fact that $\erm^c$ approximates the optimal revenue in an auction with many bidders only for regular distributions. In fact, the sample complexity literature on $\erm^c$ only studies the case of one bidder (which is, in our notation, $K_2 = 1$). In this case, i.e., if the second phase uses posted price auctions, we do not need the regularity assumption. To capture the case of general $K_2$, we make a technical observation that for regular distributions $(1-\epsilon)$ revenue optimality for a single buyer implies $(1-\epsilon)$ revenue optimality for many buyers %(Lemma C.1).
 (Lemma~\ref{lem:regular_one_to_multiple}). 
We do not know if this is true without the regularity assumption or if this observation  -- which may be of independent interest -- was previously known.}
\end{itemize}
\end{theorem}
\noindent The proof is given in % Appendix C.1. 
Appendix~\ref{sec:proof_thm:TPM_utility_revenue_combined}. 
This theorem makes explicit the fact that in order to simultaneously obtain approximate BIC and approximate revenue optimality,
%the auctioneer needs to carefully choose
$T_1$ cannot be too small nor too large: for approximate revenue optimality we need $T_1 \ll T$ and for approximate BIC we need, e.g., $T_1 \gg (v^*)^2 \msum^{4}\log^6(v^*\msum)/K_1$ for MHR distributions, and $T_1\gg D^{11}\msum^5\log^6(D\msum)/K_1$ for bounded distributions. When setting the parameters in this way, both $\epsilon_1$ and $\epsilon_2$ go to $0$ as $T\rightarrow \infty$.

\subsection{Multi-Unit Extension}\label{sec:multi-unit-extension}
% \subsubsection{Multiple Units}
The auction in the exploitation phase can be generalized to a multi-unit Vickrey auction with anonymous reserve, where $k\ge 1$ identical units of an item are sold to $K_2$ unit-demand bidders and among those bidders whose bids are greater than the reserve price $p$, at most $k$ bidders with largest bids win the units and pay the maximum between $p$ and the $(k+1)$-th largest bid. The multi-unit Vickrey auction with an anonymous reserve price is revenue-optimal when the value distribution is regular, and the optimal reserve price does not depend on $k$ or $K_2$ according to \citet{myerson1981optimal}. Thus the optimal reserve price can also be found by $\erm^c$. 
All our results concerning truthfulness, e.g., \cref{thm:connection_TPM_PRM_combined}, still hold for the multi-unit extension with any $k\ge 1$. Moreover, \cref{thm:TPM_utility_revenue_combined} also holds because we have already considered the multi-unit extension in its proof in 
% Appendix C. 
Appendix~\ref{sec:proof_thm:TPM_utility_revenue_combined}.

\subsection{Two-Phase ERM Algorithm in Repeated Auctions}
The two-phase model with ERM as the price learning function can be seen as a learning algorithm in the following setting of repeated auctions against strategic bidders: there are $T$ rounds of auctions, there are $K\ge 1$ bidders in each auction, and each bidder participates in at most $\overline{m}$ auctions.  The algorithm, which we call ``two-phase ERM'', works as follows: in the first $T_1$ rounds, run any truthful, prior-independent auction $\mathcal{M}$ (e.g., the second price auction with no reserve); in the later $T_2=T-T_1$ rounds, run second price auction with reserve $p=\erm^c(b_{1}, \ldots, b_{T_1K})$ where $b_1, \ldots, b_{T_1K}$ are the bids from the first $T_1$ auctions.  $T_1$ and $c$ are adjustable parameters.  

In repeated games, one may also consider $\epsilon$-perfect Bayesian equilibrium ($\epsilon$-PBE) as the solution concept besides $\epsilon$-BNE. A formal definition is given in Appendix \ref{sec:pbe} but roughly speaking, $\epsilon$-PBE requires that the bidding of each bidder at each round of auction $\epsilon$-approximately maximizes the total expected utility in all future rounds, conditioning on any observed history of allocations and payments.  Note that the history may leak some information about the historical bids of other buyers and these bids will affect the seller's choice of mechanisms in future rounds.  Similar to the $\epsilon$-BNE notion, we can show that the two-phase ERM algorithm satisfies: (1) truthful bidding is an $O\left(\log^2(T_1K)\sqrt[3]{\frac{D^{11} K\msum}{T_1}}\right)$-PBE; (2) $\left(1-O\left(\frac{T_1}{T} + \sqrt{\frac{D\log(T_1K)}{T_1K}}\right)\right)$ revenue optimality, for bounded distributions; and similar results for MHR distributions.  Choosing $T_1=\tilde O(T^{\frac{2}{3}})$, which maximizes the revenue, we obtain $\tilde O(T^{-\frac{2}{9}})$-truthfulness and $(1-\tilde O(T^{-\frac{1}{3}}))$ revenue optimality, where we assume $D$, $\msum$, and $K$ to be constant.\footnote{The $\tilde O$ notation omits polylogarithmic terms.} 

Under the same setting, \citet{liu2018learning} and \citet{NIPS2019_abernethy_learning} design approximately truthful and revenue optimal learning algorithms using differential privacy techniques. % instead of pure ERM.  
%to limit the extent to which a buyer's strategic behavior in a single round can affect the learning outcome in future rounds, thus establishing approximate truthfulness.
We can compare two-phase ERM and their algorithms. 
%Firstly, they make a similar assumption as ours that $\msum=o(\sqrt{T/K})$.
Firstly, they make a similar assumption as ours that $\msum=o(\sqrt{T/K})$, in order to obtain approximate truthfulness and revenue optimality at the same time. 
In terms of truthfulness notion, \citet{liu2018learning} assume that bidders play an exact PBE instead of $\epsilon$-PBE, so their quantitative result is incomparable with ours.  Their notion of exact PBE is too strong to be practical because bidders need to collect a lot of information about other bidders and do a large amount of computation to find the exact equilibrium, while our notion guarantees bidders of approximately optimal utility as long as they bid truthfully.
Although our truthfulness bound is worse than the bound of \cite{NIPS2019_abernethy_learning}, which is $\tilde O(\frac{1}{\sqrt{T}})$, we emphasize that their $\epsilon$-truthfulness notion is \emph{weaker} than ours: in their definition, each bidder cannot gain more than $\epsilon$ in current and future rounds if she deviates from truthful bidding \emph{only in the current round}, given any fixed future strategy.  But in our definition, each bidder cannot gain more than $\epsilon$ if she deviates \emph{in current and all future rounds}.  Our algorithm is easier to implement and more time-efficient than theirs, and works for unbounded distribution while theirs only support bounded distributions because they need to discretize the value space.

\section{A Second Application: Uniform-Price Auctions}\label{sec:uniform-price}

The notion of an incentive-awareness measure (recall Definition \ref{def:discount}) has implications regarding the classic uniform-price auction model, which we believe are of independent interest.  In a static uniform-price auction we have $N$ copies of a good and $N$ unit-demand bidders with i.i.d.~values $\bm{v}$ from $F$ that submit bids $\bm{b}$. The auctioneer then sets a price $p=P(\bm{b})$.  Each bidder $i$ whose value $v_i$ is above or equal to $p$ receives a copy of the good and pays $p$, obtaining a utility of $v_i-p$; otherwise the utility is zero. 
\citet{azevedo2018strategy} show that this auction is ``incentive-compatible in the large'' which means that truthfulness is an $\epsilon$-BNE and $\epsilon$ goes to zero as $N$ goes to infinity.  They assume bidders' value distribution has a finite support and their bids must be chosen from this finite support as well.  They mention that allowing continuous supports and arbitrary bids is challenging.  
%show this result only for distributions with finite support size, and mention as an open problem the case of other distributions.

In this context, taking $P=\erm^c$ is very natural when the auctioneer aims to maximize revenue. Indeed, \citet{goldberg2006competitive} suggest to use the uniform-price auction with $P=\erm^c$, where $c=\frac{1}{N}$, as a revenue benchmark for evaluating other truthful auctions they design.

When the price function is $P=\erm^{c=\frac{1}{N}}$, our analysis of the incentive-awareness measure generalizes the result of \cite{azevedo2018strategy} to bounded and to MHR distributions. Moreover, we generalize their result to the case where coalitions of at most $m$ bidders can coordinate bids and jointly deviate from truthfulness. %We call such truthfulness ``$m$-group incentive-compatibility''.

\begin{theorem}\label{thm:uniform-price-bne}
%In the uniform-price auction, assuming that at most $m$ bidders can jointly coordinate a profitable deviation, truthful bidding is an $\epsilon$-BNE (we call this \emph{$(m, \epsilon)$-group BIC}), where,
%\new{In the uniform-price auction, there is no coalition of at most m bidders that any bidder among them can obtain $\epsilon$ more utility by any unilateral joint deviation from truthful bidding, where} 
In the uniform-price auction, suppose that any $m$ bidders can jointly deviate from truthful bidding, then no bidder can obtain $\epsilon$ more utility (we call this \emph{$(m, \epsilon)$-group BIC}), where,
\begin{itemize}
\item for any $P$ and any bounded $F$, 
$\epsilon=D \Delta^{\interim}_{N, m}$, and
\item for any MHR distribution $F$, if we fix $P=\erm^c$ with $\frac{m}{N}\le c\le \cupperboundmhr$ and $m=o(\sqrt{N})$, then
$\epsilon=O\left(v^* \Delta^{\interim}_{N, m}\right)+O\left(\frac{v^*}{\sqrt{N}}\right)$,
where $v^*=\argmax_v\{v[1-F(v)]\}$.
\end{itemize}

The constants in big $O$'s are independent of $F$ and $c$. 
\end{theorem}

\begin{proof} [Proof of \cref{thm:uniform-price-bne} for bounded distributions]
Denote a coalition of $m$ bidders by an index set $I\subseteq\{1, \ldots, N\}$, and the true values of all bidders by $(v_I, v_{-I})$. When other bidders bid $v_{-I}$ truthfully, and the coalition bids $b_I$ instead of $v_I$, the reduction of price is at most
$$P(v_I, v_{-I}) - P(b_I, v_{-I}) \le P(v_I, v_{-I})\delta_I(v_I, v_{-I})\le P(v_I, v_{-I})\delta_{m}^{\interim}(v_{-I}) \le D\delta_{m}^{\interim}(v_{-I}),$$ by Definition \ref{def:discount} and by the fact that all values are upper-bounded by $D$. 
Then for each bidder $i\in I$, the increase of her utility by such a joint deviation is no larger than the reduction of price, i.e.
\begin{equation}
\begin{aligned}
\E_{v_{-I}} \left[ u_i(v_I, P(b_I, v_{-I}))-u_i(v_I, P(v_I, v_{-I}))\right] & \le \E_{v_{-I}} \left[ P(v_I, v_{-I}) - P(b_I, v_{-I})\right] \\
& \le D\E_{v_{-I}} \left[ \delta_{m}^{\interim}(v_{-I})\right] = D\Delta_{N, m}^{\interim}. \nonumber
\end{aligned}
\end{equation}
\end{proof}
\noindent
% The proof of this theorem is very similar to the proof of Theorem~\ref{thm:connection_TPM_PRM_combined} and hence omitted. 
The proof of this theorem for MHR distributions is similar to the proof of Theorem~\ref{thm:connection_TPM_PRM_combined}, thus omitted. %we omit it due to space limit.

Combining with Theorem~\ref{thm:discount_combined}, we conclude that the uniform-price auction with $P=\erm^c$ (for the $c$'s mentioned there) is $(m, \epsilon)$-group BIC with $\epsilon$ converging to zero at a rate not slower than $O(m^{2/3} \frac{\log^2 N }{N^{1/3}})$ for bounded distributions and $O(m \frac{\log^3 N }{\sqrt{N}})$ for MHR distributions (constants in these big $O$'s depend on distributions).

Theorem~\ref{thm:uniform-price-bne} also generalizes the result in \cite{lavi2019redesigning} which is only for bounded distributions and $m=1$.

%\noindent
%On top of that, Theorem \ref{thm:discount_combined} has additional implications regarding the approximate truthfulness of the uniform-price auction. Specifically, \citet{lavi2019redesigning} observes that truthful bidding in the uniform-price digital goods auction is an $\epsilon$-BNE with $\epsilon = D\Delta^{\interim}_{N, 1}$ for distributions upper-bounded by $D$. Theorem \ref{thm:discount_combined} therefore implies that, for bounded distributions as well as for MHR distributions, the uniform-price digital goods auction is approximately truthful when the price is determined using $\erm^c$.

\section{More Discussions on Incentive-awareness Measures}\label{sec:discussion}

\subsection{Overview of the Proof for Upper Bounds on $\Delta^{\interim}_{N, m}$ }\label{sec:overview_proof_delta}

Here we provide an overview of the proof of Theorem~\ref{thm:discount_combined}.  Details are in Appendix \ref{sec:discount_ratio_upper-bound}.

Firstly, we show an important property of $\erm^c$: suppose $c\ge \frac{m}{N}$, for any $m$ values $v_I$, any $N-m$ values $v_{-I}$, and any $m$ values $\overline{v}_I$ that are greater than or equal to the maximum value in $v_{-I}$, we have $\erm^c(\overline{v}_I, v_{-I}) \ge \erm^c(v_I, v_{-I})$.
%As a consequence, $\delta^{\interim}_m(v_{-I}) = \lim_{v_I \rightarrow +\infty} \delta_I(v_I, v_{-I})$. Moreover, since $\erm^c$ ignores the highest $m$ values, if we let $\overline{v}_I$ be $m$ copies of the maximal value in $v_{-I}$, we have
As a consequence, $\delta^{\interim}_m(v_{-I}) = \delta_I(\overline{v}_I, v_{-I})$. 

Based on this property, we transfer the expectation in the incentive-awareness measure in the following way:
\begin{align}
& \Delta^{\interim}_{N, m} ={} \E[\delta^{\interim}_m(v_{-I})] 
={} \E[\delta_I(\overline{v}_I, v_{-I})]
={} \int_0^1 \Pr[\delta_I(\overline{v}_I, v_{-I})>\eta]\dd \eta \nonumber \\
& \le{}\int_0^1\left(\Pr[\delta_I(\overline{v}_I, v_{-I})>\eta \mid \overline{\event} ]\Pr[\overline{\event}] + \Pr[\event]\right) \dd \eta\nonumber 
={}\int_0^1\Pr[\delta_I(\overline{v}_I, v_{-I})>\eta\, \land \overline{\event} ]\dd \eta + \Pr[\event], \nonumber
\end{align}
where $\event$ denotes the event that the index $k^*=\argmax_{i> cN}\{iv_{i}\}$ (which is the index selected by $\erm^c$) satisfies $k^*\le dN$, $\overline{\event}$ denotes the complement of $\event$, and the probability in $\Pr[\event]$ is taken over the random draw of $N-m$ i.i.d.~samples from $F$, with other $m$ samples fixed to be the upper bound (can be $+\infty$) of the distribution.
For any value distribution, we prove that the first part 
\[\int_0^1\Pr[\delta_I(\overline{v}_I, v_{-I})>\eta\, \land \overline{\event} ]\dd \eta \le O \left(\sqrt[3]{\frac{m^2} {d^8}}\frac{\log^2 N}{\sqrt[3]{N}}\right), \]
%with some constructions of auxiliary events and the union bound techniques.
with some constructions of auxiliary events and involved probabilistic argument.
And we further tighten this bound to 
$O \left(\frac{m}{d^{7/2}}\frac{\log^3 N}{\sqrt{N}}\right)$
for MHR distribution by leveraging its properties.

The final part of the proof is to bound $\Pr[\event]$.
For bounded distribution, we choose $d=1/D$.
Since the support of the distribution is bounded by $[1, D]$, $N \cdot v_N \geq N$, while for any $k\le dN$, $kv_{k} \le (\frac{1}{D}N)D = N \le N v_{N}$.
$\erm^c$ therefore never chooses an index $k\le dN$ (recall that in case of a tie, $\erm^c$ picks the larger index). This implies $\Pr[\event] = \Pr[k^*\le dN] = 0$ for bounded distribution.
For MHR distribution, we choose $d=\frac{1}{2e}$ and show $\Pr[\event]=O\left(\frac{1}{N}\right)$. As the corresponding proof is quite complicated, we omit it here.

\subsection{Lower Bounds on $\Delta^{\interim}_{N, m}$ and on the Approximate BIC Parameter, $\epsilon_1$}
Theorem~\ref{thm:discount_combined} gives an upper bound on $\Delta^{\interim}_{N, m}$ for bounded and MHR distributions and for a specific range of $c$'s.  Here we briefly discuss the lower bound, with details given in 
% Appendix F.
Appendix \ref{sec:appendix-lower-bound}.    
%When one considers respective lower bounds, a preliminary question would be: how does the choice of $c$ affect the possible lower bound? The following result shows that it is enough to prove a lower bound for one specific $c$ in the range of allowed $c$'s. The same lower bound will then hold for all $c$'s in that range.

\citet{lavi2019redesigning} show that for the two-point distribution $v=1$ and $v=2$, each w.p.~0.5,
$\Delta^{\interim}_{N, 1}=\Omega(N^{-1/2})$, when $c=1/N$.  We adopt their analysis to provide a similar lower bound for $[1, D]$-bounded distributions and the corresponding range of $c$'s. Let $F$ be a two-point distribution where for $X\sim F$, $\Pr[X=1]=1-1/D$ and $\Pr[X=D]=1/D$. 
\begin{theorem}
For the above $F$, for any $c\in[\frac{m}{N}, \frac{1}{2D}]$, $\erm^c$ gives $\Delta^{\interim}_{N, m} = \Omega(\frac{1}{\sqrt{N}})$ where the constant in $\Omega$ depends on $D$.
\end{theorem}

%{\bf A lower bound on the approximate truthfulness parameter, $\epsilon_1$. }
Note that $\Delta^{\interim}_{N, m}$ only upper bounds the $\epsilon_1$-BIC parameter $\epsilon_1$ in the two-phase model:
a lower bound on $\Delta^{\interim}_{N, m}$ does not immediately implies a lower bound on $\epsilon_1$.  Still, a direct argument will show that the above distribution $F$ gives the same lower bound on $\epsilon_1$.  For simplicity let $K_1=K_2=2$ and suppose bidder $i$ participates in $m_1$ and $m_2$ auctions in the two phases, respectively. Let $N=T_1K_1$ and assume $m_1=o(\sqrt{N})$. Suppose the first-phase mechanism $\mathcal{M}$ is the second price auction with no reserve price.
Then in the two-phase model with $\erm^c$, $\epsilon_1$ must be $\Omega( \frac{m_2}{\sqrt{N}})$ to guarantee $\epsilon_1$-BIC.
%for any $c\in[\frac{m_1}{N}, \frac{1}{2D}]$.
%where the constant in $\Omega$ depends on $D$.

It remains open to prove a lower bound for MHR distributions, and to close the gap between our $O(N^{-1/3} \log^2 N)$ upper bound and the $\Omega(N^{-1/2})$ lower bound for bounded distributions.

\subsection{Unbounded Regular Distributions}
\cref{thm:TPM_utility_revenue_combined} shows that, in the two-phase model, approximate incentive-compatibility and revenue optimality can be obtained simultaneously for bounded (regular) distributions and for MHR distributions.
%
%\new{We have shown that approximate truthfulness and revenue optimality can be obtained simultaneously for bounded (regular) distributions and for MHR distributions in \cref{thm:TPM_utility_revenue_combined}.
%
A natural question would then be: what is the largest class of value distribution we can consider? Note that for non-regular distributions, \citet{myerson1981optimal} shows that revenue optimality cannot be guaranteed by anonymous reserve price, so ERM is not a correct choice. Thus we generalize our results to the class of regular distributions that are unbounded and not MHR. Here we provide a sketch, with details given in 
% Appendix G. 
Appendix \ref{sec:appendix-unbounded-regular-distributions}.

Our results can be generalized to $\alpha$-strongly regular distributions with $\alpha > 0$. As defined in \cite{cole2014the}, a distribution $F$ with positive density function $f$ on its support $[A, B]$ where $0\le A\le B\le+\infty$ is \emph{$\alpha$-strongly regular} if the virtual value function $\phi(x) = x-\frac{1-F(x)}{f(x)}$ satisfies $\phi(y)-\phi(x)\ge \alpha (y-x)$
whenever $y>x$ (or $\phi'(x) \ge \alpha$ if $\phi(x)$ is differentiable). As special cases, regular and MHR distributions are $0$-strongly and $1$-strongly regular distributions, respectively.
For any $\alpha>0$, we obtain bounds similar to MHR distributions on $\Delta^{\interim}_{N, m}$ and on approximate incentive-compatibility in the two-phase model and the uniform-price auction.
Specifically, if $F$ is $\alpha$-strongly regular then $\Delta^{\interim}_{N,m}=O\left(m\frac{\log^3 N}{\sqrt{N}}\right)$, if $m=o(\sqrt{N})$ and $\frac{m}{N} \le \left(\frac{\log N}{N}\right)^{1/3}\le c \le \frac{\alpha^{1/(1-\alpha)}}{4}$.  
It remains an open problem for future research whether $\erm^c$ is incentive-compatible in the large for regular but not $\alpha$-strongly regular distributions for any $\alpha>0$.
For these distributions the choice of $c$ must be more sophisticated since it creates a clash between approximate incentive-compatibility and approximate revenue optimality. 
% A large $c$ (for example, a constant) will hurt revnue optimality and a too small $c$ will hurt truthfulness, 
%%Luckily, \citet{huang2018making} shows that if $c(N)\to 0$ as $N \rightarrow \infty$ then approximate revenue optimality can be satisfied.
%%On the other hand, if $c$ is too small, truthfulness will be violated, 
% as discussed in the following two examples (assume $m=1$ for simplicity).
Intuitively, a large $c$ (for example, a constant) will hurt revenue optimality and a too small $c$ will hurt incentive-compatibility. 
% In Appendix G,
In Appendix \ref{sec:appendix-unbounded-regular-distributions},
we provide examples and proofs to formally illustrate such a fact, and further discuss our conjecture that some intermediate $c$ can maintain the balance between incentive-compatibility and revenue optimality. 

\newpage

\section*{Broader Impact}

% Authors are required to include a statement of the broader impact of their work, including its ethical aspects and future societal consequences. 
% Authors should discuss both positive and negative outcomes, if any. For instance, authors should discuss a) 
% who may benefit from this research, b) who may be put at disadvantage from this research, c) what are the consequences of failure of the system, and d) whether the task/method leverages
% biases in the data. If authors believe this is not applicable to them, authors can simply state this.

% Use unnumbered first level headings for this section, which should go at the end of the paper. {\bf Note that this section does not count towards the eight pages of content that are allowed.}

% \new{
% Since this paper is about algorithmic game theory, which is a theoretical topic, broader impacts may not be applicable to us.}

This work is mainly theoretical.  It provides some intuitions and guidelines for potential practices, but does not have immediate societal consequences.  A possible positive consequence is: the auction we consider uses an anonymous reserve price, while most of the related works on repeated auctions use unfair personalized prices.  We do not see negative consequences. 

% \begin{ack}
% % Use unnumbered first level headings for the acknowledgments. All acknowledgments
% % go at the end of the paper before the list of references. Moreover, you are required to declare 
% % funding (financial activities supporting the submitted work) and competing interests (related financial activities outside the submitted work). 
% % More information about this disclosure can be found at: \url{https://neurips.cc/Conferences/2020/PaperInformation/FundingDisclosure}.

% % Do {\bf not} include this section in the anonymized submission, only in the final paper. You can use the \texttt{ack} environment provided in the style file to autmoatically hide this section in the anonymized submission.

% This research was partially supported by the ISF-NSFC joint research program (ISF grant No. 2560/17). 

% \end{ack}

% \section*{References}
\small
\bibliographystyle{plainnat}
\bibliography{bibliography}

%================================================
%================================================

  \newpage
  \appendix
  % \section{Introduction 2}
% %%%%%%%%%%%%%%%%%%%%%%%%%%%%%%%%%%%%%%%%%%%%%%%%%%%%%%%%%%%%%%%%%%%%%%%%%%%%%%
% \input{old_intro+model.tex}

% \section{Missing Proofs From Section~\ref{sec:discount-ratio-definition}}

\section{Useful Facts}
In this section we present some facts about $\erm^c$ and incentive-awareness measures, some definitions about value distributions, and some useful lemmas that will be used throughout. 

\subsection{Facts about $\erm^c$ and Incentive-Awareness Measures}
\begin{claim}\label{claim:interim_to_delta_price}
Let $P=\erm^c$, where $c\ge \frac{m}{N}$.  For any $v_I\in \mathbb R_+^m$,  $v_{-I}\in\mathbb R_+^{N-m}$, let $\overline{v}_I$ denote $m$ values such that
$\min \overline{v}_I  \ge  \max v_{-I}$. Then 
we have $P(\overline{v}_I, v_{-I}) \ge P(v_I, v_{-I})$.
\end{claim}
\begin{proof}
%\subsection{Proof of \cref{claim:interim_to_delta}}
\label{sec:proof_claim:interim_to_delta}
%Recall that $\delta^{\interim}_m(v_{-I})=\sup_{v_I}\delta_I(v_I, v_{-I})$ and $\delta_I(v_I, v_{-I}) = 1-\frac{\inf_{b_I}P(b_I, v_{-I})}{P(v_I, v_{-I})}$ where $P=\erm^c$. 
Let $\overline{v}\coloneqq\max v_{-I}$ be the largest value in $v_{-I}$ and $\overline{v}_I$ be $m$ copies of $\overline{v}$. It suffices to show that $P(\overline{v}_I, v_{-I}) \ge P(v_I, v_{-I})$ since $P=\erm^c$ ignores the largest $m$ samples, given $cN\ge m$. If $v_i \ge \overline{v}$ for each $i\in I$, then we have $P(\overline{v}_I,v_{-I})=P(v_{I},v_{-I})$ directly. If there exists some $i\in I$ such that $v_i<\overline{v}$, then we increase $v_i$ to $\overline{v}$ and show that for any such $i$ and $(v_i,v_{-i})$, $P(\overline{v}, v_{-i}) \ge P(v_i, v_{-i})$. Let $v'= P(v_i, v_{-i})$, then one can verify (assuming $\erm^c$ picks the smaller value when there are ties) that (1) $P(v', v_{-i}) = P(v_i, v_{-i})$, and (2) $P(v', v_{-i}) \le P(\overline{v}, v_{-i})$, implying $P(\overline{v}, v_{-i}) \ge P(v_i, v_{-i})$.
%show that $P(\overline{v}_{I}, v_{-I}) \ge P(v_I, v_{-I})$ for any $v_I, v_{-I}$ for $P=\erm^c$, we will repeatedly increase each $v_i,\ i\in I$ to $\overline{v}$ if $v_i < \overline{v}$.
%\old{
%It is thus sufficient to show that for any $i$ and $(v_i,v_{-i})$, $P(\overline{v}, v_{-i}) \ge P(v_i, v_{-i})$. Let $v'= P(v_i, v_{-i})$, then one can verify (assuming $\erm^c$ picks the smaller value when there are ties) that (1) $P(v', v_{-i}) = P(v_i, v_{-i})$, and (2) $P(v', v_{-i}) \le P(\overline{v}, v_{-i})$, implying $P(\overline{v}, v_{-i}) \ge P(v_i, v_{-i})$.}
\end{proof}

\begin{claim}\label{claim:interim_to_delta}
Let $P=\erm^c$, where $c\ge \frac{m}{N}$.  For any $v_I\in \mathbb R_+^m$,  $v_{-I}\in\mathbb R_+^{N-m}$, let $\overline{v}_I$ be any $m$ values that are greater than or equal to the maximal value in $v_{-I}$.  Then $\delta^{\interim}_m(v_{-I}) = \delta_I(\overline{v}_I, v_{-I})$.
\end{claim}
\begin{proof}
Recall the definition
\[
\delta^{\worst}_m(v_{-I})=
\sup_{v_I\in \mathbb R_+^m}{\delta_I(v_I, v_{-I})} 
=\sup_{v_I \in \mathbb R_+^m}\left\{ 1 - \frac{ \inf_{b_I\in\mathbb R_+^m} \erm^c(b_I, v_{-I})}
{\erm^c(v_I, v_{-I})} \right\}.\]
\cref{claim:interim_to_delta_price} immediately implies $\delta^{\interim}_m(v_{-I}) = \lim_{v_I \rightarrow +\infty} \delta_I(v_I, v_{-I})$.  
Moreover, since $\erm^c$ ignores the highest $cN\ge m$ values, we have $\erm^c(\overline{v}_I, v_{-I}) = \erm^c(\overline{v}'_I, v_{-I})$ as long as both $\overline{v}_I$ and $\overline{v}'_I$ are greater than or equal to $\max v_{-I}$, no matter what they are exactly.  Thus $\delta^{\interim}_m(v_{-I}) = \delta_I(\overline{v}_I, v_{-I}) = \delta_I(\overline{v}'_I, v_{-I})$. 
%if we let $\overline{v}_I$ be $m$ copies of the maximal value in $v_{-I}$, we have
% \delta^{\interim}_m(v_{-I}) = \delta_I(\overline{v}_I, v_{-I})$.
\end{proof}
Therefore, we will use $\overline{v}_I$ to denote any $m$ values that are greater than or equal to $\max v_{-I}$, for example, $m$ copies of $\max v_{-I}$ or $m$ copies of~``$+\infty$''.  We always have $\delta^{\interim}_m(v_{-I}) = \delta_I(\overline{v}_I, v_{-I})$. 

\subsection{Quantiles and Revenue Curves of Value Distributions}\label{sec:distribution-properties}
%{\bf Quantiles and the revenue curve. }
%\paragraph{Quantiles and the revenue curve.}
%%%%%%%%%%%%%%%%%%%%%%%%%%%%%%%%%%%%%%%
For a distribution $F(v)$, define the \emph{quantile} $q(v)=1-F(v)$ as a mapping from value space to quantile space. Inversely, $v(q)=q^{-1}(v)=F^{-1}(1-q)$ is the mapping from quantile space to value space (i.e., w.p.~$q$ a buyer's value will be at least $v(q)$).
%In a posted-price auction, with probability $q(v)=1-F(v)$ the bidder has value greater than $v$ and pays $v$ to buy the item
Define the \emph{revenue curve} $R(q)=qv(q)$ as the expected revenue for the seller by posting price $v(q)$. Let $R^*=\max_q\{R(q)\}$ denote the optimal revenue the seller can obtain with one bidder, and $q^*=\argmax_q R(q), v^*=v(q^*)$. When there are several i.i.d.~bidders with a regular value distribution, $v^*$ is the optimal reserve price in a second price auction, and such an auction is revenue optimal \cite{myerson1981optimal}.
Any bounded distribution satisfies $q^*\ge\frac{1}{D}$ because for any $q<\frac{1}{D}$, $qv(q)\le qD < 1\le 1v(1)\le R^*$.  Any MHR distribution has a unique $q^*$ and $q^* \ge \frac{1}{e}$ \cite{hartline2008optimal}.

\subsection{Concentration Inequality}
For a distribution $F$, draw $N$ samples and sort them non-increasingly, $v_1\ge v_2\ge \cdots \ge v_N$.  Let $q_j=q(v_j)$ denote their quantiles.  The ratio $j/N$ is the \emph{empirical quantile} of value $v_j$ since $j/N$ is the quantile of $v_j$ in the uniform distribution over $\{v_1, \ldots, v_N\}$. 
The following concentration inequality shows that for each value $v_j$, its empirical quantile $j/N$ is close to its true quantile $q_j$ with high probability, when $m$ samples are fixed to be $+\infty$ while other $N-m$ samples are i.i.d.~drawn from $F$.  

\begin{lemma} \label{lem:concentration}
Draw $N-m$ i.i.d.~samples from a distribution $F$, and fix $m$ samples to be $+\infty$. Sort these samples non-increasingly: $+\infty=v_1=\cdots =v_m > v_{m+1} \ge \cdots \ge v_N$.   With probability at least $1-\delta$ over the random draw of samples, we have for any $j>m$, 
\[ \left | q_j - \frac{j}{N} \right | \le \sqrt{ \frac{2\ln(2(N-m)\delta^{-1})}{N-m} } + \frac{\ln(2(N-m)\delta^{-1})}{N-m} + \frac{m}{N}. \]
\end{lemma}
\begin{proof}
\label{sec:proof_lem:concentration}
The value $v_j$ ($j>m$) is the $(j-m)$th largest value in $N-m$ i.i.d.~samples from $F$, by using Bernstein inequality (see e.g., Lemma 5 in \citet{guo2019settling}), we know that with probability at least $1-\delta$, $\left | q_j - \frac{j-m}{N-m} \right | \le \sqrt{ \frac{2\ln(2(N-m)\delta^{-1})}{N-m} } + \frac{\ln(2(N-m)\delta^{-1})}{N-m}$. Also note that
$|\frac{j}{N} - \frac{j-m}{N-m}| = \frac{(N-j)}{N-m}\frac{m}{N}< \frac{m}{N}$.  By triangular inequality, 
$\left | q_j - \frac{j}{N} \right | \le \sqrt{ \frac{2\ln(2(N-m)\delta^{-1})}{N-m} } + \frac{\ln(2(N-m)\delta^{-1})}{N-m} + \frac{m}{N}$.
\end{proof}

%%%%%%%%%%%%%%%%%%%%%%%%%%%%%%%%%%%%%%%%%%%%%
%%%%%%%%%%%%%%%%%%%%%%%%%%%%%%%%%%%%%%%%%%%%%
% \section{Missing Proofs from Section \ref{sec:discount_ratio_upper-bound}}\label{sec:proof_discount_ratio_upper-bound}

\section{Main Proof: Upper Bounds on Incentive-Awareness Measures}\label{sec:discount_ratio_upper-bound}

\subsection{Proof of Theorem~\ref{thm:discount_combined}}

Recall the setting of Definition~\ref{def:discount}:
we draw $N$ i.i.d.~values $v_1, \ldots, v_{N}$ from $F$, and we have an additional parameter $m$ which is the number of bids that can be changed in the input of the price learning function. Theorem~\ref{thm:discount_combined} then states an upper bound on $\Delta^{\interim}_{N, m}$ for $P=\erm^c$. For bounded distributions, the theorem follows immediately from the next lemma which is our main technical lemma. Note that this lemma is useful in establishing the bound on $\Delta^{\interim}_{N, m}$ not only for bounded distributions but also for all other distributions. Throughout, we assume that $v_1,\ldots,v_N$ are sorted, so that $v_1 \geq \cdots \geq v_N$.

\begin{lemma}[Main Lemma]\label{lem:main_upper-bound}
Suppose $m=o(\sqrt{N})$. Let $d$ be a constant, $0<d<1$. Suppose $\frac{m}{N}\le c\le\frac{d}{2}$. Let $\event$ be the event that the index $k^*=\argmax_{i> cN}\{iv_{i}\}$ (which is the index selected by $\erm^c$) satisfies $k^*\le dN$. For any non-negative distribution $F$, 
\[\Delta^{\interim}_{N, m} \le O \left(\sqrt[3]{\frac{m^2} {d^8}}\frac{\log^2 N}{\sqrt[3]{N}}\right) + \Pr[\event], \]
where the probability in $\Pr[\event]$ is taken over the random draw of $N-m$ i.i.d.~samples from $F$, with other $m$ samples fixed to be $+\infty$. 
%the upper bound (can be $+\infty$) of the distribution.

\end{lemma}

To see that this lemma immediately implies the theorem for bounded distribution, choose $d=1/D$.
Since the support of the distribution is bounded by $[1, D]$, $N \cdot v_N \geq N$
%the last term of the array $\{iv_{i}\}_{i=1}^N$ is at least $N\cdot 1$,
while for any $k\le dN$, %and sufficiently large $N$ such that $d<\frac{1}{D}-\frac{m}{N}$
\[kv_{k} \le (\frac{1}{D}N)D = N \le N v_{N}. \]
$\erm^c$ therefore never chooses
an index $k\le dN$ (recall that in case of a tie, $\erm^c$ picks the larger index). This implies $\Pr[\event] = \Pr[k^*\le dN] = 0$ and we have the bound in the theorem.

{\em Remark.}
For MHR distributions, Lemma~\ref{lem:mhr_event} shows that
$\Pr[\event]=O\left(\frac{1}{N}\right)$ if we choose $d=\frac{1}{2e}$, so Lemma~\ref{lem:main_upper-bound} already gives an upper bound on $\Delta^{\interim}_{N, m}$. However, we can use some additional properties of MHR distributions to strengthen the bound on the first term in the main lemma to $O \left(\frac{m}{d^{7/2}}\frac{\log^3 N}{\sqrt{N}}\right)$, as explained in Appendix~\ref{sec:improve_MHR}. 
%\end{remark}

\subsection{Proof of the Main Lemma (\cref{lem:main_upper-bound})}

%The proof is composed of three main claims. We feel that the best way to give an overview of the proof is to state these three claims in a formal way and explain why they imply \cref{lem:main_upper-bound}.

%By \cref{claim:interim_to_delta}, $\erm^c(\overline{v}_I, v_{-I}) \ge \erm^c(v_I, v_{-I})$ if $\overline{v}_I$ are greater than the values 
%\noindent The proof is in \cref{sec:proof_claim:interim_to_delta}. As a consequence, $\delta^{\interim}_m(v_{-I}) = \lim_{v_I \rightarrow +\infty} \delta_I(v_I, v_{-I})$.  
%Moreover, since $\erm^c$ ignores the highest $m$ values, if we
Let $\overline{v}_I$ be any $m$ values that are greater than the maximal value in $v_{-I}$.
By \cref{claim:interim_to_delta}, $\delta^{\interim}_m(v_{-I}) = \delta_I(\overline{v}_I, v_{-I})$. Thus,

\begin{align}
\Delta^{\interim}_{N, m} ={}& \E[\delta^{\interim}_m(v_{-I})] 
={} \E[\delta_I(\overline{v}_I, v_{-I})]
={} \int_0^1 \Pr[\delta_I(\overline{v}_I, v_{-I})>\eta]\dd \eta \nonumber \\
\le{}&\int_0^1\left(\Pr[\delta_I(\overline{v}_I, v_{-I})>\eta \mid \overline{\event} ]\Pr[\overline{\event}] ~+~ \Pr[\event]\right) \dd \eta\nonumber \\
={}&\int_0^1\Pr[\delta_I(\overline{v}_I, v_{-I})>\eta\, \land \overline{\event} ]\dd \eta ~+~ \Pr[\event], \label{eqn:Delta_interim_integral}
\end{align}
where $\overline{\event}$ denotes the complement of $\event$. Then the main effort is to upper-bound $\Pr[\delta_I(\overline{v}_I, v_{-I})>\eta\,\land\,\overline{\event}]$. After the random draw of $v_{-I}$, we sort all values non-increasingly, denoted by $\overline{v}_1= \cdots=\overline{v}_m\ge v_{m+1}\ge \cdots\ge v_N$, and let $\overline{q}_1=\cdots=\overline{q}_m\le q_{m+1}\le\cdots\le q_N$ be their quantiles, where $q_j=q(v_j)$.  We use a concentration inequality (Lemma~\ref{lem:concentration}) to argue that for each value $v_j$, its empirical quantile $j/N$ should be close to its true quantile $q_j$ with high probability, as follows: 

\begin{claim}\label{claim:concentration}
Define event $\conc$: 
\[ \conc=\left[\forall j>m, \left | q_j - \frac{j}{N} \right | \le 2\sqrt{ \frac{4\ln(2(N-m))}{N-m} } + \frac{m}{N} \right], \]
then $\Pr[\overline{\conc}] \le \frac{1}{N-m}$,
where the probability is over the random draw of the $N-m$ samples $v_{-I}$. 
\end{claim}
\begin{proof}
Set $\delta=\frac{1}{N-m}$ in \cref{lem:concentration}. 
\end{proof}

Now, define $G(\eta)=\Pr[\delta_I(\overline{v}_I, v_{-I})>\eta\ \land \overline{\event}~\land~\conc]$ for $0 \le \eta \le 1$. We have

\begin{equation}\label{eqn:delta_event_conc}
\Pr[\delta_I(\overline{v}_I, v_{-I})>\eta\,\land \overline{\event}] %\le \Pr[\delta_I(\overline{v}_I, v_{-I})>\eta\,\land\,\overline{\event}\,\land\,\conc] + \frac{1}{N-m}.
\leq G(\eta) + \frac{1}{N-m}.
\end{equation}

%\begin{lemma}\label{lem:first_term}
%For any $0<\eta<1$, if $\eta$ is at least $\Omega\left(\frac{m}{d}\sqrt{\frac{\log (N-m)}{N-m}}\right)$ (for some constant in $\Omega$ to be detailed in the proof), then
%\[\Pr[\delta_I(\overline{v}_I, v_{-I})>\eta\,\land\,\overline{\event}\,\land\,\conc] =  O\left(\frac{m\log^3 N}{d^4\sqrt{N}}\frac{1}{\eta^{3/2}}\right).\] 
%\end{lemma}

\begin{lemma}\label{lem:first_term}
There exists a constant $C=\Theta\left(\frac{m\log^3 N}{d^4\sqrt{N}}\right)$ such that
$\eta>C^{2/3} \Rightarrow G(\eta) \le \frac{C}{\eta^{3/2}}.
$
%Specifically, $C=\Theta\left(\frac{m\log^3 N}{d^4\sqrt{N}}\right)$.
\end{lemma}
Finally we upper-bound the integral in \eqref{eqn:Delta_interim_integral}:

%Denote $C=\Theta\left(\frac{m\log^3 N}{d^4\sqrt{N}}\right)$, and let $G(\eta)=\Pr[\delta_I(\overline{v}_I, v_{-I})>\eta\ \land \overline{\event}~\land~\conc]$. By \cref{lem:first_term}, when $\eta>C^{2/3}$, we have $G(\eta) \le \frac{C}{\eta^{3/2}}$. Therefore: 
\begin{align*}
     \int_0^1\Pr[&\delta_I(\overline{v}_I, v_{-I})>\eta\, \land \overline{\event} ]\dd \eta 
    \le{}
    \int_0^1 \left(G(\eta) + \frac{1}{N-m}\right)\dd \eta &&\text{By \eqref{eqn:delta_event_conc}}\\
    \le{} & \int_0^{C^{\frac{2}{3}}} 1\dd \eta + \int_{C^{\frac{2}{3}}}^1 \frac{C}{\eta^{\frac{3}{2}}}\dd \eta
    +\frac{1}{N-m} && \text{By \cref{lem:first_term}}\\
    %={}& C^\frac{2}{3} + \frac{C}{-\frac{3}{2}+1} - \frac{C}{-\frac 3 2 +1}C^{-1+\frac{2}{3}} +\frac{1}{N-m} \\
    \le{}& 3C^\frac{2}{3} +\frac{1}{N-m}
    ={} O\left(\sqrt[3]{\frac{m^2\log^6 N}{d^8 N}}\right)+\frac{1}{N-m}
    = O \left(\sqrt[3]{\frac{m^2} {d^8}}\frac{\log^2 N}{\sqrt[3]{N}}\right), 
    %&&\text{$\frac{1}{N-m} = o\left(\frac{1}{\sqrt[3]{N}}\right)$ since $m=o(\sqrt{N})$}
    % \le & \left(1+\frac{1}{\frac 3 2 -1}\right)C^\frac{2}{3} = O\left(\sqrt[3]{\frac{m^2\log^4 N}{d^6 N}}\right). 
\end{align*}
which, together with \eqref{eqn:Delta_interim_integral}, concludes the proof of \cref{lem:main_upper-bound}.

\subsection{Proof of Lemma \ref{lem:first_term}}
\label{sec:proof_lem:first_term}

Recall that we need to upper-bound $G(\eta) = \Pr[\delta_I(\overline{v}_I, v_{-I})>\eta\,\land \overline{\event}\,\land\,\conc]$ by $\Theta\left(\frac{m\log^3 N}{d^4\sqrt{N}}\frac{1}{\eta^{3/2}}\right)$. We do this via a union bound of $M+1$ events, where $M$ is a number to be chosen later.
Each event is parameterized by $\eta_t, \theta_t$ for $t=0, \ldots, M$ which are chosen to satisfy the following conditions:
\begin{itemize}
\item $\eta_0=\frac{1}{2}\eta, \eta_1=\eta, \eta_2=2\eta$.
\item For $t \geq 3$, $\eta_t$ can be chosen arbitrarily, as long as $\eta_2<\eta_3<\cdots<\eta_M < 1$.
\item $\eta_{M+1} = 1$.
\item
$\theta_0 = 1, \text{ and } \theta_t = \frac{\eta}{2\eta_{t+1}} \text{ for } t= 1, \ldots, M$.
\end{itemize}

Define the following $M+1$ events $\Bad(\first_t, \second_t)$, where $t=0,\ldots, M$:
\begin{equation}
    \Bad(\first_t, \second_t) = \left[
\text{ there exists $j\ge k^*$ such that } \begin{cases}
&v_j\le (1-\first_t)v_{k^*} \\
&jv_j\ge k^*v_{k^*}-\frac{m}{\second_t}v_{k^*} 
\end{cases} \right]
\end{equation}

The next lemma shows that the union of these events contains the event $[\delta_I(\overline{v}_I, v_{-I})>\eta]\,\land\,\overline{\event}$.

\begin{lemma}\label{lem:thresholds}
Suppose $\frac{2m}{dN}<\eta<1$ and that the parameters $\eta_t,\theta_t$ satisfy the above conditions. If $\delta_I(\overline{v}_I, v_{-I})>\eta$ and $k^*> dN$, then there exists $t \in \{0,\ldots,M\}$ such that the event $\Bad(\eta_t, \theta_t)$ holds.
\end{lemma}
\noindent The proof of this lemma is given in \cref{sec:proof_lem:thresholds}.
Moreover, the next lemma upper-bounds the probability of each of these bad events, when assuming that $\conc$ holds as well.

\begin{lemma}\label{lem:Pr_Bad}
If $\first_t$ and $\second_t$ are at least $\Omega\left(\frac{m}{d}\sqrt{\frac{\log (N-m)}{N-m}}\right)$ (for some constant in $\Omega$ to be detailed in the proof), then 
$\Pr[\Bad(\first_t, \second_t)\,\land\,\overline{\event}\,\land\,\conc] = O\left(\frac{m\log^2 N}{d^4\second_t\sqrt{\first_t^3N}}\right)$.
\end{lemma}
\noindent The proof of \cref{lem:Pr_Bad} is in \cref{sec:proof_lem:Pr_Bad}. %\old{ proof sketch here?}
Now,

\begin{flalign*}
\Pr[\delta_I(\overline{v}_I, v_{-I})>\eta\,\land\,\overline{\event}\,\land\,\conc] \le{}& \sum_{t=0}^M \Pr[\mathrm{Bad}(\eta_t, \theta_t) \,\land\,\overline{\event}\,\land\,\conc] && \text{Lemma \ref{lem:thresholds}}\\
={}& \sum_{t=0}^M O\left(\frac{m\log^2 N}{d^4\theta_t\sqrt{\eta_t^3N}}\right) && \text{Lemma~\ref{lem:Pr_Bad}}\\
%={}& \sum_{t=0}^M O\left(\frac{m\log^2 N}{d^4 \sqrt{N}}\frac{\eta_{t+1}} { \eta\sqrt{\eta_t^3}}\right) \\
={}& O\left(\frac{m\log^2 N}{d^4\sqrt{N}} \cdot \sum_{t=0}^M  \frac{\eta_{t+1}}{\eta}\frac{1}{\eta_{t}^{3/2}} \right) && \text{Definition of } \theta_t %\frac{\eta}{\eta_{t+1}} \ge \eta > \Omega(\frac{m}{dN})
\end{flalign*}
Note that because $\eta_t, \theta_t\ge \frac{\eta}{2}$, the condition of Lemma~\ref{lem:Pr_Bad} is satisfied under the assumption that $\eta \ge \Theta\left(\left(\frac{m\log^3 N}{d^4\sqrt{N}}\right)^{2/3}\right)$ in Lemma~\ref{lem:first_term}. Finally, we choose a sequence of $\{\eta_t\}$ to make the above summation small enough:

\begin{claim}
\label{cl:eta-calculation}
There exist an integer $M$ and parameters $\eta_3,\ldots,\eta_M$ that satisfy the conditions described above, such that 
\begin{align*}
    \sum_{t=0}^M  \frac{\eta_{t+1}}{\eta}\frac{1}{\eta_{t}^{3/2}} = O\left(\frac{\log\log (N-m)}{\eta^{3/2}}\right), 
\end{align*}
assuming $\eta=\Omega\left(\frac{m}{d}\sqrt{\frac{\log (N-m)}{N-m}}\right)$. 
\end{claim}
\noindent The proof of this claim is given in Appendix \ref{proof-of-eta-calculation}. To conclude the proof, 
\[\Pr[\delta_I(\overline{v}_I, v_{-I})>\eta\,\land\,\overline{\event}\,\land\,\conc] \le O\left(\frac{m\log^2 N}{d^4\sqrt{N}}\cdot \frac{ \log\log (N-m)}{\eta^{3/2}} \right)=  O\left(\frac{m\log^3 N}{d^4\sqrt{N}} \frac{1}{\eta^{3/2}}\right).\]

{\em Remark.}
This proof is inspired by a proof in \citet{yao2018incentive}. We improve upon that proof in two aspects: (1) Our definition of a sequence of bad events (Lemma~\ref{lem:thresholds}) improves upon similar single bad events defined in \citet{yao2018incentive} and \citet{lavi2019redesigning}; (2) %the proof below improves upon that of \citet{yao2018incentive} in two aspects. First, \citet{yao2018incentive} does not give a specific convergence rate of $G(\eta)$ but only shows that it convergences to zero as $N$ goes to infinity. Second,
\citet{yao2018incentive} only considers bounded and continuous distributions, while our proof works for arbitrary distributions. This is because \citet{yao2018incentive} works in the value space when upper-bounding the probability of bad events over the random draw of values $v_{-I}$ (Lemma~\ref{lem:Pr_Bad}), but we work in the quantile space, which circumvents the boundedness assumption and deals with discontinuity. To argue in the quantile space, we need $\conc$ to show that $q_jv_j$ approximates $jv_j$ in the proof of Lemma~\ref{lem:Pr_Bad}.

\subsection{Proof of Lemma \ref{lem:thresholds}}\label{sec:proof_lem:thresholds}
Suppose $\delta_I(\overline{v}_I, v_{-I})>\eta$ and $k^*> dN$. By definition, there exist $m$ bids $b_I\in\mathbb R_+^m$ such that $\erm^c(b_I, v_{-I}) < (1-\eta)v_{k^*}$. Without loss of generality, we can assume that all $m$ bids are identical, as shown in the following claim:
%As shown in the following claim, there is no need to choose different bids in $b_I$, that is, choosing $m$ copies of a single bid $b$ suffices. Moreover, such a bid satisfies $b=\erm^c(b, \ldots, b, v_{-I})$. 
\begin{claim}\label{claim:multiple_bid}
For any sorted $\bm{v}=(v_1\ge \cdots\ge v_N)$. Let $I=\{1, \ldots, m\}$, let
\[b^*=\argmin_{b\in\mathbb R_+} \erm^c(\overbrace{b, \ldots, b}^{m\text{ copies}}, v_{-I}).\]
then $b^*=\erm^c(b^*, \ldots, b^*, v_{-I})= \min_{b_I\in\mathbb{R}_+^m} \erm^c(b_I, v_{-I})$.
\end{claim}

\begin{proof}
We will show that for any vector of $m$ bids $b_I=(b_1, \ldots, b_m)$ that minimizes $\erm^c(b_I, v_{-I})$, we can construct another vector $b_I'=(b, \ldots, b)$ such that $\erm^c(b_I', v_{-I})=\erm^c(b_I, v_{-I})=b$. Because $b_I$ minimizes $\erm^c(b_I, v_{-I})$, we can assume that there is a bid $b_{i^*}$ such that $\erm^c(b_I, v_{-I})=b_{i^*}$ (otherwise, we can decrease the bids in $b_I$ without increasing the price). Let $b=b_{i^*}$. For any $b_j>b$, decrease $b_j$ to $b$, then the price does not change. For any $b_j<b$, increase $b_j$ to $b$, then the price does not increase; and if the price decreases, then it contradicts the fact that $\erm^c(b_I, v_{-I})$ is minimized. In this way, we change all bids in $b_I$ to $b$, without affecting the price. 
\end{proof}

By \cref{claim:multiple_bid}, there exists $b\in\mathbb R_+$ which equals $\erm^c(b, \ldots, b, v_{-I})$ and satisfies
\begin{equation}\label{eqn:b}
    b< (1-\eta)v_{k^*}.
\end{equation}

Choose index $i$ for which $v_{i} \ge b > v_{i+1}$. Assume for now $i\le N-1$, we will postpone the analysis for $i=N$ to the end. Now we show that setting $j=i$ or $i+1$ will satisfy the lemma. Clearly $i\ge k^*$. The change of the bids vector caused by $b$ is: 
\[(\overline{v}_1, \ldots, \overline{v}_m, v_{m+1}, \ldots, v_{k^*}, \ldots, v_i, v_{i+1}, \ldots, v_N) \to (v_{m+1}, \ldots,v_{k^*}, \ldots,  v_{i}, \overbrace{b,\ldots, b}^{m\text{ times}}, v_{i+1}, \ldots, v_N).\]  
Note that $k^*-m> dN - m\ge cN$, so $v_{k^*}$ will not be ignored by $\erm^c$ after the change of bids. Then in order for $b$ to be chosen by $\erm^c$, we need:
\begin{equation}\label{eqn:ibi}
i\cdot b \ge (k^*-m)\cdot v_{k^*} = k^*v_{k^*} - mv_{k^*}. 
\end{equation}

We will choose $j$ depending on how large $v_i$ is:
\begin{enumerate}[(a)]
    \item If $v_i<(1-\frac{1}{2}\eta)v_{k^*}=(1-\eta_0)v_{k^*}$, we set $j=i$ and $t=0$. Clearly, $iv_i\ge ib\ge k^*v_{k^*} - mv_{k^*}$.
    \item If $v_i\ge (1-\frac{1}{2}\eta)v_{k^*}$, then we set $j=i+1$ and choose the $t$ ($1\le t\le M$) such that
    \begin{equation}\label{eqn:t_sandwich}
        (1-\eta_t)v_{k^*} \ge v_{i+1} > (1-\eta_{t+1})v_{k^*}. 
    \end{equation}
    To see why $(i+1)v_{i+1}\ge k^*v_{k^*} - \frac{2\eta_{t+1}}{\eta}m v_{k^*}$ holds, first we write $b$ as a convex combination of $v_i$ and $v_{i+1}$: $b=(1-\lambda)v_i + \lambda v_{i+1}$. From \eqref{eqn:b} and \eqref{eqn:ibi}, we immediately get
    \begin{equation}\label{eqn:from_b}
        (1-\eta)v_{k^*} > (1-\lambda)v_i + \lambda v_{i+1},
    \end{equation}
    \begin{equation}\label{eqn:from_ibi}
    (1-\lambda) i v_i + \lambda (i+1) v_{i+1} \ge k^*v_{k^*} - mv_{k^*}.
    \end{equation}
    Equation \eqref{eqn:from_ibi} further implies $\lambda (i+1) v_{i+1} \ge k^*v_{k^*} - (1-\lambda) k^*v_{k^*} - mv_{k^*}$. Divide by $\lambda$, 
    \[ (i+1) v_{i+1} \ge k^*v_{k^*} - \frac{m}{\lambda}v_{k^*}.\]
    Then it remains to lower-bound $\lambda$ by $\theta_t$. Intuitively, since $v_i$ is larger than $(1-
    \frac{\eta}{2})v_{k^*}$ but $b<(1-\eta)v_{k^*}$, the coefficient of $v_{i+1}$ cannot be too small. Formally, from \eqref{eqn:from_b} and \eqref{eqn:t_sandwich}, we have: 
    \[(1-\eta)v_{k^*} > (1-\lambda) (1-\frac{1}{2}\eta)v_{k^*} + \lambda (1-\eta_{t+1})v_{k^*}, \]
    \[\implies 1-\eta > 1 -\frac{1}{2}\eta - \lambda(\eta_{t+1} -\frac{1}{2}\eta) \]
    \[\implies  \lambda > \frac{\frac{1}{2}\eta}{\eta_{t+1} - \frac{1}{2}\eta} \ge \frac{\eta}{2\eta_{t+1}}=\theta_t, \]
    concluding the proof of this case.
\end{enumerate}

Finally we return to the analysis for $i=N$. If $\frac{k^*}{N}< 1-\frac{1}{2}\eta$, then
$v_N\le \frac{k^*v_{k^*}}{N} < (1-\frac{1}{2}\eta)v_{k^*}$, so the above argument (a) can be reused. Otherwise, from \eqref{eqn:b} and (\ref{eqn:ibi}), we have: 
\[(1-\eta)v_{k^*}>b\ge \frac{(k^*-m)v_{k^*}}{N} \ge (1-\frac{1}{2}\eta - \frac{m}{N}) v_{k^*}, \]
which contradicts the assumption that $\eta>\frac{2m}{dN}$.

%  ===================  NEW argument ==============

\subsection{Proof of Lemma \ref{lem:Pr_Bad}}
\label{sec:proof_lem:Pr_Bad}
%%%%%%%%%%%%%%%%%%%%%%%%%%%%%%%%%%%%%%%
For convenience we drop the subscript $t$ and just write $\first=\eta_t, \second=\theta_t$.
Recall that we need to upper-bound $\Pr[\Bad(\first, \second) \,\land\,\overline{\event}\,\land\,\conc]$ where: 
\begin{itemize}
    \item $\Bad(\first, \second)$ implies that there exists $j\ge k^*$ such that $v_j\le (1-\first)v_{k^*}$ and $jv_j\ge k^*v_{k^*}-\frac{m}{\second}v_{k^*}$.
    \item $\overline{\event}$ is $k^*\ge dN$. 
    \item $\conc$ requires that $|q_j-\frac{j}{N}|\le 2\sqrt{ \frac{4\ln(2(N-m))}{N-m} } + \frac{m}{N}$ for any $j>m$. 
\end{itemize}

Define
\[H= \frac{m}{d\second - \frac{m}{N}} \quad\quad\text{ and }\quad\quad h = \frac{1}{2}\left(d\first -  4\sqrt{ \frac{4\ln(2(N-m))}{N-m} }  - \frac{4m}{N\second}\right).\]
Assume $H, h>0$, which can be satisfied when $\first$ and $\second$ are at least $\Omega\left(\frac{m}{d}\sqrt{\frac{\log (N-m)}{N-m}}\right)$. 

\begin{claim}\label{claim:two_consequences}
The event $[\Bad(\first, \second) \,\land\,\overline{\event}\,\land\,\conc] $ implies that there exists $j\ge k^*$ which satisfies:
\begin{enumerate}
    \item $jv_j \le k^* v_{k^*} \le (j+H)v_j$; 
    \item $q_j - q_{k^*} \ge 2h$.
\end{enumerate}
\end{claim}

\begin{proof}[Proof of \cref{claim:two_consequences}]
Choose the $j$ in $\Bad(\first, \second)$ which satisfies $jv_j\ge k^*v_{k^*} - \frac{m}{\second}v_{k^*}$. To see why the first inequality holds, note that  $dNv_{k^*} \le k^*v_{k^*} \le jv_j + \frac{m}{\second}v_{k^*} \le Nv_j + \frac{m}{\second}v_{k^*}$, subtracting the first and forth term,  we get $(dN-\frac{m}{\second})v_{k^*}\le Nv_j$, further implying $k^*v_{k^*} \le jv_j + \frac{m}{\second}\frac{Nv_j}{(dN-m/\second)}$, which is the first inequality. 

Now consider the second inequality.  
Since $\Bad(\first, \second)$ requires $jv_j\ge k^*v_{k^*} - \frac{m}{\second}v_{k^*}$ and $v_j\le(1-\eta)v_{k^*}$, we have
\begin{align*}
    j \ge \frac{k^*v_{k^*} - \frac{m}{\second}v_{k^*}} {(1-\first)v_{k^*}} = \frac{k^*-\frac{m}{\second}}{1-\first} \ge (k^*-\frac{m}{\second})(1+\first) = k^*+ k^*\first - \frac{m}{\theta}(1+\first) \ge k^*+ dN\first - 2\frac{m}{\theta}, 
\end{align*}
and dividing by $N$, 
\[  \frac{j}{N} - \frac{k^*}{N} \ge d\first - 2\frac{m}{N\theta}.\] 
Using the condition $\conc$ on $j$ and $k^*$, we can derive the relationship between $q_j$ and $q_{k^*}$ by simple calculation: 
\begin{align*}
q_j - q_{k^*} \ge \left(d\first - 2\frac{m}{N\theta}\right) - 2\left( 2\sqrt{ \frac{4\ln(2(N-m))}{N-m} } + \frac{m}{N} \right) \ge 2h.
\end{align*}

\end{proof}

Divide the quantile space $[0, 1]$ into $[0, d/2]$ and $(1-d/2)/h$ equal-length intervals with length $h$, 
\begin{equation}\label{eqn:quantile_divide}
[0, 1] = [0, \frac{d}{2}]\cup I_1\cup I_2\cup \cdots\cup I_{\frac{1-d/2}{h}},
\end{equation}
where $I_l = (d/2 + (l-1)h, d/2 + lh]$. Thus a uniformly random draw of quantile falls into $I_l$ with probability $h$. Define $i^*_l$ and $i^*_{<(l+1)}$:
%be the optimal indices of quantiles in interval $I_l$ and intervals $I_1$ to $I_l$, i.e.
\[i^*_l = \argmax_{i> cN} \big\{iv_i\mid q_i\in I_l\big\} ~\text{ or }~ i^*_l=\emptyset \text{ if there is no such }i. \]
\[i^*_{<(l+1)} = \argmax_{i> cN} \big\{iv_i \mid q_i\in I_1\cup\cdots\cup I_l\big\} ~\text{ or }~ i^*_{<(l+1)}=\emptyset \text{ if there is no such }i. \]
And
%we denote the optimal value by $A_l, A_{<(l+1)}$: 
$A_l \defeq i^*_lv_{i^*_l},~~A_{<(l+1)} \defeq i^*_{<(l+1)} v_{i^*_{<(l+1)}}$. Moreover, define event $W_l$ for each $l$, 

%$$W_l=[A_{l+2} \le A_{<(l+1)} \le (i^*_{l+2} + H)v_{i^*_{l+2}}], $$
$$W_l=\left[\left\{i^*_{l+2}\ne \emptyset\right\} \ \land\   \left\{i^*_{<(l+1)}\ne \emptyset\right\} \ \land\  \left\{A_{l+2} \le A_{<(l+1)} \le A_{l+2} + H\widetilde{v}_{l+2}\right\}\right], $$
where $\widetilde{v}_{l+2}\defeq v(d/2+(l+1)h)$ is the upper bound on the values with quantiles in $I_{l+2}$. We argue that if the event $[\Bad(\first, \second)\,\land\, \overline{\event}\,\land\,\conc]$ holds then there must exist an index $l$ such that $W_l$ holds. To see this,  consider the index $j$ that is promised to exist in $\Bad(\first, \second)$ and choose the index $l$ such that $q_j\in I_{l+2}$. Note that $[\overline{\event}\,\land\,\conc]$ implies $q_j\ge q_{k^*}> d/2$, so both $q_j$ and $q_{k^*}$ must fall in $I_1\cup I_2\cup\cdots$. To see why $W_l$ must hold, note that:
\begin{itemize}
    \item $A_{l+2}\le A_{<(l+1)}$ since $q_j - q_{k^*}>2h$, implying $q_{k^*}\in I_{<(l+1)}$ and $A_{<(l+1)}=k^*v_{k^*}$. 
    \item $A_{<(l+1)} \le A_{l+2} + H\widetilde{v}_{l+2}$ since
$k^*v_{k^*}\le (j+H)v_j\le i^*_{l+2}v_{i^*_{l+2}} + H\widetilde{v}_{l+2}$.
\end{itemize}

Therefore, a union bound over $\Pr[W_l]$ suffices to prove that $\Pr[\Bad(\first, \second)\, \land \, \overline{\event} \, \land \, \conc]$ is small. The idea to bound $\Pr[W_l]$ is a refinement of \citet{yao2018incentive}: Note that there is an interval $I_{l+1}$ with length $h$ between $I_{l+2}$ and $I_{<(l+1)}$ and consider the number $X$ of quantiles falling into $I_{l+1}$. There is enough randomness in $X$ as its variance is $\Omega(h N)$, implying that the difference between the rankings of any pair of quantiles in $I_{l+2}$ and $I_{<(l+1)}$ varies broadly. As a result, it's unlikely that $A_{<(l+1)}$ will fall in the short interval $[A_{l+2}, A_{l+2}+H\widetilde{v}_{l+2}]$. Formally, we will prove that
\begin{lemma}\label{lem:Pr_W}
For any $l$, $\Pr[W_l] \le O(\frac{H\log^2 N}{\sqrt{h d^3 N}})$. 
\end{lemma}
\noindent The proof of Lemma~\ref{lem:Pr_W} is in Appendix \ref{sec:proof_lem:Pr_W}. 
To conclude, 
\begin{equation*}\label{eqn:prob_three}
   \Pr[\Bad(\first, \second) \,\land\,\overline{\event}\,\land\,\conc] \le \sum_{l=1}^{\frac{1-d/2}{h}} \Pr[W_l] \le \frac{1}{h}O\left(\frac{H\log^2 N}{\sqrt{h d^3 N}}\right) = O\left(\frac{m\log^2 N}{d\second\sqrt{(d\first)^3 d^3 N}}\right),
\end{equation*}
where the last equality is because $H=O(\frac{m}{d\second})$ and $h=\Omega(d\first)$ under the assumption that $\first$ and $\second$ are at least $\Omega(\frac{m}{d}\sqrt{\frac{\log(N-m)}{N-m}})$.

%  ===================  NEW argument ends ==============

\subsection{Proof of Lemma \ref{lem:Pr_W}}\label{sec:proof_lem:Pr_W}
We need to upper-bound $\Pr[W_l]$ over the random draw of $v_{-I}$, or in quantile space, $q_{-I}$, which are $N-m$ i.i.d.~random draws from $\mathrm{Uniform}[0, 1]$. Let $N_L$ be the number of quantile draws that are in $L\defeq[0, d/2]\cup I_1\cup \cdots \cup I_{l+1}$. Suppose we draw the quantiles in the following procedure: first determine $N_L$, then draw $N-m-N_L$ quantiles that are not in $L$; finally draw $N_L$ quantiles that are in $L$. 

Note that $N_L$ follows a binomial distribution, and a Chernoff bound implies that
\begin{equation}\label{eq:NL-chernoff}
\Pr[N_L\ge \frac{d}{4} (N-m)] \geq 1 - \exp\left(-\frac{d(N-m)}{16}\right).
\end{equation}
We thus assume $N_L\ge d(N-m)/4$.
%and have $N_L\ge \frac{c}{4} (N-m)$ with probability at least $1 - \exp(-\frac{c(N-m)}{16})$. 

Then draw $N-m-N_L$ quantiles from $[0, 1]\backslash L$, so $i^*_{l+2}, v_{i^*_{l+2}}$ and $A_{l+2}$ are determined. Suppose $i^*_{l+2}\ne\emptyset$; otherwise, $W_l$ does not hold.

Now we draw $N_L$ quantiles, $q^{(1)}, \ldots, q^{(N_L)}$ from $L$. Consider the increment of $A_{<(l+1)}$, as a sequence $A^{(t)}, t=1, \ldots, N_L$. After the time $t-1$ when $A^{(t-1)} \ge A_{l+2}$, the index $i^*_{<(l+1)}$ is no longer $\emptyset$. When one more sample $q^{(t)}$ is generated, there are three cases: 
\begin{enumerate}
    \item If $q^{(t)}\in[0, d/2]$, $A^{(t-1)}$ increases by at least $\widetilde{v}_{{l+2}}$. This is because each term $iv_i^{(t-1)}$ increases to $(i+1)v_i^{(t)}\ge iv_i^{(t-1)}+\widetilde{v}_{{l+2}}$, for any $i$ such that $q_i\in I_1, \ldots, I_l$.
    \item If $q^{(t)}\in I_1\cup \cdots \cup I_l$, then $A^{(t-1)}$ does not decrease.
    \item If $q^{(t)} \in I_{l+1}$, $A^{(t-1)}$ does not change.
\end{enumerate}   Let $s$ be the number of quantiles that are not in $I_{l+1}$, and $A^{(t_1)}, \ldots, A^{(t_s)}$ be those steps,  and write $B^{(i)} \defeq (A^{(t_i)}-A_{l+2})/\widetilde{v}_{l+2}$ for $i=1, \ldots, s$. We have $A_{<(l+1)} = \widetilde{v}_{l+2}B^{(s)} + A_{l+2}$. Then our task is to analyze the probability that $B^{(s)}\in[0, H]$. 

We can think of the generation of $B^{(s)}$ as follows: regardless of $s$, first generate an infinite sequence $B^{(1)}, B^{(2)}, \ldots$, where at each step $i$ the value $B^{(i)}$ is increased by $1$ with probability at least $\Pr[q\in [0, d/2] \mid q\in L] \ge d/2$. Then pick an index $s$ by a binomial distribution $Bin(N_L, 1 - \Pr[q\in I_{l+1} \mid q\in L])$. Then the $s$-th value in the infinite sequence $\{B^{(i)}\}$ is chosen as $B^{(s)}$. Note that $Bin(N_L, 1 - \Pr[q\in I_{l+1} \mid q\in L])$ is dominated by $Bin(N_L, 1-h)$, so the probability that $B^{(s)}$ takes on any one of the values in the sequence $\{B^{(i)}\}$ is at most $\Pr[s=i] = O(1/\sqrt{h N_L})$. 

Then we consider the length of the sub-sequence where $B^{(i)} \in [0, H]$. Intuitively, the expected number of steps for $B^{(i)}$ to increase by $H$, is at most $H/(d/2)$. The probability that it takes more than $2H(\log N_L)^2/d$ steps implies that the sum of $2H(\log N_L)^2/d$ i.i.d.~Bernoulli variables whose success probability is at least $d/2$ does not reach $H$, which can be bounded by a Chernoff bound: 
\begin{align*}
\Pr[\text{Length}> \frac{2H(\log N_L)^2}{d}] \le{}& \Pr[Bin\left(\frac{2H(\log N_L)^2}{d}, \frac{d}{2}\right) < H] \\
\le{}& \exp\left(-\frac{1}{2}H(\log N_L)^2\left(1-\frac{1}{(\log N_L)^2}\right)^2\right) \\
={}& O\left(\exp\left(-\frac{1}{8}(\log N_L)^2\right)\right) \\ 
={}& O\left(N_L^{-\frac{1}{8}\log N_L}\right). 
\end{align*}
Assuming $\text{Length}\le (2H(\log N_L)^2)/d$, the probability that $B^{(s)} \in [0, H]$ can be bounded by a union bound:  
\[ \sum_{i: B^{(i)}\in[0, H]} \Pr[s=i] \le \text{Length}\cdot O\left(\frac{1}{\sqrt{h N_L}}\right) \le O\left(\frac{H(\log N_L)^2}{d\sqrt{h N_L}}\right). \]

Therefore,  
\begin{align*}
    \Pr[W_l]  \le{}& O\left(\frac{H(\log N_L)^2}{d\sqrt{h N_L}}\right) + O\left(N_L^{-\frac{1}{8}\log N_L}\right) \\
    \le{}& O\left(\frac{H(\log N)^2}{d\sqrt{h d N}}\right)  + \left(\frac{d}{4}(N-m)\right)^{-\frac{1}{8}\log \frac{d}{4}(N-m)} + \exp\left(-\frac{d(N-m)}{16}\right) && \text{By \eqref{eq:NL-chernoff}}\\
    ={}& O\left(\frac{H(\log N)^2}{\sqrt{h d^3 N}}\right).
\end{align*}

\subsection{Proof of Claim~\ref{cl:eta-calculation}}
\label{proof-of-eta-calculation}

We need to show that there exist an integer $M$ and parameters $\eta_0=\frac{1}{2}\eta<\eta_1=\eta< \eta_2=2\eta<\eta_3<\cdots<\eta_M<\eta_{M+1}=1$, such that 

\begin{align*}
    \sum_{t=0}^M \frac{\eta_{t+1}}{\eta}\frac{1}{\eta_{t}^{3/2}} = O\left(\frac{\log\log (N-m)}{\eta^{3/2}}\right).
\end{align*}

We start with:

\begin{equation}\label{eqn:eta-summation-1}
    \sum_{t=0}^M  \frac{\eta_{t+1}}{\eta}\frac{1}{\eta_{t}^{3/2}} = \frac{1}{\eta^{3/2}}\sum_{t=0}^M  \frac{\eta_{t+1}/\eta}{ (\eta_{t}/\eta)^{3/2}} = \frac{1}{\eta^{3/2}}\left(O(1) + \sum_{t=2}^M  \frac{\eta_{t+1}/\eta}{ (\eta_{t}/\eta)^{3/2}}\right)
\end{equation}
Let $\eta_{t+1}/\eta = (\eta_t/\eta)^{3/2}$ for any $t\ge 2$. We can recursively compute $\eta_t$ until the maximum step $t=M$ which satisfies $\eta_{M}<1$. Then \eqref{eqn:eta-summation-1} is upper-bounded by $\frac{1}{\eta^{3/2}}(O(1)+M)$. By our construction of $\{\eta_t\}$, we have
\[\frac{\eta_M}{\eta} = (\frac{\eta_2}{\eta})^{\frac{3}{2}^{M-2}} = 2^{\frac{3}{2}^{M-2}}<\frac{1}{\eta}.\]
Thus,
\[M < \log_{3/2} \log_2 \frac{1}{\eta} +2 = O(\log\log \frac{1}{\eta}) = O(\log \log (N-m)), \]
where the last equality follows from the assumption that $\eta =\Omega\left(\frac{m}{d}\sqrt{\frac{\log (N-m)}{N-m}}\right)$. Thus, the summation $\eqref{eqn:eta-summation-1}$ becomes 
\begin{align*}
    \sum_{t=0}^M \frac{\eta_{t+1}}{\eta}\frac{1}{\eta_{t}^{3/2}} = O\left(\frac{\log\log (N-m)}{\eta^{3/2}}\right).
\end{align*}
as required.

\section{Missing Proofs From Section~\ref{sec:two_phase}}
\subsection{Proof of \cref{thm:connection_TPM_PRM_combined} (for Bounded Distributions)}
%\begin{proof} [Proof of \cref{thm:connection_TPM_PRM_combined} for bounded distributions]
\label{sec:proof_of_lemma_1}
\begin{proof}[Proof of \cref{thm:connection_TPM_PRM_combined} for bounded Distributions]
First consider the reduction of reserve price caused by the deviation of bidder $i$. The true values of all bidders in the first phase are $(v_I, v_{-I})$, where bidder $i$'s true values are $v_{I}\in\mathbb{R}_+^{m_{i, 1}}$. When other bidders bid $v_{-I}$ truthfully, and bidder $i$ bids $b_I$ instead, the reserve price $p$ changes from $P(v_I, v_{-I})$ to $P(b_I, v_{-I})$ and the change is at most 
$$P(v_I, v_{-I}) - P(b_I, v_{-I}) \le P(v_I, v_{-I})\delta_I(v_I, v_{-I})\le P(v_I, v_{-I})\delta_{m_{i, 1}}^{\interim}(v_{-I}) \le D\delta_{m_1}^{\interim}(v_{-I}),$$ by Definition \ref{def:discount} and by the fact that all values are upper-bounded by $D$. Consider the increase of utility in the second phase. We claim that for any two possible reserve prices $p_2 \le p_1$, for any $v\in\mathbb R_+$, for any $K_2\ge 1$, we have
\begin{equation}\label{eqn:interim-utility-two-price}
    \usecond^{K_2}(v, p_2)-\usecond^{K_2}(v, p_1) \le p_1-p_2.
\end{equation}
To see this, first re-write $\usecond^{K_2}(v, p)$ in \eqref{eqn:def_u_second} as
$$ \usecond^{K_2}(v, p) = \E_{X_2,\ldots,X_{K_2} \sim F} \big[ (v - \max\{p, X^*\} )^+\big],$$
where $X^*\defeq \max\{X_2, \ldots, X_{K_2}\}$ and $(x)^+ \defeq \max\{x, 0\}$. Note that $(x)^+ - (y)^+ \le |x-y|$, thus
\begin{multline*}
\usecond^{K_2}(v, p_2)-\usecond^{K_2}(v, p_1) = \E \left[ (v - \max\{p_2, X^*\} )^+ -  (v - \max\{p_1, X^*\} )^+\right] \\
\le \E \left[ \left| \max\{p_1, X^*\} - \max\{p_2, X^*\} \right| \right] \le \E \left[ \left| p_1 - p_2 \right| \right] = p_1 - p_2.
\end{multline*}

%\begin{equation}\label{eqn:def_u_second}
%%	\usecond^{K_2}(v, p)=\E_{X_2, \ldots, X_{K_2}\sim F}\bigg[\big(v - \max\{p, X^*\}\big)\cdot\mathbb{I}\big[v>\max\{p, X^*\}\big]\bigg].
%	\usecond^{K_2}(v, p)=\E_{X_2, \ldots, X_{K_2}\sim F}\bigg[\big(v - \max\{p, X_2, \ldots, X_{K_2}\}\big)\cdot\mathbb{I}\big[v>\max\{p, X_2, \ldots, X_{K_2}\}\big]\bigg].
%\end{equation}
%\noindent

For the first phase we have $U_i^{\mathcal{M}}(\bm{v}_i, b_I, v_{-I}) \le U_i^{\mathcal{M}}(\bm{v}_i, v_I, v_{-I})$ since $\mathcal{M}$ is truthful. Thus, by \eqref{eqn:def_TP_utility} and \eqref{eqn:interim-utility-two-price} the difference in interim utilities is at most
\begin{align*}
     \E_{\bm{v}_{-i}} & \left[ U^{\TP}_i(\bm{v}_i, b_I, v_{-I})  -  U^{\TP}_i(\bm{v}_i, v_I, v_{-I}) \right] \\
     \le{} & \E_{\bm{v}_{-i}}\left[U_i^{\mathcal{M}}(\bm{v}_i, b_I, v_{-I}) - U_i^{\mathcal{M}}(\bm{v}_i, v_I, v_{-I})\right] \\
     & + \E_{\bm{v}_{-i}} \left[ \sum_{t=m_{i, 1}+1}^{m_{i, 1} + m_{i, 2}} \bigg[\usecond^{K_2}\left(v_{i, t}, P(b_I, v_{-I})\right) - \usecond^{K_2}\left(v_{i, t}, P(v_I, v_{-I})\right)\bigg] \right]\\
    \le{} & 0 + \E_{\bm{v}_{-i}}\left[ m_2(P(v_I, v_{-I}) - P(b_I, v_{-I})) \right] \le  \E_{\bm{v}_{-i}}\left[ m_2D \delta^{\interim}_{m_1}(v_{-I})\right] = m_2D\Delta_{T_1K_1, m_1}^{\interim}, 
\end{align*}
which indicates that truthful bidding is an $\epsilon$-BNE, where $\epsilon=m_2D \Delta^{\interim}_{T_1K_1, m_1}$. This concludes the proof for bounded distributions.
\end{proof}

{\em Remark.}
The proof for MHR distributions is trickier since the difference
$P(v_I, v_{-I}) - P(b_I, v_{-I})$ can be unbounded. Intuitively, the probability that $P(v_I, v_{-I})$ will be higher than $(1+o(1))v^*$ ($v^*$ is defined in the statement of the lemma) is exponentially small, and the main effort is to show that the expected difference
$P(v_I, v_{-I}) - P(b_I, v_{-I})$
multiplied by this exponentially small probability is negligible. Full details are given \cref{sec:proof_lemma_1_mhr}.

% \section{Missing Proofs From Section~\ref{sec:two_phase}}

\subsection{Proof of Theorem \ref{thm:TPM_utility_revenue_combined}}\label{sec:proof_thm:TPM_utility_revenue_combined}
The bound on approximate truthfulness, i.e., $\epsilon_1$, follows from \cref{thm:connection_TPM_PRM_combined} and \cref{thm:discount_combined}, where we first obtain the bound on $\Delta^\worst_{T_1K_1, m_1}$ from \cref{thm:discount_combined} by setting $N=T_1K_1$ and $m=m_1$ and then replace $m_2m_1$ with $O(\msum^{2})$ for MHR distribution and replacing $m_2m_1^{2/3}$ with $O(\msum^{5/3})$ for bounded distribution. 

It remains to consider revenue, where we will use sample complexity results to obtain the convergence rate of the revenue loss, i.e., $\epsilon_2$.
% The bound on approximate truthfulness, i.e., $\epsilon_1$, follows from \cref{thm:connection_TPM_PRM_combined} by replacing $m_2m_1$ with $O(\msum^{2})$ for MHR distribution and replacing $m_2m_1^{2/3}$ with $O(\msum^{5/3})$ for bounded distribution (first let $m=m_1$ and $N=T_1K_1$ in \cref{thm:discount_combined}).  It remains to consider revenue, where we use sample complexity results to obtain the convergence rate of the revenue loss, i.e., $\epsilon_2$.
Let $rev_1, rev_2$ be the expected revenues of the two phases in $\TP(\mathcal{M}, P; \TPT, \TPm, \TPK, S)$, and $rev^*$ be the revenue obtained by using Myerson's auction in all rounds, i.e., $rev^*=T_1\mye^{K_1} + T_2\mye^{K_2}$ where $\mye^{K}$ is the revenue of Myerson's auction with $K$ i.i.d.~bidders from $F$. For $rev_1$, we only have $rev_1\ge0$ since we do not any revenue guarantee for the arbitrary first-phase mechanism $\mathcal{M}$. Now consider $rev_2$, let $r^{K_2}(p)$ denote the expected revenue of a second price auction with reserve price $p$. Since the values in the two phases are independent, we have

\[rev_2 = T_2\cdot \E_{v_1, \ldots, v_{T_1K_1}\sim F} \left[  r^{K_2}(\erm^c(v_1, \ldots, v_{T_1K_1})) \right]. \]

\noindent
We need to compare $r^{K_2}(\erm^c(v_1, \ldots, v_{T_1K_1}))$ with $\mye^{K_2}$. Since bidders have i.i.d.~regular value distributions, Myerson's auction is exactly the second price auction with reserve price $p=v^*$. When $K_2=1$, Myerson's auction becomes a post-price auction. Let $\epsilon^{sample}(\cdot)$ be the inverse function of the required number of samples for $\erm^c$ to guarantee $(1-\epsilon^{sample})$-optimal revenue (as obtained in \citet{huang2018making}) in the posted-price auction, i.e., the expected revenue of a one-bidder auction with a posted price $p$ determined by $\erm^c$ with $N$ samples is at least $(1-\epsilon^{sample}(N))$ times the optimal expected revenue. Then for the one-bidder case, we have

\[ rev_2 = T_2 (1-\epsilon^{sample}(T_1K_1)) \mye^1.\]

\noindent
For general $K_2$, while the sample complexity literature does not analyze the revenue of the same reserve price $p=\erm^c(v_1, \ldots, v_{T_1K_1})$ in a second price auction with $K_2 \geq 2$ bidders, we are able to generalize the existing revenue guarantee to the case of multiple bidders (and multiple units) under the assumption that the distribution is regular. The generalization is made by the following lemma, which we believe is of independent interest:
\begin{lemma}\label{lem:regular_one_to_multiple}
%For any regular distribution $F$, if the expected revenue of a second price auction with reserve price $p$ and with one bidder whose value is drawn from $F$ is $(1-\epsilon)$-optimal, then the revenue of a second price auction with reserve price $p$ and with $K\ge 2$ i.i.d.~bidders from $F$ is also $(1-\epsilon)$-optimal. 
For any regular distribution $F$, if the expected revenue of a posted price auction with price $p$ and with one bidder whose value is drawn from $F$ is $(1-\epsilon)$-optimal, then the revenue of a Vickrey auction with reserve price $p$ selling at most $k\ge 1$ units of a item to $K\ge 2$ i.i.d.~unit-demand bidders with values from $F$ is also $(1-\epsilon)$-optimal. 
\end{lemma}
\noindent The proof is in \cref{sec:proof_lem:regular_one_to_multiple}. Thus for $K_2\ge 2$, we also have: $rev_2 = T_2 (1-\epsilon^{sample}(T_1K_1)) \mye^{K_2}.$

Finally, 
\begin{align*}
1 - \epsilon_2 ={}& \frac{rev_1 + rev_2}{rev^*}
\ge{} \frac{0 + T_2\mye^{K_2}\cdot (1 - \epsilon^{sample}(T_1K_1)) }{T_1 \mye^{K_1} + T_2 \mye^{K_2}} \\
\ge{}& 1 - \frac{T_1\mye^{K_1}}{T_1 \mye^{K_1} + T_2 \mye^{K_2}} - \epsilon^{sample}(T_1K_1) \\
 \ge {}& 1 - \frac{T_1}{T} - \epsilon^{sample}(T_1K_1) && \text{ ($\mye^{K_2}\ge \mye^{K_1}$ since $K_2 \ge K_1$)}
\end{align*}

%\noindent
%Let $m_2=\omega(\frac{m_1K_2}{K_1})$, we have $\epsilon_2 = O(\frac{m_1K_2}{m_2K_1} + \epsilon^{sample}(nm_1))\to 0$.
%thus as $n\rightarrow \infty$ the approximated optimal revenue is obtained.

%For bounded and regular distributions, according to \cref{cor:TPM_utility_combined}, it suffices to choose

%$$m_1 = o(n),\  m_2 = o\left(\sqrt[3]{\frac{n}{m_1\log^6(nm_1)}}\right)$$
%to guarantee the incentive compatibility when $n\rightarrow \infty$.
%According to \citet{huang2018making} %as well as Lemma~\ref{lem:regular_one_to_multiple},
From \citet{huang2018making}, we know that for bounded distributions, $\epsilon^{sample}(N)=O(\sqrt{\frac{D\cdot\log N}{N}})$ when $c\le\frac{1}{2D}$,
%so $\epsilon_2 = O(\frac{T_1}{T}+\sqrt{\frac{D\cdot\log (T_1K_1)}{T_1K_1}})$,
and for MHR distributions (MHR implies regularity),
% we choose
%$$m_1=o(n),\  m_2 = o\left(\sqrt[3]{\frac{n}{m_1\log^6(nm_1)}}\right).$$ 
%the convergence rate of $\epsilon_2$ for MHR distributions follows from
$\epsilon^{sample}(N)=O([\frac{\log  N}{N}]^{\frac{2}{3}})$ when $c\le\cupperboundmhr$. This implies the bounds on $\epsilon_2$ as stated in the theorem, and concludes the proof.
%\citet{huang2018making}.

%Note that Corollary \ref{thm:TPM_utility_revenue_combined} is valid for the following extension to multiple units case. 

\subsection{Proof of Lemma~\ref{lem:regular_one_to_multiple}}\label{sec:proof_lem:regular_one_to_multiple}
It's more convenient to work in the quantile space. Let $v(q)$, $R(q)=qv(q)$ be the value curve and revenue curve of $F$. It's well-known that the derivative $R'(q)$ equals to the virtual value $\phi(v(q)) = v - \frac{1-F(v)}{f(v)}$ and by Myerson's Lemma, the expected revenue with allocation rule $x(\cdot)$ equals to the virtual surplus: \[rev = \sum_{i=1}^{K}\E[R'(q)x_{i}(q)].\] 

Let $x^{K, k}(q)$ be the allocation to (the probability of winning of) a bidder whose value has quantile $q$ in a Vickrey auction selling $k$ units to $K$ bidders without reserve price. Specifically, $x^{1, 1}(q)=1$; for general $K$, $x^{K, 1}(q)=(1-q)^{K-1}$; for general $K, k$, $x^{K, k}(q) = \sum_{i=0}^{k-1}{K-1\choose i}q^i(1-q)^{K-1-i}$. With reserve price $p_0$, let $q_0=q(p_0)$, then the allocation becomes $x^{K, k}_{q_0}(q) = x^{K, k}(q)$ for $q< q_0$ and $x^{K, k}_{q_0}(q)=0$ otherwise. So the revenue of $p_0$ is
\[rev(K, k) = K\int_0^{q_0} R'(q)x^{K, k}(q) \dd q,\]
and the optimal revenue is: 
\[rev^*(K, k) = K\int_0^{q^*} R'(q^*)x^{K, k}(q)\dd q, \]
where $q^*$ satisfies: $R'(q)\ge 0, \forall q\le q^*$ and $R'(q)\ge0, \forall q\ge q^*$. And define: 
\[loss(K, k) = rev^*(K, k) - rev(K, k) = K\int_{q_0}^{q^*} R'(q)x^{K, k}(q)\dd q. \]

Since $p_0$ is $(1-\epsilon)$-optimal with $K=1, k=1$ (the posted-price auction), we have: 
\[loss(1, 1) \le \epsilon\cdot rev^*(1, 1).\]
Now for general $K, k$: 
\begin{itemize}
\item If $q_0>q^*$. The loss:  
\begin{align*}
loss(K, k) ={}& K\int_{q^*}^{q_0} -R'(q)x^{K, k}(q)\dd q \\
\le{}& K\int_{q^*}^{q_0} -R'(q)x^{K, k}(q^*)\dd q \\
={}& Kx^{K, k}(q^*)\int_{q^*}^{q_0} -R'(q)\dd q = Kx^{K, k}(q^*)loss(1, 1),
\end{align*}
since $x^{K, k}(q)$ is non-increasing in $q$ (actually, the monotonicity of $x^{K, k}(q)$ is the only property that is used throughout the proof), and the optimal revenue: 
\begin{align*} rev^*(K, k) \ge{}& K\int_0^{q^*}R'(q)x^{K, k}(q^*)\dd q \\
={}& Kx^{K, k}(q^*)\int_0^{q^*}R'(q)\dd q = Kx^{K, k}(q^*)rev^*(1, 1),\end{align*}
which gives: $\frac{loss(K, k)}{rev^*(K, k)} \le \frac{loss(1, 1)}{rev^*(1, 1)} \le \epsilon$.
\item If $q_0<q^*$. The loss: 
\[loss(K, k) = K\int_{q_0}^{q^*}R'(q)x^{K, k}(q)\dd q \le Kx^{K, k}(q_0)loss(1, 1), \]
and the optimal revenue: 
\begin{align*}
rev^*(K, k) ={}& K\int_{0}^{q_0}R'(q)x^{K, k}(q)\dd q + K\int_{q_0}^{q^*}R'(q)x^{K, k}(q)\dd q \\
\ge{}& K\int_{0}^{q_0}R'(q)x^{K, k}(q_0)\dd q + K\int_{q_0}^{q^*}R'(q)x^{K, k}(q)\dd q \\ 
={}& Kx^{K, k}(q_0)rev(1, 1) + loss(K, k), 
\end{align*}
which gives:
\begin{align*}
\frac{loss(K, k)}{rev^*(K, k)} \le{}& \frac{loss(K, k)}{Kx^{K, k}(q_0)rev(1, 1) + loss(K, k)} = \frac{1}{\frac{Kx^{K, k}(q_0) rev(1, 1)}{loss(K, k)} + 1} \\
\le{}& \frac{1}{\frac{rev(1, 1)}{loss(1, 1)} + 1} \le \frac{1}{\frac{1-\epsilon}{\epsilon} + 1} =  \epsilon. 
\end{align*}
\end{itemize}

\subsection{Perfect Bayesian Equilibrium}
\label{sec:pbe}
Here we consider the setting of $T$-round repeated auctions where each auction contains $K\ge 1$ bidders and each bidder participates in at most $\overline{m}$ rounds of auctions.  We use $\bm{v}_i=(v_i^t)$ to denote bidder~$i$'s profile of values, where $v_i^t$ is her value at round $t$ if she participates in that round.  Similarly denote by $\bm{b}_i = (b_i^t)$ the bids of bidder~$i$.  Values are i.i.d.~samples from some distribution $F$. 

In repeated auctions, the seller can adjust the mechanism dynamically based on the bidding history of buyers, and buyers can use historical information to adjust their bidding strategies.  The solution concept of an $\epsilon$-perfect Bayesian equilibrium ($\epsilon$-PBE) captures this dynamic nature.  For each bidder~$i$, we use $\hist_i^t$ to denote the history she can observe at the start of round $t$.  For example, $\hist_i^t$ includes her bid $b_i^{t'}$, whether she receives the item, how much she pays, etc, at round $t'<t$ if she participates in round $t'$.  We assume that bidder~$i$ cannot observe the bids in the auctions she does not participate in.  We allow bidder~$i$ to anticipate her values in future rounds, so she can make decision on her entire value profile $\bm{v}_i=(v_i^t)$.  Bidder~$i$'s strategy is thus denoted by $\strat_i=(\strat_i^t)$ where $\strat_i^t$ maps $\bm{v}_i$ and $\hist_i^t$ to a bid $b_i^t = \strat_i^t(\bm{v}_i, \hist_i^t)$.  Let $\pbeu_i^{[t: T]}(\strat; \bm{v}_i, \hist_i^t)$ be the total expected utility of bidder $i$ in rounds $t, t+1, \ldots, T$, given her value profile $\bm{v}_i$, the history $\hist_i^t$ at round $t$, and bidders playing $\strat=(\strat_i, \strat_{-i})$. 

\begin{definition}
A profile of strategy $\strat=(\strat_i, \strat_{-i})$ is an \emph{$\epsilon$-perfect Bayesian equilibrium ($\epsilon$-PBE)} if for each bidder $i$, each round $t$, any history $\hist_i^t$, any values $\bm{v}_i$, the strategy $\strat_i$ approximately maximizes bidder $i$'s expected utility from round $t$ to round $T$ up to $\epsilon$ error, i.e., $\pbeu_i^{[t: T]}(\strat; \bm{v}_i, \hist_i^t) \ge \pbeu_i^{[t: T]}(\strat_i', \strat_{-i}; \bm{v}_i, \hist_i^t) - \epsilon$ for any alternative strategy $\strat_i'$. 
\end{definition}
\begin{definition}
The seller's mechanism (or auction learning algorithm) is \emph{$\epsilon$-perfect Bayesian incentive-compatible ($\epsilon$-PBIC)} if truthful bidding (i.e., $\strat_i^t(\bm{v}_i, \hist_i^t) = v_i^t$) is an $\epsilon$-PBE. 
\end{definition}
We emphasize that the expected utility $\pbeu_i^{[t: T]}(\strat; \bm{v}_i, \hist_i^t)$ is conditioned on $\hist_i^t$.  This is because, the history $\hist_i^t$ which includes the allocation of item and the payment of bidder $i$ can leak information about other bidders' bids (or values).  Other bidders' bids will influence the mechanism the seller will use in future rounds.  Thus, based on this information, bidder $i$ can update her belief about other bidders' bids and the seller's choice of mechanisms by Bayesian rule, then she can compute her expected utility on her updated belief. 

\paragraph{PBIC of the two-phase ERM algorithm.}
As discussed, the two-phase ERM algorithm is a learning algorithm that learns approximately revenue-optimal auctions in an approximately incentive-compatible way against strategic bidders in repeated auctions. 
The algorithm, obtained by adopting the two-phase model with ERM as the price learning function and setting $K_1=K_2=K$, works as follows: 
\begin{itemize}
    \item in the first $T_1$ rounds, run any truthful, prior-independent auction $\mathcal{M}$, e.g., the second price auction with no reserve;
    \item in the later $T_2=T-T_1$ rounds, run second price auction with reserve $p=\erm^c(b_{1}, \ldots, b_{T_1K})$ where $b_1, \ldots, b_{T_1K}$ are the bids from the first $T_1$ auctions. 
\end{itemize}
$T_1$ and $c$ are adjustable parameters of the two-phase ERM algorithm. 
 
\begin{theorem}\label{thm:two-phase-PBIC}
The two-phase ERM algorithm is $\epsilon$-PBIC, where,

\begin{itemize}

\item for any bounded $F$, 
$\epsilon=\msum D \Delta^{\interim}_{T_1K, \msum K}$, and
 
\item for any MHR $F$, if $\frac{\msum}{T_1}\le c\le \cupperboundmhr$ and $\msum K=o(\sqrt{T_1K})$, then
$\epsilon=O\left(\msum v^* \Delta^{\interim}_{T_1K, \msum K}\right)+O\left(\frac{\msum v^*}{\sqrt{T_1K}}\right)$,
where $v^*=\argmax_v\{v[1-F(v)]\}$.
\end{itemize}

The constants in big $O$'s are independent of $F$ and $c$. 
\end{theorem}

Combining with the bounds on $\Delta^\worst_{N, m}$ in \cref{thm:discount_combined}, we have
\begin{corollary}\label{cor:two-phase-PBIC-revenue}
The two-phase ERM algorithm is $\epsilon_1$-PBIC, where,

\begin{itemize}

\item for any bounded $F$, 
\[\epsilon_1=O\left(\log^2(T_1K) \sqrt[3]{\frac{D^{11} K \msum^5}{T_1}} \right)\]
if $\frac{\msum}{T_1} \le c \le \frac{1}{2D}$ and $\msum K = o(\sqrt{T_1 K})$; 
 
\item for any MHR $F$, 
\[\epsilon_1=O\left(\log^3(T_1K) v^* \msum^2 \sqrt{\frac{K}{T_1}}, \right)\]
if $\frac{\msum}{T_1}\le c\le \cupperboundmhr$ and $\msum K=o(\sqrt{T_1K})$.
\end{itemize}
The constants in big $O$'s are independent of $F$ and $c$. 

The guarantee of $(1-\epsilon_2)$ revenue optimality is the same as \cref{thm:TPM_utility_revenue_combined}: 
\begin{itemize}
    \item For bounded and regular distribution, $\epsilon_2 = O\left(\frac{T_1}{T}  + \sqrt{\frac{D\cdot\log (T_1K)}{T_1K}}\right)$.
    \item For MHR distribution, 
    $\epsilon_2=O\left( \frac{T_1}{T} + \left[\frac{\log (T_1 K)}{T_1 K}\right]^{\frac{2}{3}} \right)$. 
\end{itemize}
\end{corollary}

In the rest of this section we prove \cref{thm:two-phase-PBIC} for bounded distributions.  The proof is similar to that of \cref{thm:connection_TPM_PRM_combined} except that we need to consider the effect of history $\hist_i^t$ on the conditional distribution of the values of other bidders when considering perfect Bayesian equilibrium.  The extension to MHR distributions is similar to the extension of \cref{thm:connection_TPM_PRM_combined} to MHR distributions (discussed in \cref{sec:proof_lemma_1_mhr}) and hence omitted. 

\begin{proof}[Proof of \cref{thm:two-phase-PBIC} for bounded distributions]
Assume that other bidders bid truthfully and bidder~$i$ deviates from truthful bidding to other strategy, consider the increase of bidder~$i$'s total expected utility from round $t$ to $T$, for each $t$, given any history $\hist_i^t$ and any values $\bm{v}_i$. 
If $t>T_1$, then the auctions in $t$ and later rounds are in the second phase and never change due to the deviation of bidder~$i$, thus deviation does not increase her utility. 

Then we consider $t\le T_1$.  Strategic bidding does not increase bidder $i$'s utility in rounds $t, t+1, \ldots, T_1$ because these rounds are in the first phase and the mechanism in the first phase is truthful.  Thus, strategic bidding can increase her utility only in the second phase.  Let $t'>T_1$ be a second-phase round in which she participates.  The auction at round $t'$ is a second-price auction with reserve price determined by $\erm^c$ from bids in rounds $1$ to $T_1$.  Denote by $\bm{v}^{[1: T_1]}$ the values of all bidders in rounds $1$ to $T_1$.  If bidder~$i$ bid truthfully, then the reserve price at round~$t'$ is $p_1=\erm^c(\bm{v}^{[1: T_1]})$.  Let $A\subseteq[1: T_1]$ be the set of rounds in which bidder $i$ participates from round~$1$ to round~$T_1$.  Then we can partition $\bm{v}^{[1: T_1]}$ into two parts: $\bm{v}^A$ and $\bm{v}^{[1: T_1]\backslash A}$, where $\bm{v}^A$ denotes bidders' values in the rounds in $A$, and $\bm{v}^{[1: T_1]\backslash A}$ denotes the values in the rounds not in $A$.  %Clearly, bidder~$i$'s values in the first $T_1$ rounds are contained in $\bm{v}^A$. 
There are $|A|K$ values in $\bm{v}^A$, $|A|$ of which are bidder~$i$'s values. 
By deviating, bidder~$i$ can change her values in $\bm{v}^A$ to some arbitrary bids.  We denote by $\bm{b}^A$ the bids of bidder~$i$ and the values of other bidders in $\bm{v}^A$.  After deviation, the reserve price is changed to $p_2 = \erm^c(\bm{b}^A, \bm{v}^{[1: T_1]\backslash A})$.  By \eqref{eqn:interim-utility-two-price}, the increase of bidder~$i$'s utility due to the change of reserve price is at most $p_1-p_2$, which is further upper-bounded by 
\begin{align*}
    p_1 - p_2 & = \erm^c(\bm{v}^A, \bm{v}^{[1: T_1]\backslash A}) - \erm^c(\bm{b}^A, \bm{v}^{[1: T_1]\backslash A}) \\
    & \le \erm^c(\bm{v}^A, \bm{v}^{[1: T_1]\backslash A})\cdot \delta_I(\bm{v}^A, \bm{v}^{[1: T_1]\backslash A}) \\
    & \le \erm^c(\bm{v}^A, \bm{v}^{[1: T_1]\backslash A})\cdot \delta^{\worst}_{|A|K}(\bm{v}^{[1: T_1]\backslash A}) \\
    & \le D\cdot \delta^{\worst}_{|A|K}(\bm{v}^{[1: T_1]\backslash A}).
\end{align*}
We then argue that given any history $\hist_i^t$, $\bm{v}^{[1: T_1]\backslash A}$ are still i.i.d.~samples from $F$, from bidder~$i$'s perspective.  
Note that bidder~$i$ does not participate in the auctions in rounds $[1: T_1]\backslash A$, and the auctions she does participate in before round $t$ is prior-independent, which implies that the allocation of item and the payments of bidders in any round depend only on the bids of bidders in that round but not on any information like bids from other rounds.  Moreover, other bidders' values across different rounds are independent.  Therefore, the auctions bidder~$i$ participates in leaks no information about other bidders' values in rounds $[1: T_1]\backslash A$.  

Therefore, the increase of bidder~$i$'s expected utility at round $t'$ is at most 
\begin{align*}
    \E \left[p_1 - p_2 \mid \hist_i^t, \bm{v}_i\right]
    & \le \E \left[D\cdot \delta^{\worst}_{|A|K}(\bm{v}^{[1: T_1]\backslash A}) \mid \hist_i^t, \bm{v}_i\right] \\
    & = \E_{\bm{v}^{[1: T_1]\backslash A}\sim F}\left[D\cdot \delta^{\worst}_{|A|K}(\bm{v}^{[1: T_1]\backslash A})\right] \\
    & = D\cdot \Delta^{\worst}_{T_1K, |A|K} \\
    & \le D\cdot \Delta^{\worst}_{T_1K, \msum K}, 
\end{align*}
where the last inequality is because $|A|\le \msum$ and $\Delta^{\worst}_{N, m_1} \ge \Delta^{\worst}_{N, m_2}$ for $m_1\ge m_2$. 

Since bidder~$i$ participates in at most $\msum$ auctions, the sum of increases of expected utility from round $t$ to $T$ is at most $\msum D \Delta^{\worst}_{T_1K, \msum K}$. 
%, given any history $\hist_i^t$, values $\bm{v}_i$.
\end{proof}

%===========================================================================================

\section{Analysis for MHR Distributions}
%A distribution $F$ has monotone hazard rate if the hazard rate $\frac{f(x)}{1-F(x)}$ is non-decreasing. We say $F$ is MHR if it has monotone hazard rate. 
Recall that a distribution $F$ is MHR if its hazard rate $\frac{f(x)}{1-F(x)}$ is monotone non-decreasing. 

\subsection{Properties of MHR Distributions}
Recall that $R(q)=qv(q)$ is the revenue curve of distribution $F$, where $q(v)=1-F(v)$. And $q^*=\argmax_q R(q)$ is the quantile of the optimal reserve price $v^*=\argmax_v [1-F(v)]v=v(q^*)$. 

For MHR distributions, we first introduce a lemma which says that $q^*$ is bounded away from 0 by a constant.
\begin{lemma}[\citet{hartline2008optimal}]\label{lem:mhr_optimal_q}
Any MHR distribution has a unique $q^*$, and $q^* \ge \frac{1}{e}$.
\end{lemma}

Moreover, the revenue curve decreases quadratically from $q^*$. 
\begin{lemma}[\citet{huang2018making}, Lemma 3.3]\label{lem:mhr_quadratic}
For any MHR $F$, for any $0\le q\le 1$, $R(q^*) - R(q) \ge \frac{1}{4}(q^*-q)^2R(q^*)$. 
\end{lemma}

The following lemma shows that samples from an MHR distribution are rarely too large.
\begin{lemma}\label{lem:mhr_largest_value}
Let $F$ be an MHR distribution. Let $X=\max\{v_1, \ldots, v_N\}$ where $v_1, \ldots, v_N$ are $N$ i.i.d.~samples from $F$. For any $x\ge v^*$, we have $\Pr[X>x] \le Ne^{-x/v^*+1}$.
\end{lemma}
\begin{proof}
 Note that  $1-F(x)=\exp\left\{-\int_0^x\frac{f(v)}{1-F(v)}\dd v\right\} \le \exp\left\{-\int_{v^*}^x\frac{f(v)}{1-F(v)}\dd v\right\}$. By the definition of $v^*$ we know $(v^*[1-F(v^*)])'=0$, or $\frac{f(v^*)}{1-F(v^*)}=\frac{1}{v^*}$. By the definition of MHR, we have $\frac{f(x)}{1-F(x)} \ge \frac{1}{v^*}$ for any $x\ge v^*$, thus 
 \[1-F(x) \le \exp\left\{-\int_{v^*}^x\frac{1}{v^*}\dd v\right\} = \exp\left\{-\frac{x-v^*}{v^*}\right\}.\]
Then the lemma follows from a simple union bound: 
\[\Pr[X>x] = \Pr[\exists i, v_i>x] \le N[1-F(x)] \le N\exp\left\{-\frac{x}{v^*}+1\right\}.\]
\end{proof}

We will use above lemmas to prove some further lemmas which characterize the behavior of $\erm^c$ on samples from a MHR distribution, where $c$ can be any value between $m/N$ and $1/(2e)$. The samples we consider consist of $m$ copies of $+\infty$, denoted by $v_I$, and $N-m$ random draws from $F$. We sort the samples non-increasingly and use
\[v_{-I}=(v_{m+1}\ge v_{m+2}\ge \cdots\ge v_N)\]
to denote the random draws. Let $q_{m+1}\le q_{m+2}\le\cdots\le q_N$ denote their quantiles where $q_j=q(v_j)$. 
\begin{lemma}\label{lem:mhr_revenue_modified}
Let $F$ be an MHR distribution. Suppose $m=o(\sqrt{N})$. Fix $m$ values $v_I$ to be $+\infty$, and randomly draw $N-m$ values $v_{-I}$ from $F$. Let $k^*=\argmax_{i> cN}\{iv_{i}\}$, i.e., the index selected by $\erm^c$, where $\frac{m
}{N}\le c\le \frac{1}{2e}$. Then we have 
\[R(q_{k^*}) \ge \left(1-O\left(\sqrt{\frac{\log N}{N}}\right)\right)R(q^*),\] 
with probability at least $1-O\left(\frac{1}{N}\right)$.
\end{lemma}

\begin{proof}
Let $\gamma \defeq 2\sqrt{ \frac{4\ln(2(N-m))}{N-m} } + \frac{m}{N} = O\left(\sqrt{\frac{\log N}{N}}\right)$ as in \cref{claim:concentration}. We have $| q_j - \frac{j}{N}|\le \gamma$ for any $j>m$ with probability at least $1-\frac{1}{N-m}$. We thus assume $| q_j - \frac{j}{N}|\le \gamma$.

The intuition is follows: The product $jv_j$ divided by $N$ approximates $R(q_j)=q_j v_j$ up to an $O(\gamma)$ error. Our proof consists of three steps: The first step is to show that with high probability, there must be some sample with quantile $q_i$ that is very close to $q^*$ so its revenue $R(q_i)\approx R(q^*) \approx\frac{i}{N}v_i$. The second step is to argue that all samples with quantile $q_j<\frac{1}{2e}$ are unlikely to be chosen by $\erm^c$ because $q_j$ is too small and the gap between $q^*$ and $\frac{1}{2e}$ leads to a large loss in revenue, roughly speaking, $\frac{j}{N}v_j\approx R(q_j) < (1-\frac{1}{4}(\frac{1}{2e})^2) R(q^*)\approx (1-\Omega(1))\frac{i}{N}v_i$. The final step is to show that if a quantile $q_j>\frac{1}{2e}$ is to be chosen by $\erm^c$, then it must have equally good revenue as $q_i$.

Formally:
\begin{enumerate}
\item Firstly, consider the quantile interval $[q^*-\gamma, q^*]$. Each random draw $q_i$, if falling into this interval, will satisfy: 
\begin{equation}\label{eqn:mhr_ivi}
\frac{i}{N}v_i \ge (q_i-\gamma)v_i\ge (q^*-2\gamma)v_i \ge (q^*-2\gamma)v^* \ge (1-2e\gamma)q^*v^*, 
\end{equation}
where the last but one inequality is because $q_i\le q^*$ and the last one follows from $q^*\ge \frac{1}{e}$. The probability that no quantile falls into $[q^*-\gamma, q^*]$ is at most
\[(1-\gamma)^{N-m}=\left(1-O\left(\sqrt{\frac{\log N}{N}}\right)\right)^{N-m} = o(\frac{1}{N}).\]

\item 
For the second step, first note that the $q_i\in[q^*-\gamma, q^*]$ in the first step will be considered by $\erm^c$ since $i\ge(q_i-\gamma)N\ge (q^*-2\gamma)N\ge (\frac{1}{e}-2\gamma)N > cN$. Then suppose $\erm^c$ chooses another quantile $q_j$ instead of $q_i$, we must have
\begin{equation}\label{eqn:mhr_jvj_vs_ivi}
\frac{j}{N}v_j\ge \frac{i}{N}v_i.
\end{equation}
We will show that such probability is small if $q_j<\frac{1}{2e}+\gamma$. Pick a threshold quantile $\frac{1}{T}$ where $T=N^{1/4}$. Consider two cases: 
\begin{itemize}
    \item If $0\le q_j<\frac{1}{T}$. We argue that $\erm^c$ picks $q_j$ with probability at most $o(\frac{1}{N})$. Note that
    \begin{equation}\label{eqn:mhr_jvj_small_q_j}
        \frac{j}{N}v_j \le (q_j+\gamma) v_j \le (\frac{1}{T} + \gamma)v_j,  
    \end{equation}
    together with \eqref{eqn:mhr_jvj_vs_ivi} and \eqref{eqn:mhr_ivi}, we obtain $(\frac{1}{T} + \gamma)v_j \ge (1-2e\gamma)q^*v^*$, implying 
    \[ v_j\ge \frac{1-2e\gamma}{e} \frac{Tv^*}{1+T\gamma} = \Omega(Tv^*), \]
    since $T\gamma=O\left(\frac{\sqrt{\log N}}{N^{1/4}}\right)\to 0$. According to \cref{lem:mhr_largest_value}, the probability that there exists $v_j>\Omega(Tv^*)$ is at most
    \[N \exp\left\{-\frac{\Omega(Tv^*)}{v^*}+1\right\} = o(\frac{1}{N}). \]
    
    \item If $\frac{1}{T}\le q_j<\frac{1}{2e}+\gamma$. We argue that $\erm^c$ will never choose such $q_j$. Note that
    \begin{equation}\label{eqn:mhr_jvj_middle_q_j}
        \frac{j}{N}v_j \le (q_j+\gamma) v_j \le (1+T\gamma)q_jv_j,  
    \end{equation}
    together with \eqref{eqn:mhr_jvj_vs_ivi} and \eqref{eqn:mhr_ivi}, we obtain $(1+T\gamma)q_jv_j> (1-2e\gamma)q^*v^*$. Then by \cref{lem:mhr_quadratic},
    \begin{equation}\label{eqn:mhr_case_middle_conclude}
    \frac{1-2e\gamma}{1+T\gamma} \le \frac{q_jv_j}{q^*v^*} \le 1-\frac{1}{4}(q_j-q^*)^2 \le 1-\frac{1}{4}(\frac{1}{2e}-\gamma)^2.
    \end{equation}
    However, the left hand side of \eqref{eqn:mhr_case_middle_conclude} approaches to 1 since $\gamma$ and $T\gamma$ approach 0 while the right hand side is strictly less than 1, a contradiction. So this case never happens.
    
\end{itemize}

\item Finally, if $q_j\ge \frac{1}{2e}+\gamma$. We argue that if $\erm^c$ picks $q_j$ instead of $q_i$, then $R(q_j)$ approximates $R(q_i)$ well, satisfying the conclusion in the lemma. This is because
    \begin{align*}
        R(q_j) = q_jv_j \ge{}& (\frac{j}{N}-\gamma)v_j \\
        \ge{}& (1-2e\gamma)\frac{j}{N}v_j && \frac{j}{N}\ge q_j-\gamma\ge \frac{1}{2e}\\
        \ge{}& (1-2e\gamma)\frac{i}{N}v_i && \text{\cref{eqn:mhr_jvj_vs_ivi}}\\
        \ge{}& (1-2e\gamma)(1-2e\gamma)q^*v^* && \text{\cref{eqn:mhr_ivi}} \\
        ={}& (1-O(\gamma))R(q^*).
    \end{align*}
\end{enumerate}

Combining above three steps and the event in the beginning of the proof, we have $R(q_{k^*}) \ge (1-O(\sqrt{\frac{\log N}{N}}))R(q^*)$ except with probability at most 
\[\frac{1}{N-m} + o(\frac{1}{N}) + o(\frac{1}{N}) = O(\frac{1}{N}). \]
\end{proof}

\begin{lemma}\label{lem:mhr_event}
Let $F$ be an MHR distribution. Suppose $m=o(\sqrt{N})$. Fix $m$ values $v_I$ to be $+\infty$, and randomly draw $N-m$ values $v_{-I}$ from $F$. Let $k^*=\argmax_{i> cN}\{iv_{i}\}$, i.e., the index selected by $\erm^c$, where $\frac{m
}{N}\le c\le \frac{1}{2e}$. Let $\epsilon=\sqrt[4]{\frac{\log N}{N}}$. Then with probability at least $1-O\left(\frac{1}{N}\right)$, the following inequalities hold:
\begin{enumerate}
    \item $q_{k^*}\ge q^*-O(\epsilon)$;
    \item $k^* \ge [q^*-O(\epsilon)] N > \frac{1}{2e}N$; 
    \item $v_{k^*}\le [1+O(\epsilon)]v^*$. 
\end{enumerate}
\end{lemma}
\begin{proof}
For inequality (1), by \cref{lem:mhr_quadratic} and \cref{lem:mhr_revenue_modified}, with probability at least $1-O(\frac{1}{N})$, we have
\[ \frac{1}{4}(q_{k^*} - q^*)^2 \le \frac{R(q^*) - R(q_{k^*})}{R(q^*)} \le O\left(\sqrt{\frac{\log N}{N}}\right). \]
Taking the square root, we obtain $q_{k^*}\ge q^* - O\left(\sqrt[4]{\frac{\log N}{N}}\right)$. 

Assume that (1) holds. To prove (2), note that by \cref{claim:concentration}, we have $\frac{k^*}{N}\ge q_{k^*} - O\left(\sqrt{\frac{\log N}{N}}\right) \ge q^*-O(\epsilon)$ except with probability at most $O(\frac{1}{N})$, and $q^*>\frac{1}{e}$. 

Finally, inequality (3) follows from
\[\frac{v_{k^*}}{v^*}=\frac{R(q_{k^*})}{q_{k^*}}\frac{q^*}{R(q^*)}\le 1\cdot \frac{q^*}{q_{k^*}}\le \frac{q^*}{q^*-O(\epsilon)} = 1+\frac{O(\epsilon)}{q^*-O(\epsilon)} \le 1+O(e\epsilon).  \]
\end{proof}

\subsection{Detailed Proof of Theorem \ref{thm:connection_TPM_PRM_combined} for MHR Distributions}
\label{sec:proof_lemma_1_mhr}

Let $\Delta U(\bm{v}_i, b_I, v_{-I}) = U_i^{\TP}(\bm{v}_i, b_I, v_{-I}) - U_i^{\TP}(\bm{v}_i, v_I, v_{-I})$. 
Similar to the proof for bounded distributions, we have for any $\bm{v}_i, b_I, v_{-I}$, 
\begin{align*}
\Delta U(\bm{v}_i, b_I, v_{-I}) 
\le{} m_2\cdot \bigg(\erm^c(v_I, v_{-I}) - \erm^c(b_I, v_{-I}) \bigg) \le{} m_2 \cdot \erm^c(v_I, v_{-I}) \cdot \delta^{\interim}_{m_1}(v_{-I}).
\end{align*}

By \cref{claim:interim_to_delta}, we have $\erm^c(v_I, v_{-I})\le \erm^c(\overline{v}_I, v_{-I})$ where $\overline{v}_I$ can be any $m_1$ values (e.g., $+\infty$) that are greater than the maximal value in $v_{-I}$, when $c\ge \frac{m_1}{T_1K_1}$. 

Let $N=T_1K_1$, define two threshold prices $T_1=\sqrt{N}v^*$ and $T_2=[1+O(\epsilon)]v^*$ where $\epsilon=\sqrt[4]{\frac{\log N}{N}}$ as in \cref{lem:mhr_event}. 
%$\Theta(\log N)$ (for some constant in $\Theta$ to be decided later), $T_2=(1+\sqrt{\epsilon})v^*$ where $v^*=\argmax_{v}\{v[1-F(v)]\}$ and $\epsilon = O\left(\sqrt{\frac{\log N}{N}}\right)$ (for some constant in $O$ to be decided later). 
Note that for sufficiently large $N$, $T_1>T_2$. With the random draw of $v_{-I}$ from $F$, denote the random variable $\erm^c(\overline{v}_I, v_{-I})$ by $P$, we have: 

\begin{align}
\E_{\bm{v}_{-i}}\left[ \Delta U(\bm{v}_i, b_I, v_{-I})\right] 
= {} & \E_{v_{-I}}\left[ \Delta U(\bm{v}_i, b_I, v_{-I}) \mid P\le T_2\right]\cdot \Pr[P\le T_2] \nonumber\\
+ {} & \E_{v_{-I}}\left[ \Delta U(\bm{v}_i, b_I, v_{-I}) \mid T_2<P\le T_1\right]\cdot \Pr[T_2<P\le T_1] \nonumber \\
+ {} &\E_{v_{-I}}\left[ \Delta U(\bm{v}_i, b_I, v_{-I}) \mid P>T_1\right]\cdot \Pr[P>T_1] \nonumber \\
\defeq{}& \E_1 + \E_2 + \E_3. 
\end{align}

\begin{enumerate}
\item For the first term $\E_1$, 
\begin{align*}
\E_1 ={} &\E_{v_{-I}}\left[ \Delta U(\bm{v}_i, b_I, v_{-I}) \mid P\le T_2\right]\cdot \Pr[P\le T_2] \\
\le{}& \E_{v_{-I}}\left[ m_2\cdot P \cdot \delta^{\interim}_{m_1}(v_{-I}) \mid P\le T_2\right]\cdot \Pr[P\le T_2] \\
\le{} & m_2 \cdot T_2\cdot \E_{v_{-I}}\left[\delta^{\interim}_{m_1}(v_{-I})  \mid P\le T_2\right] \cdot \Pr[P\le T_2] \\
\le{} & m_2 \cdot [1+O(\epsilon)]v^*\cdot \E_{v_{-I}}\left[ \delta^{\interim}_{m_1}(v_{-I})\right] \\
={} & O\left(m_2 \cdot v^*\cdot \Delta^{\interim}_{N, m_1}  \right).
\end{align*}

\item For the second term, we claim that $\E_2=O(\frac{m_2v^*}{\sqrt{N}})$.

By \cref{lem:mhr_event}, we have $\Pr[P>[1+O(\epsilon)]v^*]\le O(\frac{1}{N})$. 
Therefore, 
\begin{align*}
\E_2 ={} &\E_{v_{-I}}\left[ \Delta U(\bm{v}_i, b_I, v_{-I}) \mid T_2< P\le T_1\right]\cdot \Pr[T_2< P\le T_1] \\
\le{} & \E_{v_{-I}}\left[ m_2\cdot P \cdot 1 \mid T_2< P\le T_1\right]\cdot \Pr[T_2< P\le T_1] \\
\le{} & m_2 \cdot T_1\cdot \Pr[P> T_2] \\
\le{} & m_2 \cdot \sqrt{N}v^* \cdot O \left(\frac{1}{N}\right) \\
={} & O\left(\frac{m_2v^*}{\sqrt{N}}\right).\\
\end{align*}

\item For the third term, we claim that $\E_3 = o(\frac{m_2v^*}{N})$.

Let $B$ be the upper bound on the support of $F$ ($B$ can be $+\infty$). Let $F_{P}(x)$ be the distribution of $P$. For convenience, suppose it is continuous and has density $f_P(x)$. We have: 
\begin{align*}
\E_3 ={}& \E_{v_{-I}}\left[ \Delta U(\bm{v}_i, b_I, v_{-I}) \mid P>T_1\right]\cdot \Pr[P>T_1] \\
\le{}& \E_{v_{-I}}\left[m_2\cdot P\cdot 1\mid P>T_1 \right]\cdot\Pr[P>T_1]\\
= {}& m_2\cdot \E_{v_{-I}}\left[ P \mid P>T_1 \right]\cdot\Pr[P>T_1] \\
={}& m_2\cdot \int_{T_1}^B xf_{P}(x)\dd x \\
={}& m_2\cdot\left( \int_{T_1}^B [1-F_{P}(x)]\dd x + T_1[1-F_{P}(T_1)]\right).
\end{align*}

Let $\max\{v_{-I}\}$ denote the maximum value in the $N-m_1$ samples $v_{-I}$. By \cref{lem:mhr_largest_value}, we have for any $x\ge v^*$, 
\[1-F_{P}(x)= \Pr[P>x] \le \Pr[\max\{v_{-I}\}>x] \le Ne^{-\frac{x}{v^*}+1}.\]
Thus, 
\begin{align*}
\int_{T_1}^B [1-F_P(x)]\dd x + T_1[1-F_{P}(T_1)] \le{}& v^*Ne^{-\frac{T_1}{v^*}+1} + T_1Ne^{-\frac{T_1}{v^*}+1} \\
={}& v^*N(1+\sqrt{N})e^{-\sqrt{N}+1} \\
={}& o(\frac{v^*}{N}),
\end{align*}
as desired.
\end{enumerate}

Combining the three items, 
\begin{align*} \E_{\bm{v}_{-i}}\left[\Delta U(\bm{v}_i, b_I, v_{-I})\right] = O\left(m_2v^* \Delta^{\interim}_{N, m_1}  \right) + O\left(\frac{m_2v^*}{\sqrt{N}}\right). \end{align*}

\subsection{An Improved Bound on Incentive-Awareness Measure for MHR Distributions}\label{sec:improve_MHR}
Here we improve the upper bound on $\Delta^{\interim}_{N, m}$ for MHR distributions by proving:

\begin{lemma}[Tigher bound for MHR distributions]\label{lem:mhr_main_upper-bound}
Moreover, if $F$ is MHR, let $d=\frac{1}{2e}$, and suppose $\frac{m}{N}\le c\le \frac{1}{4e}$, we have 
%\[\Delta^{\interim}_{N, m} \le O \left(\frac{m}{d^{7/2}}\frac{\log^3 N}{\sqrt{N}}\right) + \Pr[\event]. \]
\[\Delta^{\interim}_{N, m} \le O \left(\frac{m} {d^{7/2}}\frac{\log^3 N}{\sqrt{N}}\right) + \Pr[\event]. \]
\end{lemma}

The main idea is to limit the range of the quantile $q_j$ of the ``bad value'' $v_j$ in $\Bad(\eta_t, \theta_t)$ in \cref{lem:Pr_Bad}. Recall that in the proof of \cref{lem:Pr_Bad} we assume $q_j$ can take any value in $[0, 1]$, divide $[0, 1]$ into $O(1/h)$ intervals (as in \eqref{eqn:quantile_divide}), and take a union bound to upper-bound the probability that a bad $q_j$ exists.  For MHR distributions, however, we will show that $q_jv_j$ is a $(1-O(\sqrt{\frac{\log N}{N}}))$ approximation to $R(q^*)$, thus we can use \cref{lem:mhr_quadratic} to reduce the possible range of $q_j$ from 1 to $O(\sqrt[4]{\frac{\log N}{N}})$.

\subsubsection{Proof of \cref{lem:mhr_main_upper-bound}}

We repeat the argument for \cref{lem:main_upper-bound} until \cref{claim:concentration}, before which we have:
\begin{align}
\Delta^{\interim}_{N, m} \le {}&\int_0^1\Pr[\delta_I(\overline{v}_I, v_{-I})>\eta\, \land \overline{\event} ]\dd \eta ~+~ \Pr[\event]. \label{eqn:Delta_interim_integral_MHR}
\end{align}

Let $\gamma=O(\sqrt{\frac{\log N}{N}})$ be the upper bound on $|q_j-\frac{j}{N}|$ in $\conc$. With the random draw of $N-m$ samples $v_{-I}$ from $F$ (and assume other $m$ samples $\overline{v}_I$ are equal to $\max\{v_{-I}\}$), we have $|q_j-\frac{j}{N}|<\gamma$ for any $j>m$ with probability at least $1-\frac{1}{N-m}$. Moreover, by \cref{lem:mhr_revenue_modified} and \cref{lem:mhr_event}, with probability at least $1-O(\frac{1}{N})$ we have $R(q_{k^*}) = q_{k^*}v_{k^*} \ge (1-\epsilon) R(q^*)$ and $q_{k^*} \ge q^*-O(\sqrt{\epsilon}) > \frac{1}{2e}$, 
where $\epsilon=O(\sqrt{\frac{\log N}{N}})$. Combine the above two inequalities with $\conc$ and denote the combined event by $\conc'$, i.e., 
\[\conc' \defeq \conc\,\land\, \left[R(q_{k^*})\ge (1-\epsilon) R(q^*)\right] \,\land\, \left[q_{k^*} \ge q^*-O(\sqrt{\epsilon}) > \frac{1}{2e}\right]. \]
We have $\Pr[\overline{\conc'}]\le O(\frac{1}{N})$. Re-define $G(\eta) = \Pr[\delta_I(\overline{v}_I, v_{-I})>\eta\,\land\,\overline{\event}\,\land\,\conc']$, and re-write \eqref{eqn:delta_event_conc}:
\begin{equation}\label{eqn:delta_event_conc_mhr}
\Pr[\delta_I(\overline{v}_I, v_{-I})>\eta\,\land \overline{\event}] \le G(\eta) + O\left(\frac{1}{N}\right).
\end{equation}

The following steps of bounding $G(\eta) =\Pr[\delta_I(\overline{v}_I, v_{-I})>\eta\,\land\,\overline{\event}\,\land\,\conc']$ are the same as before (in particular, \cref{lem:thresholds} in \cref{lem:first_term}), until upper-bounding $\Pr[\Bad(\first, \second)\,\land\,\overline{\event}\,\land\,\conc']$ (\cref{lem:Pr_Bad}), where we improve the bound by a factor of $\sqrt[4]{\frac{\log N}{N}}$. 
\begin{lemma}[Improved \cref{lem:Pr_Bad} for MHR distributions]\label{lem:Pr_Bad_MHR}
Let $d=\frac{1}{2e}$. 
If $\first$ and $\second$ are at least $\Omega\left(\frac{m}{d}\sqrt{\frac{\log (N-m)}{N-m}}\right)$, then 
$\Pr[\Bad(\first, \second)\,\land\,\overline{\event}\,\land\,\conc'] = O\left(\sqrt[4]{\frac{\log N}{N}}\frac{m\log^2 N}{d^4\second\sqrt{\first^3N}}\right)$.
\end{lemma}

\begin{proof}
The proof is the same as that of \cref{lem:Pr_Bad} (in \cref{sec:proof_lem:Pr_Bad}), except that before dividing the quantile space $[0, 1]$, we argue that the space to be divided can be shortened to $[q^*-O(\sqrt{\epsilon}), q^*+O(\sqrt{\epsilon})]$. 

Consider the index $j$ that is promised to exist in $\Bad(\first, \second)$,
\begin{align*}
R(q_j) ={} & q_jv_j \\
\ge{}& (\frac{j}{N}-\gamma) v_j  && \text{$\conc'$} \\
\ge{}& \frac{k^*}{N}v_{k^*}-\frac{m}{N\second}v_{k^*} - \gamma v_{k^*} && \text{$jv_j\ge k^*v_{k^*}-\frac{m}{\second}v_{k^*}$ and $v_{j^*}\le v_{k^*}$}
\\
\ge{}& (q_{k^*}-\gamma)v_{k^*} -\frac{m}{N\second}v_{k^*} - \gamma v_{k^*} && \conc' \\
={}& \left[q_{k^*} - (2\gamma+ \frac{m}{N\second})\right] v_{k^*} \\
={}&  \left[1 - \frac{2\gamma+ \frac{m}{N\second}}{q_{k^*}} \right] R(q_k^*) \\
\ge{}& \left[1 - 2e(2\gamma+ \frac{m}{N\second}) \right] (1-\epsilon)R(q^*) && \text{$q_{k^*}\ge \frac{1}{2e}$ and $R(q_{k^*})\ge(1-\epsilon)R(q^*)$ in $\conc'$} \\ 
={}& \left[1 - O\left(\sqrt{\frac{\log N}{N}}\right) \right] R(q^*)  && \text{Definition of $\gamma$ and $\epsilon$, and $\theta=\Omega\left(m\sqrt{\frac{\log N}{N}}\right)$}. 
\end{align*}

By \cref{lem:mhr_quadratic}, we have 
\begin{equation*}%\label{eqn:q_j_range}
\left| q_j - q^* \right | \le 2\sqrt{O\left(\sqrt{\frac{\log N}{N}}\right)} = O\left(\sqrt[4]{\frac{\log N}{N}}\right).
\end{equation*}

Now we modify the analysis after \cref{claim:two_consequences}. Consider those intervals $I_l$'s with length $h$ in \eqref{eqn:quantile_divide} that intersect with 
\[I_{q_j}= \left[q^*-O\left(\sqrt[4]{\frac{\log N}{N}}\right), q^*+O\left(\sqrt[4]{\frac{\log N}{N}}\right)\right].\]
There are at most $O\left(\frac{1}{h}\sqrt[4]{\frac{\log N}{N}}\right)$ such intervals and we denote the set of (indices of) those intervals by $\mathcal{L}$. 
%\begin{equation}\label{eqn:quantile_divide}
%[0, 1] = [0, \frac{d}{2}]\cup I_1\cup I_2\cup \cdots\cup I_{\frac{1-d/2}{h}},
%\end{equation}
%where $I_l = (d/2 + (l-1)h, d/2 + lh]$.
%Thus a uniformly random draw of quantile falls into $I_l$ with probability $h$.
The definitions of $i^*_l$, $i^*_{<(l+1)}$, $A_l$, $A_{<(l+1)}$, and $W_l$ remain unchanged. 
%be the optimal indices of quantiles in interval $I_l$ and intervals $I_1$ to $I_l$, i.e.
%\[i^*_l = \argmax_{i> cN} \big\{iv_i\mid q_i\in I_l\big\} ~\text{ or }~ i^*_l=\emptyset \text{ if there is no such }i. \]
%\[i^*_{<(l+1)} = \argmax_{i> cN} \big\{iv_i \mid q_i\in I_1\cup\cdots\cup I_l\big\} ~\text{ or }~ i^*_{<(l+1)}=\emptyset \text{ if there is no such }i. \]
%And
%we denote the optimal value by $A_l, A_{<(l+1)}$: 
%$A_l \defeq i^*_lv_{i^*_l},~~A_{<(l+1)} \defeq i^*_{<(l+1)} v_{i^*_{<(l+1)}}$. Moreover, define event $W_l$ for each $l$, 
%$$W_l=[A_{l+2} \le A_{<(l+1)} \le (i^*_{l+2} + H)v_{i^*_{l+2}}], $$
%$$W_l=\left[\left\{i^*_{l+2}\ne \emptyset\right\} \ \land\   \left\{i^*_{<(l+1)}\ne \emptyset\right\} \ \land\  \left\{A_{l+2} \le A_{<(l+1)} \le A_{l+2} + H\widetilde{v}_{l+2}\right\}\right], $$
%where $\widetilde{v}_{l+2}\defeq v(d/2+(l+1)h)$ is the upper-bound on the values with quantiles in $I_{l+2}$.
By choosing the index $l$ such that $q_j\in I_{l+2}$, we know that if the event $[\Bad(\first, \second)\,\land\, \overline{\event}\,\land\,\conc']$ holds then $W_l$ must hold for some $l$ such that $l+2\in\mathcal{L}$. By \cref{lem:Pr_W} we have $\Pr[W_l] \le O(\frac{H\log^2 N}{\sqrt{h d^3 N}})$. Taking a union bound over $l$, we obtain 
\begin{equation*}
   \Pr[\Bad(\first, \second) \,\land\,\overline{\event}\,\land\,\conc'] \le O\left(\frac{1}{h}\sqrt[4]{\frac{\log N}{N}} \cdot \frac{H\log^2 N}{\sqrt{h d^3 N}}\right) = O\left(\sqrt[4]{\frac{\log N}{N}}\frac{m\log^2 N}{d\second\sqrt{(d\first)^3 d^3 N}}\right),
\end{equation*}
where the last equality is because $H=O(\frac{m}{d\second})$ and $h=\Omega(d\first)$ under the assumption that $\first$ and $\second$ are at least $\Omega(\frac{m}{d}\sqrt{\frac{\log(N-m)}{N-m}})$.
\end{proof}

Then we improve \cref{lem:first_term} based on \cref{lem:Pr_Bad_MHR}.  
\begin{lemma}[Improved \cref{lem:first_term}] \label{lem:first_term_MHR}
%Let $d=\frac{1}{2e}$. 
%There exists a constant $C=\Theta\left(\frac{m^{3/2}\log^3 N}{d^4N^{3/4}}\right)$, such that $\eta > C^{2/3} \implies G(\eta) < \frac{C}{\eta^{3/2}}$. 
 If $\eta$ is at least $\Omega\left(\frac{m}{d}\sqrt{\frac{\log (N-m)}{N-m}}\right)$, then $G(\eta) =  O\left(\frac{m\log^3 N}{d^4N^{3/4}}\frac{1}{\eta^{3/2}}\right)$. 
%\[\Pr[\delta_I(\overline{v}_I, v_{-I})>\eta\,\land\,\overline{\event}\,\land\,\conc'] =  O\left(\frac{m\log^3 N}{d^4N^{3/4}}\frac{1}{\eta^{3/2}}\right).\] 
\end{lemma}
\begin{proof} Modify the end of \cref{sec:proof_lem:first_term}, 
\begin{flalign*}
\Pr[\delta_I(\overline{v}_I, v_{-I})>\eta\,\land\,\overline{\event}\,\land\,\conc'] \le{}& \sum_{t=0}^M \Pr[\mathrm{Bad}(\eta_t, \theta_t) \,\land\,\overline{\event}\,\land\,\conc'] && \text{Lemma \ref{lem:thresholds}}\\
={}& \sum_{t=0}^M O\left(\sqrt[4]{\frac{\log N}{N}}\frac{m\log^2 N}{d^4\theta_t\sqrt{\eta_t^3N}}\right) && \text{Lemma~\ref{lem:Pr_Bad_MHR}}\\
={}& \sum_{t=0}^M O\left(\sqrt[4]{\frac{\log N}{N}}\frac{m\log^2 N}{d^4 \sqrt{N}}\frac{\eta_{t+1}} { \eta\sqrt{\eta_t^3}}\right) && \text{Definition of } \theta_t\\
={}& O\left(\sqrt[4]{\frac{\log N}{N}}\frac{m\log^2 N}{d^4\sqrt{N}} \cdot \sum_{t=0}^M  \frac{\eta_{t+1}}{\eta}\frac{1}{\eta_{t}^{3/2}} \right) && %\frac{\eta}{\eta_{t+1}} \ge \eta > \Omega(\frac{m}{dN})
\end{flalign*}
Note that because $\eta_t, \theta_t\ge \frac{\eta}{2}$, the condition of \cref{lem:Pr_Bad_MHR} is satisfied when $\eta=\Omega(\frac{m}{d}\sqrt{\frac{\log (N-m)}{N-m}})$. 
%\new{???}

%$\eta\ge \Theta\left(\left(\frac{m^{3/2}\log^3 N}{d^4N^{3/4}}\right)^{2/3}\right)=\Theta\left(\frac{m\log^2 N}{d^{8/3}N^{1/2}}\right)$.

By \cref{cl:eta-calculation}, we can choose a sequence of $\{\eta_t\}$ such that  
\begin{align*}
    \sum_{t=0}^M  \frac{\eta_{t+1}}{\eta}\frac{1}{\eta_{t}^{3/2}} = O\left(\frac{\log\log (N-m)}{\eta^{3/2}}\right). 
\end{align*}
Therefore, 
\[\Pr[\delta_I(\overline{v}_I, v_{-I})>\eta\,\land\,\overline{\event}\,\land\,\conc'] \le O\left(\frac{m\log^{2+1/4} N}{d^4 N^{1/2 + 1/4}}\cdot \frac{ \log\log (N-m) )}{\eta^{3/2}} \right)=  O\left(\frac{m\log^3 N}{d^4 N^{3/4}} \frac{1}{\eta^{3/2}}\right).\] 
\end{proof}

We finish the proof of Lemma~\ref{lem:mhr_main_upper-bound} by computing the integral in \eqref{eqn:Delta_interim_integral_MHR}. Let $C=\Theta\left(\frac{m\log^3 N}{d^4 N^{3/4}}\right)$ be the bound on $G(\eta)$ in Lemma~\ref{lem:first_term_MHR}, and let $A=\Theta(\frac{m}{d}\sqrt{\frac{\log (N-m)}{N-m}})=\Theta(\frac{m}{d}\sqrt{\frac{\log N}{N}})$ be the condition on the lower bound on $\eta$ in Lemma~\ref{lem:first_term_MHR}. Then
$G(\eta) < \frac{C}{\eta^{3/2}}$
when $\eta>\max\{C^{2/3}, A\}$. If $C^{2/3}>A$, then we have 
%$G(\eta)=\Pr[\delta_I(\overline{v}_I, v_{-I})>\eta\ \land \overline{\event}~\land~\conc]$. When $\eta>C^{2/3}$, we have $G(x) \le \frac{C}{x^{3/2}}$. Therefore: 
\begin{align*}
    \int_0^1\Pr[\delta_I(\overline{v}_I, v_{-I})>\eta\, \land \overline{\event} ]\dd \eta \le{}&
    \int_0^1 \left(G(x) + O\left(\frac{1}{N}\right)\right)\dd x && \text{by \eqref{eqn:delta_event_conc_mhr}}\\
    \le{} & \int_0^{C^{\frac{2}{3}}} 1\dd x + \int_{C^{\frac{2}{3}}}^1 \frac{C}{x^{\frac{3}{2}}}\dd x 
    +O\left(\frac{1}{N}\right) \\
    ={}& C^\frac{2}{3} + \frac{C}{-\frac{1}{2}} - \frac{C}{-\frac{1}{2}}C^{-\frac{1}{3}} +O\left(\frac{1}{N}\right) \\
    \le{}& 3C^\frac{2}{3} +O\left(\frac{1}{N}\right)\\
    ={}& O \left(\frac{m^{2/3}} {d^{8/3}}\frac{\log^2 N}{\sqrt{N}}\right)+O\left(\frac{1}{N}\right)\\
    ={}& O \left(\frac{m^{2/3}} {d^{8/3}}\frac{\log^2 N}{\sqrt{N}}\right). 
\end{align*}
If $A>C^{2/3}$, then we have: 
\begin{align*}
    \int_0^1\Pr[\delta_I(\overline{v}_I, v_{-I})>\eta\, \land \overline{\event} ]\dd \eta \le{}&
    \int_0^1 \left(G(x) + O\left(\frac{1}{N}\right)\right)\dd x && \text{by \eqref{eqn:delta_event_conc_mhr}}\\
    \le{} & \int_0^{A} 1\dd x + \int_{A}^1 \frac{C}{x^{\frac{3}{2}}}\dd x 
    +O\left(\frac{1}{N}\right) \\
    ={}& A + \frac{C}{-\frac{1}{2}} - \frac{C}{-\frac{1}{2}}A^{-\frac{1}{2}} +O\left(\frac{1}{N}\right) \\
    \le {}& \Theta\left(\frac{m}{d}\sqrt{\frac{\log N}{N}}\right) + \Theta\left(\frac{m^{1-1/2}\log^{3-1/4}N}{d^{4-1/2}N^{3/4-1/4}}\right) +O\left(\frac{1}{N}\right)\\
    ={}& O \left(\frac{m} {d^{7/2}}\frac{\log^3 N}{\sqrt{N}}\right), 
\end{align*}
which, together with \eqref{eqn:Delta_interim_integral_MHR}, concludes the proof of Lemma~\ref{lem:mhr_main_upper-bound}.

%==========================================================================================
%==========================================================================================

\section{Analysis for $\alpha$-Strongly Regular Distributions}\label{sec:appendix-alpha-strong}

\subsection{Useful Lemmas}

\begin{lemma}[\citet{cole2014the}]\label{lem:alpha_optimal_q}
Any $\alpha$-strongly regular distribution has a unique $q^*$, and $q^* \ge \alpha^{\frac{1}{1-\alpha}}$.
\end{lemma}

\begin{lemma}[\citet{huang2018making}, Lemma 3.5]\label{lem:alpha_quadratic}
For any $\alpha$-strongly regular distribution $F$, for any $0\le q\le 1$, $R(q^*) - R(q) \ge \frac{\alpha}{3}(q^*-q)^2R(q^*)$. 
\end{lemma}

\begin{lemma}\label{lem:alpha_largest_value}
Let $F$ be an $\alpha$-strongly regular distribution. Let $X=\max\{v_1, \ldots, v_N\}$ where $v_1, \ldots, v_N$ are $N$ i.i.d.~samples from $F$. For any $x\ge v^*$, we have $\Pr[X>x] \le N\left( \frac{v^*}{(1-\alpha)x+\alpha v^*} \right)^{\frac{1}{1-\alpha}}$.
\end{lemma}
\begin{proof}
 Note that $1-F(x)=\exp\left\{-\int_0^x\frac{f(v)}{1-F(v)}\dd v\right\} \le \exp\left\{-\int_{v^*}^x\frac{f(v)}{1-F(v)}\dd v\right\}$. By the definition of $v^*$ we know $(v^*[1-F(v^*)])'=0$, or $\frac{f(v^*)}{1-F(v^*)}=\frac{1}{v^*}$. By the definition of $\alpha$-strong regularity, we have 
 \begin{equation*}
     \left( \frac{1-F(x)}{f(x)} \right)'=1-\frac{d\phi}{dx} \le 1-\alpha
 \end{equation*}
 and 
 \begin{equation*}
     \frac{1-F(x)}{f(x)} \le \frac{1-F(v^*)}{f(v^*)} + (1-\alpha) (x-v^*).
 \end{equation*}
 Thus 
 \begin{align*}
     \int_{v^*}^x\frac{f(v)}{1-F(v)} & \ge \int_{v^*}^x \frac{1}{\frac{1-F(v^*)}{f(v^*)} + (1-\alpha) (v-v^*)}\dd v = \frac{1}{1-\alpha} \left[ \ln \left(v^*+(1-\alpha)(x-v^*)\right) - \ln v^* \right]
 \end{align*}
 and 
 \begin{equation*}
     1-F(x) \le \exp \left\{-\frac{1}{1-\alpha} \ln \frac{v^*+(1-\alpha)(x-v^*)}{v^*} \right\} = \left( \frac{v^*}{(1-\alpha)x+\alpha v^*} \right)^{\frac{1}{1-\alpha}}.
 \end{equation*}
Then the lemma follows from a simple union bound: 
\[\Pr[X>x] = \Pr[\exists i, v_i>x] \le N[1-F(x)] \le N\left( \frac{v^*}{(1-\alpha)x+\alpha v^*} \right)^{\frac{1}{1-\alpha}}.\]
\end{proof}

\begin{claim} \label{claim:alpha_concentration}
(Improved \cref{claim:concentration})
Define event $\conc$: 
\[ \conc=\left[\forall j>m, \left | q_j - \frac{j}{N} \right | \le 2\sqrt{\frac{3}{\alpha} \frac{\ln(2(N-m))}{N-m} } + \frac{m}{N} \right], \]
then $\Pr[\overline{\conc}] \le \frac{1}{(N-m)^{\frac{3-2\alpha}{2\alpha}}}$,
where the probability is over the random draw of the $N-m$ samples $v_{-I}$.
\end{claim}
\begin{proof}
Set $\delta=\frac{1}{(N-m)^{\frac{3-2\alpha}{2\alpha}}}$ in \cref{lem:concentration}. 
\end{proof}

\begin{lemma}\label{lem:alpha_revenue_modified}
Let $F$ be an $\alpha$-strongly regular distribution. Suppose $m=o(\sqrt{N})$. Fix $m$ values $v_I$ to be $+\infty$, and randomly draw $N-m$ values $v_{-I}$ from $F$. Let $k^*=\argmax_{i> cN}\{iv_{i}\}$, i.e., the index selected by $\erm^c$, where $\left(\frac{\log N}{N}\right)^{\frac{1}{3}}\le c\le \frac{\alpha^{1/(1-\alpha)}}{2}$. Then we have 
\[R(q_{k^*}) \ge \left(1-O\left(\sqrt{\frac{\log N}{N}}\right)\right)R(q^*),\] 
with probability at least $1-O\left(\frac{1}{N^{\frac{3-2\alpha}{2\alpha}}}\right)$. The constants in the big $O$'s depend on $\alpha$. 

\end{lemma}

\begin{proof}
Let $\gamma \defeq 2\sqrt{\frac{3}{\alpha} \frac{\log(2(N-m))}{N-m} } + \frac{m}{N} = O\left(\sqrt{\frac{\log N}{N}}\right)$ as in \cref{claim:alpha_concentration}. We have $| q_j - \frac{j}{N}|\le \gamma$ for any $j>m$ with probability at least $1-\frac{1}{(N-m)^{\frac{3-2\alpha}{2\alpha}}}$. We thus assume $| q_j - \frac{j}{N}|\le \gamma$. For simplicity, we define $e(\alpha)=\alpha^{1/(1-\alpha)}$, and \cref{lem:alpha_optimal_q} implies $q^* \ge e(\alpha)$.

The intuition is follows: The product $jv_j$ divided by $N$ approximates $R(q_j)=q_j v_j$ up to an $O(\gamma)$ error. Our proof consists of three steps: The first step is to show that with high probability, there must be some sample with quantile $q_i$ that is very close to $q^*$ so its revenue $R(q_i)\approx R(q^*) \approx\frac{i}{N}v_i$. The second step is to argue that all samples with quantile $q_j<\frac{e(\alpha)}{2}$ are unlikely to be chosen by $\erm^c$ because $q_j$ is too small and the gap between $q^*$ and $\frac{e(\alpha)}{2}$ leads to a large loss in revenue, roughly speaking, $\frac{j}{N}v_j\approx R(q_j) < (1-\frac{\alpha}{3}(\frac{e(\alpha)}{2})^2) R(q^*)\approx (1-\Omega(1))\frac{i}{N}v_i$. The final step is to show that if a quantile $q_j>\frac{e(\alpha)}{2}$ is to be chosen by $\erm^c$, then it must have equally good revenue as $q_i$.

Formally:
\begin{enumerate}
    \item Firstly, consider the quantile interval $[q^*-\gamma, q^*]$. Each random draw $q_i$, if falling into this interval, will satisfy: 
\begin{equation}\label{eqn:alpha_ivi}
\frac{i}{N}v_i \ge (q_i-\gamma)v_i\ge (q^*-2\gamma)v_i \ge (q^*-2\gamma)v^* \ge (1-2\gamma/e(\alpha))q^*v^*, 
\end{equation}
where the last but one inequality is because $q_i\le q^*$ and the last one follows from $q^*\ge e(\alpha)$. The probability that no quantile falls into $[q^*-\gamma, q^*]$ is at most
\[(1-\gamma)^{N-m}=\left(1-O\left(\sqrt{\frac{\log N}{N}}\right)\right)^{N-m} = o(\frac{1}{N^{\frac{3-2\alpha}{2\alpha}}}).\]

\item For the second step, first note that the $q_i\in[q^*-\gamma, q^*]$ in the first step will be considered by $\erm^c$ since $i\ge(q_i-\gamma)N\ge (q^*-2\gamma)N\ge (e(\alpha)-2\gamma)N > cN$. Then suppose $\erm^c$ chooses another quantile $q_j$ instead of $q_i$, we must have
\begin{equation}\label{eqn:alpha_jvj_vs_ivi}
\frac{j}{N}v_j\ge \frac{i}{N}v_i.
\end{equation}
We will show that such $q_j$ does not exist. 

Suppose $\erm^c$ chooses $q_j$, then $j$ must satisfy $j/N > c > \left( \frac{\log N}{N} \right)^{\frac{1}{3}}$, and as a result, $q_j > \left( \frac{\log N}{N} \right)^{\frac{1}{3}} - \gamma$.

If $\left( \frac{\log N}{N} \right)^{\frac{1}{3}} - \gamma < q_j<\frac{e(\alpha)}{2}+\gamma$, note that
\begin{equation}\label{eqn:alpha_jvj_middle_q_j}
\frac{j}{N}v_j \le (q_j+\gamma) v_j \le \left(1+\frac{\gamma}{\left(\frac{\log N}{N}\right)^{\frac{1}{3}}-\gamma}\right)q_jv_j,  
\end{equation}
together with \eqref{eqn:alpha_jvj_vs_ivi} and \eqref{eqn:alpha_ivi}, we obtain $\left(1+\frac{\gamma}{\left(\frac{\log N}{N}\right)^{\frac{1}{3}}-\gamma}\right)q_jv_j> (1-2\gamma/e(\alpha))q^*v^*$. Then by \cref{lem:alpha_quadratic},
\begin{equation}\label{eqn:alpha_case_middle_conclude}
\frac{1-2\gamma/e(\alpha)}{1+\frac{\gamma}{\left(\frac{\log N}{N}\right)^{\frac{1}{3}}-\gamma}} \le \frac{q_jv_j}{q^*v^*} \le 1-\frac{\alpha}{3}(q_j-q^*)^2 \le 1-\frac{\alpha}{3}(\frac{e(\alpha)}{2}-\gamma)^2.
\end{equation}
However, the left hand side of \eqref{eqn:alpha_case_middle_conclude} approaches to 1 while the right hand side is strictly less than 1, a contradiction. So this case never happens.

    \item Finally, if $q_j\ge \frac{e(\alpha)}{2}+\gamma$. We argue that if $\erm^c$ picks $q_j$ instead of $q_i$, then $R(q_j)$ approximates $R(q_i)$ well, satisfying the conclusion in the lemma. This is because
    \begin{align*}
        R(q_j) = q_jv_j \ge{}& (\frac{j}{N}-\gamma)v_j \\
        \ge{}& (1-\frac{2\gamma}{e(\alpha)})\frac{j}{N}v_j && \frac{j}{N}\ge q_j-\gamma\ge \frac{e(\alpha)}{2}\\
        \ge{}& (1-\frac{2\gamma}{e(\alpha)})\frac{i}{N}v_i && \text{\cref{eqn:alpha_jvj_vs_ivi}}\\
        \ge{}& (1-\frac{2\gamma}{e(\alpha)})(1-\frac{2\gamma}{e(\alpha)})q^*v^* && \text{\cref{eqn:alpha_ivi}} \\
        ={}& (1-O(\gamma))R(q^*).
    \end{align*}
\end{enumerate}

Combining above three steps and the event in the beginning of the proof, we have $R(q_{k^*}) \ge (1-O(\sqrt{\frac{\log N}{N}}))R(q^*)$ except with probability at most 
\[\frac{1}{(N-m)^{\frac{3-2\alpha}{2\alpha}}} + o(\frac{1}{N^{\frac{3-2\alpha}{2\alpha}}}) = O(\frac{1}{N^{\frac{3-2\alpha}{2\alpha}}}). \]
\end{proof}

\begin{lemma}\label{lem:alpha_event}
Let $F$ be an $\alpha$-strongly regular distribution. Suppose $m=o(\sqrt{N})$. Fix $m$ values $v_I$ to be $+\infty$, and randomly draw $N-m$ values $v_{-I}$ from $F$. Let $k^*=\argmax_{i> cN}\{iv_{i}\}$, i.e., the index selected by $\erm^c$, where $\left(\frac{\log N
}{N}\right)^{\frac{1}{3}}\le c\le \frac{\alpha^{1/(1-\alpha)}}{2}$. Let $\epsilon=\sqrt[4]{\frac{\log N}{N}}$. Then with probability at least $1-O\left(\frac{1}{N^{\frac{3-2\alpha}{2\alpha}}}\right)$, the following three inequalities hold:
\begin{enumerate}
    \item $q_{k^*}\ge q^*-O(\epsilon)$;
    \item $k^* \ge [q^*-O(\epsilon)] N > \frac{\alpha^{1/(1-\alpha)}}{2}N$; 
    \item $v_{k^*}\le [1+O(\epsilon)]v^*$. 
\end{enumerate}
The constants in the big $O$'s depend on $\alpha$. 
\end{lemma}
\begin{proof}
Define $e(\alpha)=\alpha^{1/(1-\alpha)}$. We have $q^*\ge e(\alpha)$. 

For inequality (1), by \cref{lem:alpha_quadratic} and \cref{lem:alpha_revenue_modified}, with probability at least $1-O(\frac{1}{N^{\frac{3-2\alpha}{2\alpha}}})$, we have
\[ \frac{\alpha}{3}(q_{k^*} - q^*)^2 \le \frac{R(q^*) - R(q_{k^*})}{R(q^*)} \le O\left(\sqrt{\frac{\log N}{N}}\right). \]
Taking the square root, we obtain $q_{k^*}\ge q^* - O\left(\sqrt[4]{\frac{\log N}{N}}\right)$. 

To prove (2), note that by \cref{claim:alpha_concentration}, we have $\frac{k^*}{N}\ge q_{k^*} - O\left(\sqrt{\frac{\log N}{N}}\right) \ge q^*-O(\epsilon)$. 

Finally, (3) follows from
\[\frac{v_{k^*}}{v^*}=\frac{R(q_{k^*})}{q_{k^*}}\frac{q^*}{R(q^*)}\le 1\cdot \frac{q^*}{q_{k^*}}\le \frac{q^*}{q^*-O(\epsilon)} = 1+\frac{O(\epsilon)}{q^*-O(\epsilon)} \le 1+O\left(\frac{\epsilon}{e(\alpha)}\right).  \]
\end{proof}

\subsection{Detailed Proof of Theorem \ref{thm:connection_TPM_PRM_combined} for $\alpha$-Strongly Regular Distributions}
\label{sec:proof_lemma_1_alpha}

Let $\Delta U(\bm{v}_i, b_I, v_{-I}) = U_i^{\TP}(\bm{v}_i, b_I, v_{-I}) - U_i^{\TP}(\bm{v}_i, v_I, v_{-I})$. 
Similar to the proof for bounded distributions, we have for any $\bm{v}_i, b_I, v_{-I}$, 
\begin{align*}
\Delta U(\bm{v}_i, b_I, v_{-I}) 
\le{} m_2\cdot \bigg(\erm^c(v_I, v_{-I}) - \erm^c(b_I, v_{-I}) \bigg) \le{} m_2 \cdot \erm^c(v_I, v_{-I}) \cdot \delta^{\interim}_{m_1}(v_{-I}).
\end{align*}

By \cref{claim:interim_to_delta}, we have $\erm^c(v_I, v_{-I})\le \erm^c(\overline{v}_I, v_{-I})$ where $\overline{v}_I$ can be any $m_1$ values (e.g., $+\infty$) that are greater than the maximal value in $v_{-I}$, when $c\ge \frac{m_1}{T_1K_1}$. 

Let $N=T_1K_1$, define two threshold prices $T_1=N^{\frac{3(1-\alpha)}{2\alpha}}v^*$ and $T_2=[1+O(\epsilon)]v^*$ where $\epsilon=\sqrt[4]{\frac{\log N}{N}}$ as in \cref{lem:alpha_event}. 

Note that for sufficiently large $N$, $T_1>T_2$. With the random draw of $v_{-I}$ from $F$, denote the random variable $\erm^c(\overline{v}_I, v_{-I})$ by $P$, we have: 

\begin{align}
\E_{\bm{v}_{-i}}\left[ \Delta U(\bm{v}_i, b_I, v_{-I})\right] 
= {} & \E_{v_{-I}}\left[ \Delta U(\bm{v}_i, b_I, v_{-I}) \mid P\le T_2\right]\cdot \Pr[P\le T_2] \nonumber\\
+ {} & \E_{v_{-I}}\left[ \Delta U(\bm{v}_i, b_I, v_{-I}) \mid T_2<P\le T_1\right]\cdot \Pr[T_2<P\le T_1] \nonumber \\
+ {} &\E_{v_{-I}}\left[ \Delta U(\bm{v}_i, b_I, v_{-I}) \mid P>T_1\right]\cdot \Pr[P>T_1] \nonumber \\
\defeq{}& \E_1 + \E_2 + \E_3. 
\end{align}

\begin{enumerate}
\item For the first term $\E_1$, 
\begin{align*}
\E_1 ={} &\E_{v_{-I}}\left[ \Delta U(\bm{v}_i, b_I, v_{-I}) \mid P\le T_2\right]\cdot \Pr[P\le T_2] \\
\le{}& \E_{v_{-I}}\left[ m_2\cdot P \cdot \delta^{\interim}_{m_1}(v_{-I}) \mid P\le T_2\right]\cdot \Pr[P\le T_2] \\
\le{} & m_2 \cdot T_2\cdot \E_{v_{-I}}\left[\delta^{\interim}_{m_1}(v_{-I})  \mid P\le T_2\right] \cdot \Pr[P\le T_2] \\
\le{} & m_2 \cdot [1+O(\epsilon)]v^*\cdot \E_{v_{-I}}\left[ \delta^{\interim}_{m_1}(v_{-I})\right] \\
={} & O\left(m_2 \cdot v^*\cdot \Delta^{\interim}_{N, m_1}  \right).
\end{align*}

\item For the second term, we claim that $\E_2=O(\frac{m_2v^*}{\sqrt{N}})$.

By \cref{lem:alpha_event}, we have $\Pr[P>[1+O(\epsilon)]v^*]\le O(\frac{1}{N^{\frac{3-2\alpha}{2\alpha}}})$. 
Therefore, 
\begin{align*}
\E_2 ={} &\E_{v_{-I}}\left[ \Delta U(\bm{v}_i, b_I, v_{-I}) \mid T_2< P\le T_1\right]\cdot \Pr[T_2< P\le T_1] \\
\le{} & \E_{v_{-I}}\left[ m_2\cdot P \cdot 1 \mid T_2< P\le T_1\right]\cdot \Pr[T_2< P\le T_1] \\
\le{} & m_2 \cdot T_1\cdot \Pr[P> T_2] \\
\le{} & m_2 \cdot N^{\frac{3(1-\alpha)}{2\alpha}}v^* \cdot O \left(\frac{1}{N^{\frac{3-2\alpha}{2\alpha}}}\right) \\
={} & O\left(\frac{m_2v^*}{\sqrt{N}}\right).\\
\end{align*}

\item For the third term, we claim that $\E_3 = o(\frac{m_2v^*}{\sqrt{N}})$.

Let $B$ be the upper bound on the support of $F$ ($B$ can be $+\infty$). Let $F_{P}(x)$ be the distribution of $P$. For convenience, suppose it is continuous and has density $f_P(x)$. We have: 
\begin{align*}
\E_3 ={}& \E_{v_{-I}}\left[ \Delta U(\bm{v}_i, b_I, v_{-I}) \mid P>T_1\right]\cdot \Pr[P>T_1] \\
\le{}& \E_{v_{-I}}\left[m_2\cdot P\cdot 1\mid P>T_1 \right]\cdot\Pr[P>T_1]\\
= {}& m_2\cdot \E_{v_{-I}}\left[ P \mid P>T_1 \right]\cdot\Pr[P>T_1] \\
={}& m_2\cdot \int_{T_1}^B xf_{P}(x)\dd x \\
={}& m_2\cdot\left( \int_{T_1}^B [1-F_{P}(x)]\dd x + T_1[1-F_{P}(T_1)]\right).
\end{align*}

Let $\max\{v_{-I}\}$ denote the maximum value in the $N-m_1$ samples $v_{-I}$. By \cref{lem:alpha_largest_value}, we have for any $x\ge v^*$, 
\[1-F_{P}(x)= \Pr[P>x] \le \Pr[\max\{v_{-I}\}>x] \le N\left( \frac{v^*}{(1-\alpha)x+\alpha v^*} \right)^{\frac{1}{1-\alpha}}.\]
Thus, 
\begin{align*}
& \int_{T_1}^B [1-F_P(x)]\dd x + T_1[1-F_{P}(T_1)] \\
\le{}& \int_{T_1}^B N\left( \frac{v^*}{(1-\alpha)x+\alpha v^*} \right)^{\frac{1}{1-\alpha}}\dd x + T_1 N \left( \frac{v^*}{(1-\alpha)T_1+\alpha v^*} \right)^{\frac{1}{1-\alpha}} \\
\le {}& \frac{N}{\alpha} \cdot \frac{(v^*)^{\frac{1}{1-\alpha}}}{[(1-\alpha)T_1+\alpha v^*]^{\frac{\alpha}{1-\alpha}}} + T_1 N \left( \frac{v^*}{(1-\alpha)T_1+\alpha v^*} \right)^{\frac{1}{1-\alpha}}\\
={}& \frac{N}{\alpha} \cdot (T_1+\alpha v^*) \left( \frac{v^*}{(1-\alpha)T_1+\alpha v^*} \right)^{\frac{1}{1-\alpha}} \\
={}& O\left(\frac{v^*}{\sqrt{N}}\right),
\end{align*}
as desired.
\end{enumerate}

Combining the three items, 
\begin{align*} \E_{v_{-I}}\left[\Delta U(v_I, b_I, v_{-I})\right] = O\left(m_2v^* \Delta^{\interim}_{N, m_1}  \right) + O\left(\frac{m_2v^*}{\sqrt{N}}\right), \end{align*}
where the constants in $O$'s depend on $\alpha$.

\iffalse
\section{Get Rid of $c$}
\citet{guo2019settling} proposed a ``shaded'' version of $\erm$, which circumvents the need to choose constant $c$ for $\erm^c$. Define a function: 
\begin{equation}
    s_{N, \delta}(q) =\max\left\{0, q - \sqrt{\frac{2q(1-q)\ln(2N\delta^{-1})}{N}} - \frac{4\ln(2N\delta^{-1})}{N} \right\}.
\end{equation}
Denote their algorithm by $\erm^s$. It replaces the empirical quantile $\frac{i}{N}$ for the $i$-th largest value $v_i$ by $s_{M, \delta}(\frac{i}{N})$. Then instead of maximizing over $\{\frac{i}{N}v_i\}$, it computes: 
\begin{equation}
    k^*=\argmax_{1\le i\le N}\{s_{M, \delta}(\frac{i}{N}) v_i\}, ~~\erm^s(v_1, \ldots, v_n) \defeq v_{k^*}.
\end{equation}

Again, consider $\erm^s(\overline{v}_I, v_{-I})$ where $v_I$ is $m$ copies of $+\infty$, and $v_{-I}$ are $N-m$ i.i.d.~random draws from $F$. 
\old{ Can you show the following?}

\begin{lemma}
Choose some $\delta$ independently of distribution families. 
\begin{itemize} 
\item For bounded distributions, there exists some constant $d$ (e.g. $d=\frac{1}{2D}$) such that $\Pr[k^*<dN]$ is very small (e.g., less than $O(\frac{\delta}{N}$)?
\item For MHR distributions, there exists some constant $d$ (e.g., $d=\frac{1}{9e}$), such that $\Pr[k^*<dN]$ is very small (e.g., less thant $O(\frac{\delta}{N}$)? (See \citet{guo2019settling} Page 45, Appendix D.2)
\end{itemize}
\end{lemma}
\fi

\section{Lower Bounds}\label{sec:appendix-lower-bound}
%\subsection{Lower Bounds on Incentive-Awareness Measure}\label{sec:lower-bound-delta}
\subsection{Discussion}
\noindent
{\bf A lower bound on $\Delta^{\interim}_{N, m}$. }
Theorem~\ref{thm:discount_combined} gives an upper bound on $\Delta^{\interim}_{N, m}$ for a specific range of $c$'s. When one considers respective lower bounds, a preliminary question would be: how does the choice of $c$ affect the possible lower bound? The following result shows that it is enough to prove a lower bound for one specific $c$ in the range of allowed $c$'s. The same lower bound will then hold for all $c$'s in that range.

\begin{proposition}\label{prop:choice_of_c}
Let $\Delta^{\interim}_{N, m}(c)$ denote the worst-case incentive-awareness measure of $\erm^c$. Suppose $m=o(\sqrt{N})$. 
\begin{itemize}
    \item For bounded distributions, $\Delta^{\interim}_{N, m}(c) = \Delta^{\interim}_{N, m}(\frac{m}{N})$, for any $c\in[\frac{m}{N}, \frac{1}{2D}]$. 
    \item For MHR distributions, 
    $\Delta^{\interim}_{N, m}(c)$ is bounded by $\Delta^{\interim}_{N, m}(\frac{m}{N}) \pm O(\frac{1}{N})$, for any $c\in[\frac{m}{N}, \cupperboundmhr]$. 
\end{itemize}
\end{proposition}

\begin{proof}\label{proof:prop-choice-of-c}
Fix any $d\in(0, 1)$ and any $c\in[\frac{m}{N}, \frac{d}{2}]$.
Recall that $\Delta_{N, m}^{\interim}(c) = \E_{v_{-I}\sim F}[ \delta_m^{\interim}(v_{-I},c)]$,
where, letting $\overline{v}_I$ be a vector of $m$ identical values that are equal to $\max v_{-I}$, then by Claim~\ref{claim:interim_to_delta}, 
$$\delta^{\interim}_m(v_{-I},c) = 1 - \frac{\inf_{b_I\in\mathbb{R}_+^m} P(b_I, v_{-I},c)}{P(\overline{v}_I, v_{-I},c)}.$$
Let $k^*(v,c)=\argmax_{i>cN}\{iv_{(i)}\}$ where $v=(\overline{v}_I, v_{-I})$, i.e.~the index of $P(v,c)$ in $v$. We show that, if $k^*(v,c) > dN$ for $c=\frac{m}{N}$, then 
%\in[m/N, d/2]$, then:

\begin{itemize}

\item $P(\overline{v}_I,v_{-I}, c')=P(\overline{v}_I,v_{-I}, \frac{m}{N})$ for any $c'\in [\frac{m}{N}, d]$. To see this, note that
\[k^*(v, \frac{m}{N})=\argmax_{i>\frac{m}{N}N}\{iv_{(i)}\} = \argmax_{i>dN}\{iv_{(i)}\} = \argmax_{i>c'N}\{iv_{(i)}\} = k^*(v, c'),  \]
where the second equality follow from our assumption that $k^*(v, \frac{m}{N}) > dN>m$ and the third equality is because $\frac{m}{N}\le c' \le d$. 

\item $\inf_{b_I\in\mathbb{R}_+^m} P(b_I, v_{-I}, c') = \inf_{b_I\in\mathbb{R}_+^m} P(b_I, v_{-I}, \frac{m}{N})$ for any $c'\in[\frac{m}{N}, \frac{d}{2}]$. Fix $c=\frac{m}{N}$ and consider any $c'\in[\frac{m}{N}, \frac{d}{2}]$. 
Let $b_I\in\mathbb{R}_+^m$ be any bids such that $P(b_I, v_{-I}, c) < P(\overline{v}_I, v_{-I}, c)$. Let $v^b=(b_I, v_{-I})$. Consider $k^*(v^b, c)$, we have $v^b_{(k^*(v^b, c))} < v_{(k^*(v, c))}$, so $k^*(v^b, c)$ must be greater than the index of $v_{(k^*(v,c))}$ in $v^b$. The index of $v_{(k^*(v,c))}$ in $v^b$ is at least $k^*(v,c)-m$, thus $k^*(v^b, c) > k^*(v,c)-m> dN-m \ge \frac{d}{2} N$. 
%It means there exists some bid $v^*\in (b_I,v_{-I})$ and $P(b_I,v_{-I},c)=v^*<v_{k^*(v,c)}$. Now consider the index of $v^*$ in $(b_I,v_{-I},c)$, it is at least $k^*(v,c)-m>dN-m>dN/2$, which implies that $P(b_I,v_{-I},c')=v^*$ for any $c'\in[m/N,d/2]$. 
We claim that $P(b_I, v_{-I}, c') = P(b_I, v_{-I}, c)$. To see this, note that 
\[k^*(b_I, v_{-I}, \frac{m}{N})=\argmax_{i>\frac{m}{N}N}\{iv^b_{(i)}\} = \argmax_{i>\frac{d}{2}N}\{iv^b_{(i)}\} = \argmax_{i>c'N}\{iv^b_{(i)}\} = k^*(b_I, v_{-I}, c'),  \]
where the second equality is because $k^*(v^b, \frac{m}{N})> \frac{d}{2}N$, and the third equality follows from $\frac{m}{N}\le c'\le \frac{d}{2}$. 
\end{itemize}
Thus, $\delta^{\interim}_m(v_{-I}, c) = \delta^{\interim}_m(v_{-I}, \frac{m}{N})$ for any $c\in[\frac{m}{N}, \frac{d}{2}]$. Define $\event(c)=[k^*\le dN]$, then 
\begin{align}\label{eqn:Delta_interim_any_c_bounded}
\Delta^{\interim}_{N, m}(c) ={}& \E[\delta^{\interim}_m(v_{-I}, c)] \nonumber \\
={}& \E[\delta^{\interim}_m(v_{-I}, c)\mid \overline{\event(c)}]\cdot\Pr[\overline{\event(c)}] +  \E[\delta^{\interim}_m(v_{-I}, c)\mid \event(c)]\cdot\Pr[\event(c)]  \nonumber\\
={}& \E[\delta^{\interim}_m(v_{-I}, \frac{m}{N})\mid \overline{\event(c)}]\cdot\Pr[\overline{\event(c)}] +  \E[\delta^{\interim}_m(v_{-I}, c)\mid \event(c)]\cdot\Pr[\event(c)]. \nonumber
\end{align}

For bounded distributions, consider $d=\frac{1}{D}$ and $c\in[\frac{m}{N},\frac{1}{2D}]$. We have proved in Theorem~\ref{thm:discount_combined} that $\Pr[\event(c)]=0$ for any $c\in[\frac{m}{N}, \frac{d}{2}]$. Thus $\Delta^{\interim}_{N, m}(c) = \E[\delta^{\interim}_m(v_{-I}, \frac{m}{N})\mid \overline{\event(c)}]\cdot\Pr[\overline{\event(c)}] = \Delta^{\interim}_{N, m}(\frac{m}{N})$.

For MHR distributions, let $d=\frac{1}{2e}$, as proved in  Lemma~\ref{lem:mhr_event}, $\Pr[\event]=O(\frac{1}{N})$ for $c\in[\frac{m}{N}, \frac{1}{4e}]$. Then 
$0\leq \E[\delta^{\interim}_m(v_{-I}, c)\mid \event(c)]\cdot \Pr[\event(c)] \leq 1\cdot \Pr[\event(c)] = O(\frac{1}{N})$. Thus  $\left| \Delta^{\interim}_{N, m}(c) - \Delta^{\interim}_{N, m}(\frac{m}{N})\right| \le O(\frac{1}{N})$ for any $c\in[\frac{m}{N}, \frac{1}{4e}]$.

\end{proof}

\citet{lavi2019redesigning} show a lower bound that can be compared to our upper bounds in Theorem~\ref{thm:discount_combined}. Specifically, they show that for the two-point distribution $v=1$ and $v=2$, each w.p.~0.5,
$\Delta^{\interim}_{N, 1}=\Omega(1/\sqrt{N})$. It is easy to adopt their analysis to any two-point distribution with $v_1=1$ and $v_2>1$. Since this is a bounded distribution, we obtain the following corollary: 

\begin{corollary}
For the class of bounded distribution with support in $[1, D]$, and any choice of $\frac{m}{N} \leq c \leq \frac{1}{2D}$, $\erm^c$ gives $\Delta^{\interim}_{N, m}(c) = \Omega(\frac{1}{\sqrt{N}})$ where the constant in $\Omega$ depends on $D$.
\end{corollary}

%\noindent
It remains open to prove other lower bounds on $\Delta^{\interim}_{N, m}$, especially for MHR distributions.

\noindent
{\bf A lower bound on the approximate truthfulness parameter, $\epsilon_1$. }
Since $\Delta^{\interim}_{N, m}$ only upper bounds the approximate truthfulness parameter $\epsilon_1$,
a lower bound on $\Delta^{\interim}_{N, m}$ does not immediately implies a lower bound on $\epsilon_1$. However, an argument similar to above shows the same lower bound directly on $\epsilon_1$. Consider the two-point distribution $F$ where for $X\sim F$, $\Pr[X=1]=1-\frac{1}{D}$ and $\Pr[X=D]=\frac{1}{D}$. For simplicity let $K_1=K_2=2$ and suppose bidder $i$ participates in $m_1$ and $m_2$ auctions in the two phases, respectively. Let $N=T_1K_1$ and assume $m_1=o(\sqrt{N})$. Suppose the first-phase mechanism $\mathcal{M}$ is the second price auction with no reserve price.
Then,
\begin{proposition}\label{prop:lower-bound-two-phase}
In the above setting, $\epsilon_1 = \Omega\left( \frac{m_2}{\sqrt{N}}\right)$ for any $c\in[\frac{m_1}{N}, \frac{1}{2D}]$,
where the constant in $\Omega$ depends on $D$.
\end{proposition}
\noindent The proof is in \cref{sec:lb-on-epsilon-1}. Once again, it remains open (and interesting, we believe) to prove a lower bound for MHR distributions, and to close the gap between our upper bound for bounded distributions which is $O(N^{-1/3} \log^2 N)$.

\subsection{Proof of \cref{prop:lower-bound-two-phase}: Lower Bound for the Two-Phase Model}\label{sec:lb-on-epsilon-1}
Consider the two-point distribution $F$ where for $X\sim F$, $\Pr[X=1]=1-\frac{1}{D}$ and $\Pr[X=D]=\frac{1}{D}$. For simplicity let $K_1=K_2=2$, and suppose bidder $i$ participates in $m_1$ and $m_2$ auctions in the two phases, respectively. Let $N=T_1K_1$ and assume $m_1=o(\sqrt{N})$. Suppose the first-phase mechanism $\mathcal{M}$ is the second price auction with no reserve price. We argue that to satisfy $\epsilon_1$-approximate truthfulness, $\epsilon_1$ must be $\Omega\left( \frac{m_2(D-1)^2}{\sqrt{(D-1)N}}\right)$, for $c\in[\frac{m}{N}, \frac{1}{2D}]$. 

Suppose the values of bidder $i$ across two phases are all $D$'s, i.e., 
\[\bm{v}_i=(\overbrace{D, \ldots, D}^{m_1}, \overbrace{D, \ldots, D}^{m_2}),\] 
and bidder $i$ bids $m_1$ $(D-\epsilon)$'s with $\epsilon=\frac{D^2}{N}$ in the first phase (assume $N\gg D^2$), 
\[b_I = (\overbrace{D-\epsilon, \ldots, D-\epsilon}^{m_1}). \]
Recall the definition of the interim utility of bidder $i$: 
\[\E_{v_{-I}}\left[ U^{\TP}_i(\bm{v}_i, b_I, v_{-I})\right] = \E_{v_{-I}}\left[ U^{\mathcal{M}}_i(\bm{v}_i, b_I, v_{-I})
    +
    m_2 \usecond^{K_2}(D, P(b_I, v_{-I})) \right].
\]
\begin{itemize} 
\item First consider the increase of interim utility in the second phase. If the reserve price is $D$, then bidder $i$'s utility is
\begin{equation}\label{eq:utility-phase-2-reserve-1}
 \usecond^{K_2}(D, D)=0.
\end{equation}
If the reserve price is $1$, then her utility becomes
\begin{equation}\label{eq:utility-phase-2-reserve-D}
 \usecond^{K_2}(D, 1) = (1-\frac{1}{D})\cdot (D-1) + \frac{1}{D}\cdot 0 = \frac{(D-1)^2}{D}.
\end{equation}
We then consider the probability that the reserve price is decreased from $D$ to $1$ because bidder $i$ deviates from $D$ to $D-\epsilon$. This probability is over the random draw of $N-m_1$ values $v_{-I}$. Suppose there are exactly $(\frac{1}{D}N+1-m_1)$ $D$'s in $v_{-I}$. Then when bidder $i$ bids truthfully, there are $(\frac{1}{D}N+1)$ $D$'s in $(v_I, v_{-I})$ in total, which results in $P(v_I, v_{-I}) = D$ because $(\frac{1}{D}N+1)\cdot D > N\cdot 1$. However, if bidder $i$ deviates to $b_I$, then $P(b_I, v_{-I})$ becomes $1$, because
\[ (\frac{1}{D}N+1)\cdot(D-\epsilon) =  N(1+\frac{D}{N})(1-\frac{D}{N}) < N \cdot 1,\quad\text{ and }\quad (\frac{1}{D}N+1-m_1)\cdot D \le N\cdot 1.\]
Thus, the reserve price is decreased from $D$ to $1$ with probability at least: 
\begin{equation}\label{eq:prob-decrease-D-1}
\Pr[Bin(N-m_1, \frac{1}{D})=\frac{1}{D}N+1-m_1] = \Omega\left(\frac{1}{\sqrt{\frac{1}{D}(1-\frac{1}{D})N}}\right). 
\end{equation}
Combining \eqref{eq:prob-decrease-D-1}, \eqref{eq:utility-phase-2-reserve-D} and $\eqref{eq:utility-phase-2-reserve-1}$, we obtain 
\begin{multline}\label{eq:utility-phase-2-difference}
 \E_{v_{-I}}\left[ m_2 \usecond^{K_2}(D, P(b_I, v_{-I}))  - m_2 \usecond^{K_2}(D, P(v_I, v_{-I}))\right] \\
 \ge  \Omega\left(\frac{1}{\sqrt{\frac{1}{D}(1-\frac{1}{D})N}}\right)m_2\left( \frac{(D-1)^2}{D} - 0 \right) = \Omega\left( \frac{m_2(D-1)^2}{\sqrt{(D-1)N}}\right). 
\end{multline}
\item Then we upper bound the utility loss due to non-truthful bidding in the first phase for bidder $i$. Note that since $\mathcal{M}$ is the second price auction with no reserve price, no matter bidder $i$ bids $D$ or $D-\epsilon>1$, her interim utility is the same: 
\[ \E_{v_{-I}}\left[ U^{\mathcal{M}}_i(\bm{v}_i, v_I, v_{-I}) \right] = m_1 \left((1-\frac{1}{D})(D-1) + \frac{1}{D}\cdot 0\right) = \frac{m_1(D-1)^2}{D}, \]
\[ \E_{v_{-I}}\left[ U^{\mathcal{M}}_i(\bm{v}_i, b_I, v_{-I}) \right]= m_1 \left((1-\frac{1}{D})(D-1) + \frac{1}{D}\cdot 0\right) = \frac{m_1(D-1)^2}{D}. \]
Thus
\begin{equation}\label{eq:utility-phase-1-difference}
\E_{v_{-I}}\left[ U^{\mathcal{M}}_i(\bm{v}_i, b_I, v_{-I})  - U^{\mathcal{M}}_i(\bm{v}_i, v_I, v_{-I}) \right] = 0.
\end{equation}
\end{itemize}
Finally, by \eqref{eq:utility-phase-2-difference} and \eqref{eq:utility-phase-1-difference}, we have
\[\E_{v_{-I}}\left[ U^{\TP}_i(\bm{v}_i, b_I, v_{-I}) - U^{\TP}_i(\bm{v}_i, v_I, v_{-I})\right] \ge \Omega\left( \frac{m_2(D-1)^2}{\sqrt{(D-1)N}}\right), \] 
which gives a lower bound on $\epsilon_1$.

\section{Unbounded Regular Distributions}
\label{sec:appendix-unbounded-regular-distributions}

\subsection{Discussion}
\cref{thm:TPM_utility_revenue_combined} shows that approximate truthfulness and revenue optimality can be obtained simultaneously for bounded (regular) distributions and for MHR distributions.
%
%\new{We have shown that approximate truthfulness and revenue optimality can be obtained simultaneously for bounded (regular) distributions and for MHR distributions in \cref{thm:TPM_utility_revenue_combined}.
%
A natural question would then be: what is the largest class of value distribution that we can consider? 
If $K_1 > 1$ or $K_2 > 1$ (i.e., each auction includes multiple bidders, at least $2$), then running a second price auction with an anonymous reserve price may not be optimal if the distribution is non-regular \cite{myerson1981optimal}. Moreover, even in the one-bidder case, the sample complexity literature analyzes the ERM algorithm only for regular or for non-regular and bounded distributions. For other classes of distributions, ERM does not seem to be the correct choice. Thus, the class of general unbounded regular distributions is the largest class we can consider. Still, our results do not cover this entire class since MHR distributions is a strict sub-class of regular distributions and for regular but non-MHR distributions we assume boundedness.

Our results can be generalized to the class of $\alpha$-strongly regular distributions with $\alpha > 0$. As defined in \cite{cole2014the}, a distribution $F$ with positive density function $f$ on its support $[A, B]$ where $0\le A\le B\le+\infty$ is \emph{$\alpha$-strongly regular} if the virtual value function $\phi(x) = x-\frac{1-F(x)}{f(x)}$ satisfies $\phi(y)-\phi(x)\ge \alpha (y-x)$
whenever $y>x$ (or $\phi'(x) \ge \alpha$ if $\phi(x)$ is differentiable). As special cases, regular and MHR distributions are $0$-strongly and $1$-strongly regular distributions, respectively.
For $\alpha>0$, we obtain bounds similar to MHR distributions on $\Delta^{\interim}_{N, m}$ and approximate incentive-compatibility in the two-phase model and the uniform-price auction.  Specifically, \cref{thm:discount_combined} can be extended to any $\alpha$-strongly regular distribution with $\alpha>0$ as follows: 
\begin{theorem}\label{thm:delta-upper-bound-alpha}
If $F$ is $\alpha$-strongly regular for $0< \alpha \le 1$, then $\Delta^{\interim}_{N,m}=O\left(m\frac{\log^3 N}{\sqrt{N}}\right)$, when $m=o(\sqrt{N})$ and $\frac{m}{N} \le \left(\frac{\log N}{N}\right)^{1/3}\le c \le \frac{\alpha^{1/(1-\alpha)}}{4}$. 
The constants in $O$ and $o$ depend on $\alpha$. 
\end{theorem}
%in Section~\ref{sec:appendix-alpha-strong}. 
Note that this bound holds only for large enough $N$'s since $\alpha^{1/(1-\alpha)} \rightarrow 0$ as $\alpha \rightarrow 0$. However, for any fixed $\alpha > 0$ there exists a large enough $N$ such that the relevant range for appropriate $c$'s will be non-empty.\footnote{The constant in the big $O$ depends on $\alpha$, so it is constant only if the distribution function is fixed. However, it goes to infinity as $\alpha\to 0$.} The proof of the upper bound on $\Delta^{\interim}_{N, m}$ is similar to the proof for MHR distributions (Lemma \ref{lem:mhr_main_upper-bound}), except that, we need $c\ge \left(\frac{\log N}{N}\right)^{1/3}$ to guarantee $\Pr[\event]=O(\frac{1}{N})$ (\cref{lem:alpha_event}); thus we omit the proof.

Similarly, both of \cref{thm:connection_TPM_PRM_combined} which says that the two-phase model is $
(O(m_2v^* \Delta^\worst_{T_1K_1, m_1}) + O(\frac{m_2v^*}{\sqrt{T_1K_1}}))$-BIC and \cref{thm:uniform-price-bne} which says that the uniform-price auction is $
(m, (O(v^* \Delta^\worst_{N, m}) + O(\frac{v^*}{\sqrt{N}})))$-group BIC, hold for any $\alpha$-strongly regular distribution with $\alpha>0$.  The proof of the former is in \cref{sec:proof_lemma_1_alpha} and the latter is omitted.

\begin{figure}[tb]
    \centering
    \begin{minipage}{.45\textwidth}
        \centering
        \includegraphics[width=\linewidth]{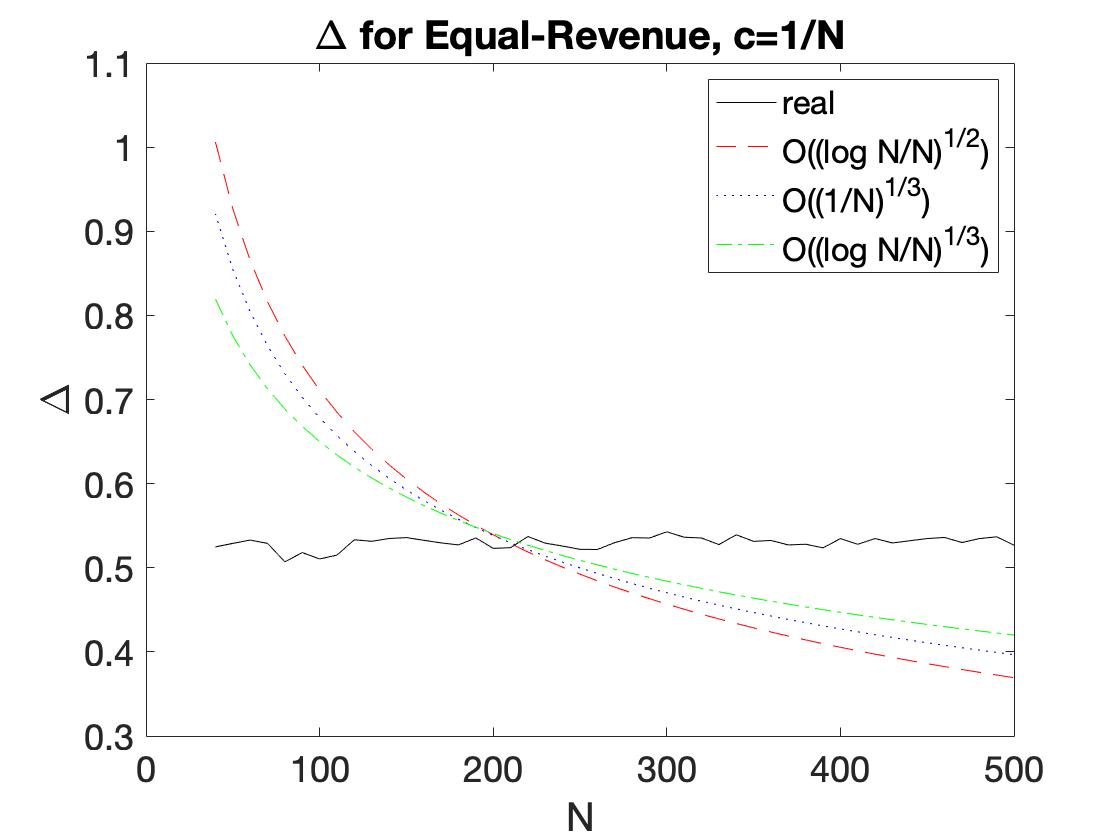}
        \caption{For equal-revenue distribution, $\Delta^{\interim}_{N, 1}$ (in black) on the $y$-axis as a function of $N$ on the $x$-axis, with $c=1/N$. Three other functions are plotted for reference.}
        \label{fig:ER-constant}
    \end{minipage}%
    \hfill 
    \begin{minipage}{0.45\textwidth}
        \centering
        \includegraphics[width=\linewidth]{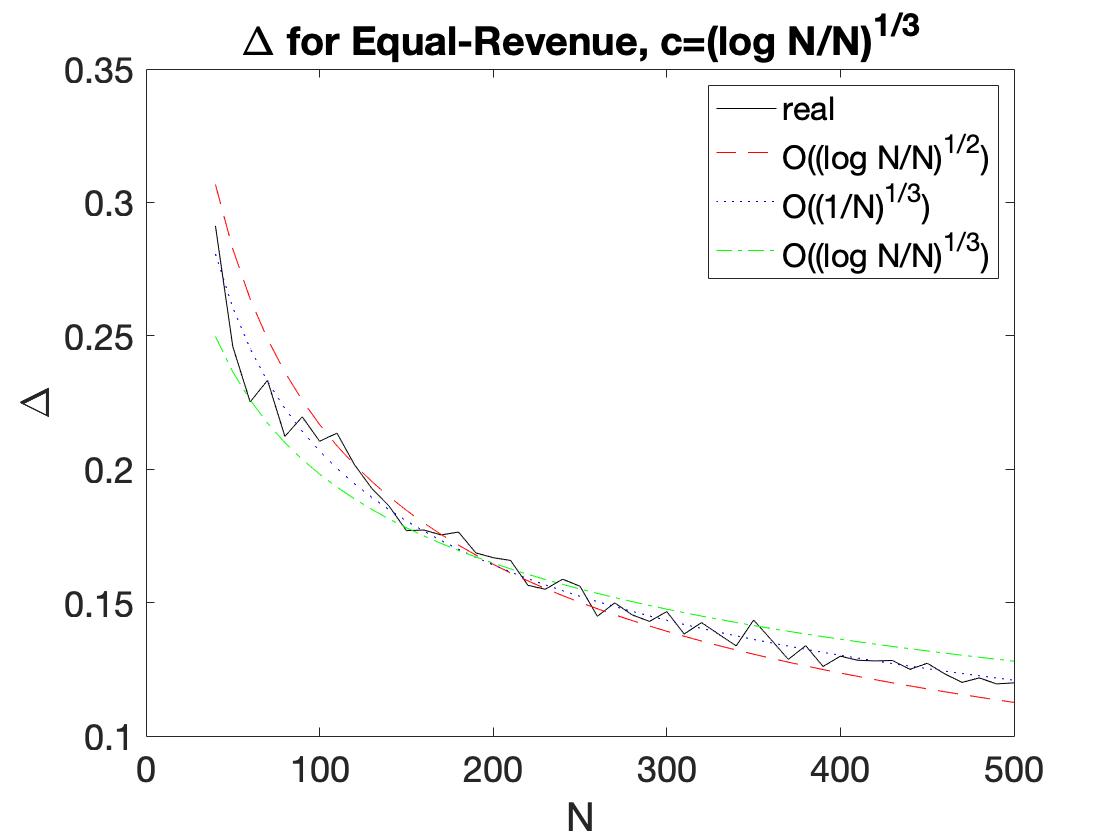}
        \caption{For equal-revenue distribution, $\Delta^{\interim}_{N, 1}$ (in black) on the $y$-axis as a function of $N$ on the $x$-axis, with $c=\Theta((\log N/N)^{1/3})$. Three other functions are plotted for reference.}
        \label{fig:ER-non-constant}
    \end{minipage}
\end{figure}

%\subsubsection{Arbitrary unbounded regular distributions}
It remains an open problem for future research whether $\erm^c$ is incentive-aware in the large for regular distributions that are not $\alpha$-strongly regular for any $\alpha>0$. For these distributions additional technical challenges exist since the choice of $c$ in $\erm^c$ creates a clash between approximate truthfulness and approximate revenue optimality. Unlike MHR and bounded regular distributions for which we can fix $c=m/N$ to obtain approximate truthfulness and revenue optimality, for arbitrary unbounded distribution we have to choose $c$ more carefully. If $c$ is too large, for example, a positive constant, then we cannot obtain nearly optimal revenue.

Specifically, to obtain close-to-optimal revenue for all bounded distributions on $[1, D]$ it is easy to verify that we need $c \leq 1/D$. Since the class of unbounded regular distributions contains all bounded regular distributions for all $D \in \mathbb R_+$, it follows that $c$ cannot be a constant. We therefore need to consider a non-constant $c(N)$. In fact, it has been shown in \cite{huang2018making} that if $c(N)\to 0$ as $N \rightarrow \infty$ then approximate revenue optimality can be satisfied. However, if $c$ is too small, truthfulness will be violated, as discussed in the following two examples (assume $m=1$ for simplicity).
%so the number of samples collected from the first phase equals $n$.

\begin{example}[Small $c$ hurts truthfulness]\label{app-example:ER-constant-delta}
Suppose we choose $c(N)=\frac{1}{N}$, that is, $\erm^c$ ignores only the largest sample. Consider the equal-revenue distribution $F(v)=1-\frac{1}{v}$ for $v\in[1, +\infty)$. Note that this is a $0$-strongly regular distribution but not $\alpha$-strongly regular for any $\alpha>0$ since for any $x$, $\phi(x) = 0$. Similarly to \citet{yao2018incentive}, we prove in Appendix~\ref{sec:ER-proof} that $\Delta^{\interim}_{N, 1}$ does not go to $0$ as $N\to+\infty$. This is also visible in Fig.~\ref{fig:ER-constant}, which shows in black $\Delta^{\interim}_{N, 1}$ (on the $y$-axis) as a function of $N$ (on the $x$-axis). This was obtained via simulation, for $c=1/N$. Three other functions are plotted in other colors, for reference. However, whether $\Delta^{\interim}_{N, 1}\to 0$ crucially depends on the choice of $c$, as can be seen in Fig.~\ref{fig:ER-non-constant}, where $\Delta^{\interim}_{N, 1}$ seems to converge to zero with $c=\Theta((\log N / N)^{1/3})$.
%Unfortunately, this choice of a larger $c$ does  guarantee approximately optimal revenue for all regular distributions discussed in the next example.
\end{example}

\begin{example}[Does an intermediate $c$ hurt truthfulness as well?]
Now assume $c=\Theta((\log N/N)^{1/3})$, and consider the ``triangular'' distribution $F(v)=1-\frac{1}{v+1}$ for $v\in [0, +\infty)$.
This distribution can be seen as the limit of a series of bounded regular distributions whose upper bounds and optimal reserve prices both tend to $+\infty$. Note that it is a regular distribution but not $\alpha$-strongly regular for any $\alpha>0$. We do not know whether
$\Delta^{\interim}_{N, 1}\to 0$ as $N \rightarrow \infty$ (see Fig.~\ref{fig:delta-triangular}). In particular, our main lemma (Lemma~\ref{lem:main_upper-bound}) may not suffice to analyze this distribution as $\Pr[\event]$ (as defined in that lemma) is unlikely to go to zero as $N$ goes to infinity (see Fig.~\ref{fig:triangular_Pr_E} for simulation results). 
%whose revenue curve is $R(q)=1-q$ (note that $v(q)=\frac{1}{q}-1$). The optimal reserve price is $+\infty$, with quantile 0, yielding optimal revenue $R(0)=1$. 
%Using $\erm^c$ with any constant $c>0$, by \cref{claim:concentration}, with high probability, each sample $v_{i}$, $i> cN$ has quantile $q_{i}>c-\gamma$ where $\gamma=O(\sqrt{\frac{\log N}{N}})$. Thus $\erm^c$ will output a reserve price with quantile $q^c>c-\gamma$, whose revenue is $R(q^c)<1 - c + \gamma$, less than $R(0)$ by a constant.  
\end{example}

\begin{figure}
\centering
\begin{minipage}{.45\textwidth}
  \centering
  \includegraphics[width=\linewidth]{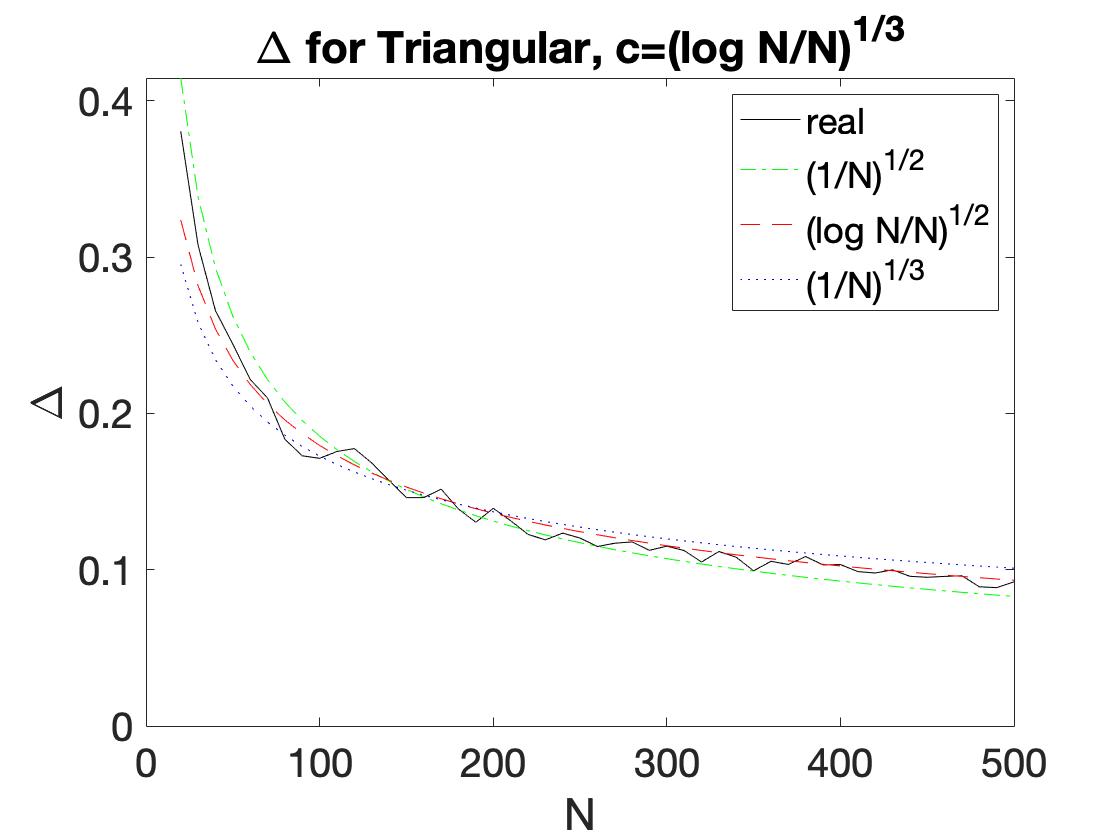}
  \caption{$\Delta^{\interim}_{N, 1}$ for the triangular distribution, with $c=\Theta((\log N/N)^{1/3})$}
  \label{fig:delta-triangular}
\end{minipage}%
\hfill
\begin{minipage}{.45\textwidth}
  \centering
  \includegraphics[width=\linewidth]{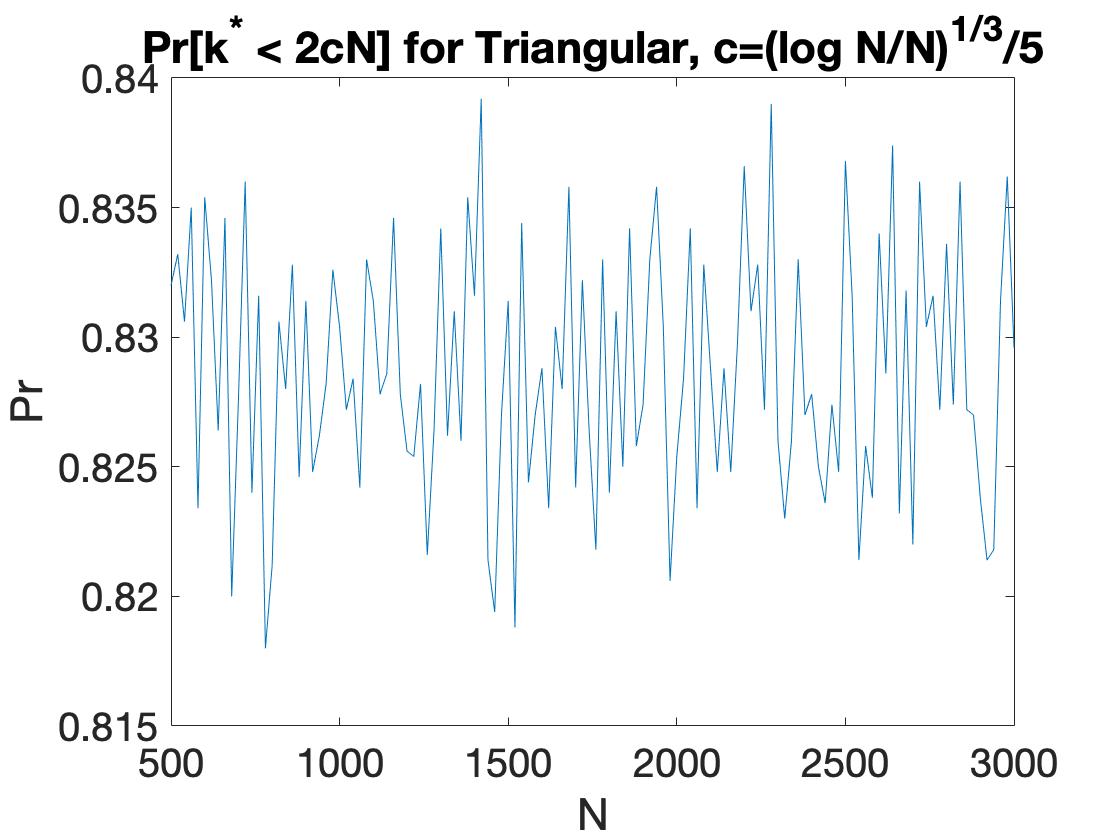}
  \caption{$\Pr[k^* < 2cN]$ as a function of $N$ for the triangular distribution, with $c=(\log N/N)^{1/3}/5$. Our simulations show that $\Pr[\event]\ge\Pr[k^*<2cN]\ge 0.815$.}
  \label{fig:triangular_Pr_E}
\end{minipage}
\end{figure}

\subsection{Proof of Example \ref{app-example:ER-constant-delta}: $\Delta^{\interim}_{N, 1}\not\to 0$ for the Equal-Revenue Distribution and $c=1/N$}
\label{sec:ER-proof}

We show that when $F$ is the equal-revenue distribution, $F(v)=1-\frac{1}{v}$, ($v\in[1, +\infty)$), and $c=1/N$, $\Delta^{\interim}_{N, 1}$ does not go to $0$ as $N\to+\infty$.

Recall Definition \ref{def:discount}, $\Delta^{\interim}_{N, 1}=\E_{v_{-I}}[\delta^{\interim}_{N, 1}(v_{-I})]$: we draw $N-1$ i.i.d. values $v_{-I}=\{v_2,\dots,v_N\}$ from $F$, and a bidder can change any value $v_{I}=v_1$ to any non-negative bid $b_1$. Let the other values $v_{-I}=\{v_2,\dots,v_N\}$ be sorted as $v_{(2)}\geq \dots \geq v_{(N)}$.
Let $\lambda=10, \lambda'=1$, and let $T_{N-1}$ be the event that $v_{(2)} > \lambda (N-1)$ and $v_{(3)} < \lambda' (N-1)$.
Let $V_{N-1}$ be the event that $\max_{3\leq i \leq N}\{i\cdot v_{(i)}\}\leq \lambda (N-1)$.

When $v_{-I}$ satisfies event $T_{N-1}\land V_{N-1}$, then
\begin{equation}\label{eq:delta_lambda}
\delta_{m}^{\interim}(v_{-I}) \geq 1-\frac{\lambda'}{\lambda}.
\end{equation}
This is because when $T_{N-1}\land V_{N-1}$ happens, there exists some $v_1 > \lambda (N-1)$, resulting $P(v_{I},v_{-I}) > \lambda (N-1)$, while the bidder can strategically bid $b_1 < \lambda' (N-1)$ and change the price to $P(b_{I},v_{-I}) < \lambda' (N-1)$.

Moreover, we will show that the probability that the event $T_{N-1}\land V_{N-1}$ happens satisfies 
\begin{equation}\label{eqn:tnvn}
\Pr[T_{N-1}\land V_{N-1}] \geq 0.9\cdot \frac{1}{\lambda}e^{-\frac{2}{\lambda'}}.
\end{equation}
%This is exactly the same event defined in the Lemma 5 of \citet{yao2018incentive} except replacing the number of values from $n$ to $N-1$.
Combining \eqref{eq:delta_lambda}  and \eqref{eqn:tnvn}, we know
$$\Delta_{m}^{\interim} \geq (1-\frac{\lambda'}{\lambda})\cdot(0.9\cdot \frac{1}{\lambda}e^{-\frac{2}{\lambda'}}) > 0.$$
%The rest proof is exactly as in \citet{yao2018incentive}, and is given below only for completeness.

The proof of \eqref{eqn:tnvn} is separated into two parts.
Firstly, 
$$\Pr[T_{N-1}] = \binom{N-1}{1}\frac{1}{\lambda (N-1)}(1-\frac{1}{\lambda'(N-1)})^{N-2}\geq \frac{1}{\lambda}e^{-\frac{2}{\lambda'}}.$$
Then we show that
\[ \Pr[\overline{V_{N-1}} \mid T_{N-1}]<0.1.\]
Let $z_3, \dots, z_{N}$ be i.i.d.~random draws according to $F$ conditioning on $z<\lambda'(N-1)$, i.e., for any $t\in [1,N-1]$, recalling that $\lambda'=1$, 
$$\Pr_{z\sim F}[ z>t \mid z<\lambda'(N-1)]=\frac{1}{1-\frac{1}{N-1}}(\frac{1}{t}-\frac{1}{N-1}).$$
Let $Y_{N-1}^{max} = \max_{3\leq i \leq N} \{i\cdot z_{(i)}\}$ where $z_{(3)} \ge \cdots \ge z_{(N)}$ is the sorted list of $z_i$'s. 
Clearly,
$$\Pr[\overline{V_{N-1}} \mid T_{N-1}] =  \Pr[Y_{N-1}^{max} \geq \lambda (N-1)].$$
For any $t\geq 1$, let $M_t$ be the number of $z_i$'s ($3\leq i\leq N$) satisfying $z_i\geq t$, and $B_t$ be the event that $t\cdot (M_t+2)\geq \frac{\lambda(N-1)}{2}$.
Let $t_k=\frac{N-1}{2^k}$ for $1\leq k \leq \lfloor \log_2 (N-1)\rfloor$.
As the event $Y_{N-1}^{max}\geq \lambda(N-1)$ implies $\bigvee_{1\leq k\leq\lfloor \log_2 (N-1)\rfloor}B_{t_k}$, we have
\begin{equation}\label{eqn:yn}
\Pr[Y_{N-1}^{max} \geq \lambda (N-1)]\leq \sum\limits_{k=1}^{\lfloor \log_2 (N-1)\rfloor} \Pr[B_{t_k}].
\end{equation}
Note that $\E[M_t] = \frac{N-1}{t} - 1$.
Using Chernoff's bound, we have
\begin{equation}\label{eqn:bt}
\Pr[B_t] = \Pr[M_t\ge \frac{\lambda(N-1)}{2t}-2] \leq e^{-\frac{1}{3}(\frac{N-1}{t}-1)(\frac{\lambda}{2}-1)^2}.
\end{equation}
Combining \eqref{eqn:yn} and \eqref{eqn:bt}, we know 
\begin{align*}
\Pr[\overline{V_{N-1}} \mid T_{N-1}] = \Pr[Y_{N-1}^{max} \geq \lambda (N-1)]\leq \sum_{k=1}^{\lfloor \log_2 (N-1)\rfloor} e^{-\frac{1}{3}(2^k-1)(\frac{\lambda}{2}-1)^2}  < \sum_{k=1}^{+\infty} e^{-\frac{1}{3}k(\frac{\lambda}{2}-1)^2} < 0.1, 
\end{align*}
and this completes the proof of \eqref{eqn:tnvn}.

\end{document}